%% file: main.tex
\documentclass[11pt]{article} \usepackage[numbers]{natbib}

\usepackage{fullpage}
\usepackage{color}
\usepackage{booktabs}
\usepackage[linesnumbered,ruled,vlined]{algorithm2e}
\usepackage{algpseudocode}
\usepackage{amsfonts}
\usepackage{amssymb}
\usepackage{amsthm}
\usepackage{amsmath}
\usepackage{mathptmx}
\usepackage{soul}
\usepackage{graphicx} \usepackage{subcaption}
\usepackage{float}
\usepackage{placeins}
\usepackage{tcolorbox}

\usepackage{ifpdf}
\newcommand{\mydriver}{hypertex}
\ifpdf
\renewcommand{\mydriver}{pdftex}
\fi
\usepackage[breaklinks,\mydriver]{hyperref}
\usepackage{svg}

\usepackage{url}

\usepackage[capitalise]{cleveref}
\usepackage{footnote}
\usepackage{lipsum}
\usepackage{thmtools}
\usepackage{thm-restate}

\makesavenoteenv{tabular}
\makesavenoteenv{algorithm}

\newcommand{\nis}{\mathrm{nis}}

\newcommand{\Var}{\mathrm{Var}}

\newcommand{\eps}{\varepsilon}

\newcommand{\E}{\mathop{\mathbf{E}}}

\newcommand{\poly}{\mathrm{poly}}

\newcommand{\cost}{\mathrm{cost}}
\newcommand{\vol}{\mathrm{vol}}

\newcommand{\MST}{\mathrm{MST}}

\definecolor{mulberry}{rgb}{0.77, 0.29, 0.55}

\newtheorem{theorem}{Theorem}[section]
\newtheorem{fact}[theorem]{Fact}
\newtheorem{lemma}[theorem]{Lemma}

\newtheorem{corollary}[theorem]{Corollary}

\theoremstyle{definition}
\newtheorem{definition}[theorem]{Definition}
\newtheorem{example}[theorem]{Example}
\newtheorem{proposition}[theorem]{Proposition}

\providecommand{\abs}[1]{\ensuremath{\left\lvert#1\right\rvert}}

\definecolor{offwhite}{rgb}{0.98, 0.98, 0.98}

\allowdisplaybreaks

\title{\Large Sublinear Algorithms for Estimating Single-Linkage Clustering Costs}
\author{Pan Peng\thanks{School of Computer Science and Technology, University of Science and Technology of China. Email: \href{mailto:ppeng@ustc.edu.cn}{ppeng@ustc.edu.cn}. Supported in part by NSFC grant 62272431.} \and Christian Sohler\thanks{Department of Mathematics and Computer Science, University of Cologne. Email: \href{mailto:csohler@uni-koeln.de}{csohler@uni-koeln.de}. Funded by the Deutsche Forschungsgemeinschaft (DFG, German Research Foundation) – Project Number 459420781.}\and Yi Xu\thanks{School of Computer Science and Technology, University of Science and Technology of China. Email: \href{mailto:yi\_xu@mail.ustc.edu.cn}{\mbox{yi\_xu@mail.ustc.edu.cn}}. Supported in part by NSFC grant 62272431.}}

\date{}

\begin{document}
\maketitle

\begin{abstract}
Single-linkage clustering (SLC) is a fundamental method for hierarchical data
analysis. In the distance setting, a $k$-clustering produced by SLC can be
obtained by computing a minimum spanning tree (MST) and deleting its $k-1$
heaviest edges. This naturally induces a cost profile for the SLC hierarchy:
for each $k\in[n]$, we define $\cost_k$ to be the weight of the resulting
$k$-component spanning forest, equivalently, the minimum total weight of any
spanning forest with exactly $k$ connected components. The corresponding
\emph{SLC cost profile} is $(\cost_1,\ldots,\cost_n)$, and the scalar quantity
$\cost(G)=\sum_{k=1}^{n}\cost_k$ is the area under this profile.

We study the problem of approximating these quantities in sublinear time. We
assume that the input is a weighted graph $G$ of average degree $d$ with edge
weights in $\{1,\dots,W\}$, accessed through adjacency-list queries; missing
edges are treated as having infinite distance. Our main result is a
sampling-based algorithm that outputs a succinct sketch of the entire SLC cost
profile in the distance setting. The algorithm runs in
$\widetilde{O}(d\sqrt{W}/\varepsilon^3)$ time and returns a sketch from which one can derive estimates
$(\widehat{\cost}_1,\ldots,\widehat{\cost}_n)$ satisfying
\[
\sum_{k=1}^{n}\bigl|\widehat{\cost}_k-\cost_k\bigr|
\le \varepsilon\,\cost(G).
\]

Thus, we obtain an $\ell_1$ approximation to the full profile whose error is at
most an $\varepsilon$-fraction of the area under the true profile. In
particular, this yields a $(1\pm\varepsilon)$-approximation to $\cost(G)$
within the same running time. We also prove a nearly matching lower bound of
$\Omega(d\sqrt{W}/\varepsilon^2)$ queries for estimating $\cost(G)$.

We further extend our results to the similarity setting, where SLC is defined
via a maximum spanning tree. In this case, we obtain algorithms with running
time $\widetilde{O}(dW/\varepsilon^3)$ for both the profile and the total cost,
together with a nearly matching lower bound of $\Omega(dW/\varepsilon^2)$
queries. These bounds reveal a genuine separation between the distance and
similarity settings. Finally, we extend our algorithms to metric spaces, where
we obtain $\widetilde{O}(n/\varepsilon^7)$-query algorithms for both distance
and similarity metrics, and we complement our theory with experimental
validation.
\end{abstract}

\thispagestyle{empty}

\newpage
\tableofcontents
\thispagestyle{empty}

\newpage
\pagenumbering{arabic}
\section{Introduction}
Hierarchical clustering is a fundamental and widely used technique in data analysis and machine learning. In agglomerative (bottom-up) clustering, we start with a set of $n$ objects and a defined distance or similarity relationship between these objects, as well as a method to evaluate the distance or similarity between clusters, i.e. groups of objects. The clustering process begins by treating each object as its own individual cluster. It then iteratively merges pairs of clusters based on a specified criterion related to their distance or similarity, continuing this process until only a single cluster remains. This merging process implicitly defines a hierarchy of clusters, and for every choice of the number of clusters $k$, the corresponding clustering is the one obtained after the first $n-k$ merges. Agglomerative clustering—typically implemented through methods such as single-linkage, complete-linkage, and average linkage—has been extensively studied and widely used in various fields, including data analysis and machine learning (see e.g. \cite{hastie2009elements,fortunato2010community}). These methods are well-established in standard data analysis toolkits. 

In this work, we focus on one of the most widely used agglomerative clustering methods: single-linkage clustering (SLC). When the underlying relationship between objects is defined in terms of \emph{distance}, SLC operates by iteratively identifying the \emph{closest} pair of objects that belong to different clusters and connecting them with an edge. This approach emphasizes the merging of clusters based on proximity. To mathematically represent this process, we utilize a graph $ G $ where each vertex corresponds to an object, and the weight of an edge indicates the distance between these objects. In this framework (see e.g. \citep{gower1969minimum}), SLC can be formalized by first constructing a minimum spanning tree (MST) $ T $ of $ G $. The edges of $ T $ are then considered in \emph{non-decreasing} order of their weights to facilitate cluster merging. This iterative process continues until the desired number of clusters, $ k $, is reached, resulting in a $ k $-clustering. In particular, when $ k=1 $, it reveals the complete hierarchical structure of the graph. In cases where the relationship between objects is based on \emph{similarity} rather than distance, SLC operates by iteratively identifying the \emph{most similar} pair of objects that belong to different clusters and can be defined using a \emph{maximum} spanning tree. Specifically, we first construct a maximum spanning tree (MaxST) $ T $ of graph $ G $ and then order the edges in $ T $ in \emph{non-increasing} order of their weights, subsequently merging clusters iteratively. 

SLC typically requires at least linear time for implementation. In many applications, the underlying graph can be massive, making computational analysis challenging -- simply reading the input may become impractical or even impossible. This motivates to design \emph{sublinear-time algorithms} that only read a small portion of the input to extract the SLC information. Since producing the full clustering hierarchy necessarily takes linear time, we instead ask: 
\begin{center}
\begin{tcolorbox}[colback=offwhite, colframe=lightgray, boxrule=1.5pt, arc=5mm, width=\textwidth]
\emph{Can we define \textbf{cost measures} that effectively capture the behavior and structure of SLC, and can we approximate these measures -- as well as key summary statistics, such as a \textbf{succinct representation} of the hierarchy -- in sublinear time?}
\end{tcolorbox}
\end{center}

\subsection{Our Contributions}
We answer the above question in the affirmative. Our main contribution is to formalize the \emph{single-linkage cost profile} and to give nearly optimal sublinear-time algorithms for approximating this profile, the scalar summaries
derived from it, and their analogues in the similarity and metric settings.

\subsubsection{Formalizing the SLC Cost Profile in Distance Setting}
\paragraph{Distance graphs and the single-linkage cost profile.} 
We are given a connected weighted graph $G=(V,E,w)$ with edge weights $w(e)\in [1,W]$, for some parameter $W\geq 1$. The edge weights are the distances between the objects; for any pair $i,j$ that does not form an edge, we define $w((i,j))=\infty$.

In single-linkage clustering, the partition with $k$ clusters is obtained by
deleting the $k-1$ heaviest edges of an MST. This
suggests a natural notion of cost at level $k$: the total weight of the
resulting spanning forest. Let
$
w_1 \le \cdots \le w_{n-1}
$ 
be the edge weights of an MST of $G$, listed in nondecreasing order. We define
the \emph{$k$-clustering cost} by
\[
\cost_k := \sum_{i=1}^{n-k} w_i .
\]

Equivalently, $\cost_k$ is the weight of the forest obtained by deleting the
$k-1$ largest-weight edges of an MST. 

The central object of this paper is the resulting \emph{single-linkage cost
profile}
\[
\mathbf{p}(G) := (\cost_1,\cost_2,\ldots,\cost_n).
\]
The profile vector records, for every level $k$, the cost of the $k$-clustering
produced by SLC. Thus, rather than summarizing the hierarchy by a single
number, $\mathbf{p}(G)$ captures its behavior across all resolutions.

The next proposition shows that $\cost_k$ is not an external objective imposed
on single-linkage clustering. Rather, it is exactly the flat spanning-forest objective
optimized by SLC at level $k$.

\begin{restatable}[Flat spanning-forest characterization of SLC]{prop}{flatforest}
\label{prop:slc_flat_spanning_forest}
Let $G=(V,E,w)$ be a connected weighted graph, and let $\mathcal{F}_k$ denote
the set of spanning forests of $G$ with exactly $k$ connected components. Then
\[
\cost_k = \min_{F\in \mathcal{F}_k} w(F),
\qquad\text{where } w(F):=\sum_{e\in F} w(e).
\]

Moreover, if $T$ is an MST of $G$, then deleting the $k-1$ largest-weight edges
of $T$ yields an optimal forest attaining $\cost_k$.
\end{restatable}

This is a standard consequence of Kruskal's algorithm and the cut and cycle
properties of minimum spanning trees (see \cref{subsec:deferred_proofs} for the proof).

\paragraph{Interpreting the SLC profile.} 
The profile has a simple and useful derivative interpretation. If the MST edge
weights are $w_1\le \cdots \le w_{n-1}$, then
\[
\Delta_k \;:=\; \cost_k-\cost_{k+1} \;=\; w_{n-k}.
\]

Hence $\Delta_k$ is exactly the cost of the merge that reduces the hierarchy
from $k+1$ clusters to $k$ clusters. Large gaps in $\Delta_k$ therefore
indicate natural scales at which further merging becomes expensive.

\begin{example}
As a simple illustration, consider two graphs on five vertices whose MST edge
weights are
\[
G_A:\ 1,1,1,9
\qquad\text{and}\qquad
G_B:\ 1,1,5,5.
\]

Both have MST weight $12$, but their SLC profiles are different:
\[
\mathbf{p}(G_A)=(12,3,2,1,0),
\qquad
\mathbf{p}(G_B)=(12,7,2,1,0).
\]

The first profile has a sharp final merge, indicating two well-separated
high-level clusters, whereas the second has a much less pronounced elbow.
Thus, the profile captures multiscale merge structure that is not visible from
the MST weight alone. 
\end{example}

\cref{subsec:profile_validation} provides empirical evidence that the profile
and its discrete derivative recover natural cluster scales and reflect
hierarchical structure in both synthetic and real data.

\paragraph{Total cost as area under the profile.}
Alongside the full profile, we study the scalar quantity, which we call the \emph{total cost} of the SLC hierarchy: 

\begin{align}
    \cost(G)\;:=\;\sum_{k=1}^{n} \cost_k =\sum_{k=1}^{n-1} \cost_k =\sum_{i=1}^{n-1} w_i+\sum_{i=1}^{n-2}w_i+\dots+\sum_{i=1}^2w_i+w_1= \sum_{i=1}^{n-1}(n-i)w_i \label{eqn:costG_MST}
\end{align}

We interpret $\cost(G)$ as the \emph{area under the SLC cost profile}. It is
not intended to replace the level-wise quantities $\cost_k$; rather, it serves
as a compact summary of the entire hierarchy.

\subsubsection{Sublinear Algorithms for Approximating the Profile and the Cost}
Our first main result is a sublinear-time
procedure that constructs a succinct sketch (also referred to as a succinct representation) of the entire SLC cost profile in
the distance setting.

\paragraph{Computational model.}
Except in the metric-space setting discussed later, our algorithms operate in
the \emph{adjacency-list model}. In this model, given a vertex $v$, the
algorithm may query its $i$-th neighbor (together with the corresponding edge
weight) in $\Theta(i)$ time; see \Cref{sec:preliminaries}. The query complexity
of an algorithm is the maximum number of such queries made on any input.
Our lower bounds continue to hold even in the stronger model in which the
$i$-th neighbor can be accessed in $O(1)$ time.

With this access model in place, we obtain the following profile guarantee.

\begin{restatable}{thm}{distanceprofile}
\label{thm:profile}
Assume that $\sqrt{W}\leq n$ and let $0< \varepsilon<1$ 
be a parameter. 
\cref{alg:appprofile} generates a succinct representation
    of an approximation $(\widehat{\cost}_1,\dots,\widehat{\cost}_{n})$ of the SLC profile 
    vector $(\cost_1,\dots,\cost_{n})$ in the distance graph such that  
    with probability at least $3/4$, it holds that
    \[
    \sum_{k=1}^{n} \abs{\widehat{\cost}_k-\cost_k}\leq \varepsilon\cdot\cost(G).
    \]
    
    The query complexity and running time 
    of the algorithm are $O(\frac{\sqrt{W}d}{\varepsilon^3}\log^4(\frac{Wd}{\varepsilon}))$ in expectation.
\end{restatable}

The guarantee in \Cref{thm:profile} is an $\ell_1$ approximation to the entire
SLC profile. It should \emph{not} be interpreted as a uniform multiplicative
approximation that holds simultaneously for every level $k$; such a guarantee
is impossible in sublinear time (see \cref{prop:uniform_per_level_impossible}). 

Nevertheless, the sketch supports efficient point queries. Given the succinct
representation and any specified $k\in\{1,\dots,n\}$,
\Cref{alg:profile_oracle} returns an estimate $\widehat{\cost}_k$ of $\cost_k$
in time $O(\log(\log W/\varepsilon))$. Moreover,
\Cref{lem:single_costk_error} shows that, for every $k$,
\[
\bigl|\widehat{\cost}_k-\cost_k\bigr|
=
O\!\left(
\varepsilon\bigl(\cost_k+\max\{k,n/\sqrt{W}\}\,W\bigr)
\right).
\]

We also obtain a sublinear-time estimator for the scalar quantity $\cost(G)$. 

\begin{restatable}{thm}{distance}
\label{thm:cost_distance}
Let $G$ be a weighted graph with edge weights in $\{1,\dots, W\}$ with average 
(unweighted) 
degree $d$. 
Assume that $\sqrt{W}\leq n$ and let $0< \varepsilon<1$ 
be a parameter.
\cref{alg:appcost} outputs an estimate $\widehat{\cost}(G)$ of the 
single-linkage clustering cost $\cost(G)$ in the distance graph 
such that with probability at least $3/4$,
\[
(1-\varepsilon)\cost(G)\leq \widehat{\cost}(G)\leq (1+\varepsilon)\cost(G).
\]

The query complexity and running time 
of the algorithm are $O(\frac{\sqrt{W}d}{\varepsilon^3}\log^4(\frac{Wd}{\varepsilon}))$ in expectation.
\end{restatable}

To complement the upper bound, we prove a nearly matching lower bound for
estimating $\cost(G)$ in the distance setting.

\begin{restatable}{thm}{distancelowerbound}
\label{thm:cost_distance_lowerbound}
Let $\frac{W^{1/4}}{\sqrt{40 n}}<\varepsilon<\frac{1}{2}$ and $W>1$. Any algorithm that $(1+\varepsilon)$-approximates the SLC cost $\cost(G)$ in the distance graph with success probability at least $3/4$ needs to make $\Omega(d\sqrt{W}/\varepsilon^2)$ queries.  
\end{restatable}

\paragraph{More discussions on the profile guarantee.}

Note that \Cref{thm:profile} gives only an average guarantee across the levels
of the hierarchy: informally, it approximates the cost of the $k$-clustering
within a $(1+\eps)$ factor on average over $k$. We further show that this is
essentially the strongest per-level guarantee attainable in sublinear time, as
no sublinear-time algorithm can achieve a $(1+\eps)$-approximation to $\cost_k$
uniformly over all levels $k$.

\begin{restatable}[Uniform per-level relative approximation is impossible]{prop}{perlevelimpossible}
\label{prop:uniform_per_level_impossible}
No sublinear-time algorithm can give a uniform multiplicative approximation to
$\cost_k$ for all $k$. In particular, even approximating $\cost_{n-1}$ within a
factor $1+\varepsilon$, for any $\varepsilon<\sqrt{2}-1$, requires $\Omega(m)$
edge inspections in the worst case, even when all edge weights lie in
$\{1,2\}$.
\end{restatable}

The obstruction already appears at the finest nontrivial level:
$\cost_{n-1}$ is exactly the minimum edge weight in the graph. Distinguishing
an all-$2$-weight graph from one containing a single hidden edge of weight $1$
already requires linear inspection in the worst case. See \cref{subsec:deferred_proofs} for the proof.

\subsubsection{Similarity Setting and Metric Spaces}
We next extend the framework to the similarity setting and generalize the algorithms to metric spaces.

\paragraph{SLC in the similarity graph.} 
We consider weighted graphs $G=(V,E,w)$ in which each edge weight
$w(e)\in[1,W]$ represents similarity, and non-edges are assigned similarity
$0$. Let
$
w^{(s)}_1 \ge \cdots \ge w^{(s)}_{n-1}
$ 
be the edge weights of a maximum spanning tree (MaxST) of $G$, listed in
non-increasing order. For each $k\in\{1,\ldots,n\}$, define
\[
\cost_k^{(s)} := \sum_{i=1}^{n-k} w_i^{(s)}.
\]

Equivalently, $\cost_k^{(s)}$ is the weight of the forest obtained by deleting
the $k-1$ smallest-weight edges of a MaxST. The corresponding
\emph{similarity cost profile} is
\[
\mathbf{p}^{(s)}(G)
=
(\cost_1^{(s)},\cost_2^{(s)},\ldots,\cost_n^{(s)}).
\]

As in the distance setting, this quantity admits a flat optimization
interpretation.

\begin{proposition}[Flat spanning-forest characterization in the similarity setting]
\label{prop:slc_flat_spanning_forest_sim}
Let $\mathcal{F}_k$ denote the set of spanning forests of $G$ with exactly $k$
connected components. Then
\[
\cost_k^{(s)}=\max_{F\in\mathcal{F}_k} w(F).
\]

Moreover, if $T$ is a maximum spanning tree of $G$, then deleting the
$k-1$ smallest-weight edges of $T$ yields an optimal forest attaining
$\cost_k^{(s)}$.
\end{proposition}

We also define the total similarity hierarchy cost
\begin{equation}
\cost^{(s)}(G) = \sum_{k=1}^{n}\cost_k^{(s)} = \sum_{k=1}^{n-1}\cost_k^{(s)} = \sum_{i=1}^{n-1}(n-i) \cdot w_i^{(s)}, \label{eqn:costG_similarity}
\end{equation}
which is the area under the similarity profile. If the MaxST edge weights are
sorted as above, then the discrete derivative satisfies
\[
\Delta_k^{(s)}
:=
\cost_k^{(s)}-\cost_{k+1}^{(s)}
=
w_{n-k}^{(s)}.
\]

Thus, the derivative records the exact similarity threshold of the merge that
reduces the hierarchy from $k+1$ clusters to $k$ clusters.

\paragraph{Sublinear Algorithms for Similarity Graphs}
We obtain sublinear-time algorithms for approximating both the similarity
profile and the total hierarchy cost.

\begin{restatable}{thm}{similarityprofile}
\label{thm:profile_sim}
Assume that $W\le n$ and let $\varepsilon<1$.
    \cref{alg:appprofile_similarity} generates a succinct representation
    of an approximation $(\widehat{\cost^{(s)}}_1,\dots,\widehat{\cost^{(s)}}_{n})$ of 
the SLC profile vector $(\cost_1^{(s)},\dots,\cost_{n}^{(s)})$ in the similarity graph such that 
    with probability at least $3/4$, it holds that
    \[
    \sum_{k=1}^{n}|\widehat{\cost^{(s)}}_k-\cost^{(s)}_k|\leq\varepsilon\cdot\cost^{(s)}(G).
    \]
    The query complexity and running time 
    of the algorithm are $O(\frac{Wd}{\varepsilon^3}\log^4(\frac{Wd}{\varepsilon}))$ in expectation. 
\end{restatable}

\begin{restatable}{thm}{similarity}
\label{thm:cost_similarity}
Let $G$ be a weighted graph with edge weights in $\{1,\dots, W\}$ with average (unweighted) degree $d$. 
Assume that $W\le n$ and let $0<\varepsilon<1$ be a parameter.
\cref{alg:appcost_sim} outputs an estimate $\widehat{\cost^{(s)}}(G)$ 
of the single-linkage clustering cost $\cost^{(s)}(G)$ in the similarity graph 
such that with probability at least $3/4$,
\[
(1-\varepsilon)\cost^{(s)}(G)\leq \widehat{\cost^{(s)}}(G)\leq 
(1+\varepsilon)\cost^{(s)}(G).
\]

The query complexity and running time of the algorithm are $O(\frac{Wd}{\varepsilon^3}\log^4(\frac{Wd}{\varepsilon}))$ in expectation.
\end{restatable}

To complement our algorithmic result, we establish a nearly matching information-theoretic lower bound for estimating $\cost^{(s)}(G)$.

\begin{restatable}{thm}{similaritylowerbound}
\label{thm:cost_similarity_lowerbound}
Let $\sqrt{\frac{W}{40 n}}<\varepsilon<\frac{1}{2}$ and $W>10$. Any algorithm that $(1+\varepsilon)$-approximates the SLC cost $\cost^{(s)}(G)$ in the similarity graph with success probability at least $3/4$ needs to make $\Omega(d W/\varepsilon^2)$ queries.  
\end{restatable}

Together, these bounds reveal a fundamental separation between the two settings: for any small constant $\eps>0$, 
distance-based estimation admits $\widetilde \Theta(d\sqrt{W})$ complexity, whereas
similarity-based estimation inherently requires $\widetilde \Theta(dW)$ complexity.

\paragraph{Metric Space.}
Finally, we extend our results to metric spaces, where the relationship between vertices must satisfy the triangle inequality. In this model, the algorithm does not receive an explicit graph; instead, it may query the distance or similarity of any specified pair of points in $O(1)$ time.

When the metric represents distances, we achieve the following guarantee:

\begin{restatable}{thm}{metricdistance}
\label{thm:cost_metric_distance}
Let $G$ be an $n$-point graph in metric space, where each edge weight represents \emph{distance} between two vertices, and let $0< \varepsilon<1$. \cref{alg:appcost_metric} outputs an estimate $\widehat{\cost}(G)$ of single-linkage clustering cost $\cost(G)$ in metric space, such that with probability at least $3/4$,
\[
(1-\varepsilon)\cost(G)\leq \widehat{\cost}(G)\leq (1+\varepsilon)\cost(G).
\]

The query complexity and running time of the algorithm are $\Tilde{O}(n/\varepsilon^7)$ in expectation.
\end{restatable}

An analogous result, with the same approximation guarantee and query
complexity, holds when the metric encodes similarities; see
\Cref{thm:cost_metric_similarity}. Moreover, these bounds are nearly optimal in
their dependence on $n$: even for constant $\eps$, any
$(1+\eps)$-approximation requires $\Omega(n)$ oracle queries, as shown in \cref{subsec:deferred_proofs}.

\subsection{Technical Overview}
Although the SLC profile is the primary object of the paper, the algorithmic
development begins with the scalar quantity $\cost(G)$, the area under the
profile. This is where the central technical difficulty already appears, while also being conceptually simpler. Once we
show how to estimate $\cost(G)$ in sublinear time, the same machinery can be
lifted to a succinct oracle for the entire profile. 

Our starting point is the component-counting paradigm of Chazelle, Rubinfeld,
and Trevisan~\cite{chazelle2005approximating} (CRT) for estimating the MST weight, but the single-linkage
setting introduces two difficulties that are absent from MST estimation. First,
the MST weight is a \emph{linear} functional of the connected-component counts
of threshold graphs, whereas the SLC hierarchy cost is a \emph{quadratic}
functional. Second, our goal is not merely to estimate one scalar, but to
construct a succinct approximation to an entire monotone hierarchy. Moreover,
the algorithm only has access to \emph{noisy} estimates of the component counts,
so the monotonicity needed for binary search is not directly available. The technical
overview below explains how we overcome these obstacles.

\paragraph{Starting Point: the CRT Paradigm.} Given a weighted graph $G$ with edge weights in $\{1,\dots, W\}$, let $G_j$ be the subgraph containing all edges of weight at most $j$, and let
$c_j$ be its number of connected components. CRT showed that the MST weight of $G$ can
be written as
\[
\cost(\MST)=n-W+\sum_{j=1}^{W-1} c_j,
\]
and gave a sublinear algorithm for estimating each $c_j$ by sampling vertices
and performing truncated BFS. Their estimator achieves additive
error $\varepsilon n$ in expected time $\widetilde{O}(d/\varepsilon^2)$, and
combining the estimates over all thresholds yields a
$\widetilde{O}(dW/\varepsilon^2)$-time approximation to the MST weight.

We use this framework as a starting point, but the profile problem cannot be
reduced to a direct application of CRT. The reason is that neither the target
quantity nor the required accuracy regime matches MST estimation.

\paragraph{Distance Setting: Estimating $\cost(G)$}
While \cref{eqn:costG_MST} defines the total cost $\cost(G)$ based on the sorted MST edge weights, finding the exact MST requires at least linear time.
We therefore reformulate $\cost(G)$ in terms of the threshold graphs
$G_j$. We show in \cref{thm:formula_cost} that 
\[
\cost(G)
=
\frac{n(n-1)}{2}
+
\frac{1}{2}\sum_{j=1}^{W-1}(c_j^2-c_j).
\]

This identity is the first major departure from CRT: the hierarchy cost is a
quadratic, not linear, functional of the connected-component counts.

A naive strategy would therefore estimate every $c_j$ separately and plug the
results into the formula above. However, to obtain a $(1+\varepsilon)$
approximation to $\cost(G)$, one would need much finer additive control on the
$c_j$'s than in CRT, and doing this for all $W$ thresholds would be too
expensive.

To address this, we first refine the component estimator itself. We show in \Cref{lem:appncc} that
each $c_j$ can be estimated within additive error
\[
O\!\left(\varepsilon\max\left\{\frac{n}{\sqrt{W}},\,c_j\right\}\right)
\]
in expected time $\widetilde{O}(\sqrt{W}d/\varepsilon^2)$. This already
improves substantially over the naive reduction, but estimating all $W$ values
individually would still be too slow.

The second step is to exploit monotonicity. Since $G_j$ gains
edges as $j$ increases, the sequence $(c_1,\ldots,c_W)$ is nonincreasing, and
hence so is $(c_j^2-c_j)$. This suggests partitioning the range of
component counts into intervals, representing all values in one interval by a
single breakpoint, and then using binary search to determine how many values
fall into each interval. If one could access the exact sequence $(c_j)$, this
would reduce the number of required probes from $W$ to only
$O(\log^2 W/\varepsilon)$.

The main obstacle is that we do \emph{not} observe $(c_j)$ exactly. We only have
noisy estimates $\widehat{c}_j$, and the sequence
$(\widehat{c}_1,\ldots,\widehat{c}_W)$ need not be monotone. This leads to our
main algorithmic abstraction.

\paragraph{Noisy Monotone-Sequence Sketching.}

In \cref{sec:binarysearch}, we isolate the following abstract problem: given a nonincreasing
sequence
\[
X=(x_1,\ldots,x_W)
\]
and only noisy estimates $\widehat{x}_j$ satisfying
\[
|\widehat{x}_j-x_j|\le T_j,
\]
how can one build a succinct approximation of the entire sequence without
querying every entry?

Our solution is a robust binary-search framework for noisy monotone sequences.
We define a \emph{valid discretization} of the value range into intervals whose
widths are chosen to dominate the local error bounds $T_j$. This guarantees that
if $x_j$ lies in a given interval, then $\widehat{x}_j$ lies in the same
interval or a neighboring one. Even though the noisy sequence may not be
sorted, binary search on the $\widehat{x}_j$'s still preserves a weak
monotonicity: searching for smaller keys never returns smaller indices. This is
enough to recover consistent interval counts and hence a succinct approximation
to the sequence.

Applying this abstraction to $(c_j)$ yields a sublinear-time estimator for
$\cost(G)$ using only $O(\log^2 W/\varepsilon)$ threshold queries, for a total
running time of
$
\widetilde{O}\!\left(\frac{\sqrt{W}d}{\varepsilon^3}\right).
$

\paragraph{From the Scalar Estimator to a Profile Oracle.}

We then lift the same sketching framework from the scalar quantity $\cost(G)$ to
the entire profile $\mathbf{p}(G)=(\cost_1,\ldots,\cost_n)$.

The key additional observation (see \cref{eqn:cost_k}) is that for each $k$,
\[
\cost_k
=
n+\sum_{j=1}^{w_{n-k}-1} c_j - k\,w_{n-k},
\]
where $w_{n-k}$ is the $(n-k)$-th MST edge weight. In particular, when
$k=c_j$, this gives an explicit formula for $\cost_k$ in terms of the threshold
counts up to level $j$.

Using the same interval endpoints obtained from the noisy binary-search
procedure, we define representative values $\overline{\cost}_{B_i}$ at the
breakpoints and use them as a succinct profile sketch. Since the profile is
monotone in $k$, every $\cost_k$ can be approximated by the representative
associated with the interval containing $k$. This yields a succinct profile
oracle with two properties: it answers any specified level $k$ in
$O(\log(\log W/\varepsilon))$ time, and the total $\ell_1$ error over all
levels is at most $\varepsilon\cost(G)$.

Thus, although the scalar quantity $\cost(G)$ is the technical entry point, the
same machinery ultimately provides a succinct approximation to the full
hierarchy.

\paragraph{Similarity Setting.}

The similarity setting is \emph{not} a formal mirror image of the distance case. Here
the natural threshold graph $G_j^{(s)}$ contains all edges of similarity at
least $j$, and if $c_j^{(s)}$ denotes its number of connected components, then
we derive the identity
\[
\cost^{(s)}(G)
=
\frac{1}{2}\sum_{j=1}^{W}\bigl(c_j^{(s)}+n-1\bigr)\bigl(n-c_j^{(s)}\bigr).
\]
The delicate term is now not just $c_j^{(s)}$, but also
\[
D_j := n-c_j^{(s)},
\]
which measures how far the threshold graph is from being completely disconnected.

A direct adaptation of the distance-case component estimator is insufficient,
because when $c_j^{(s)}$ is large, the quantity $D_j$ may be small, and an
additive error proportional to $c_j^{(s)}$ is too crude. To resolve this, we
separately estimate the number of non-isolated vertices and the number of
connected components among the non-isolated part of the threshold graph. This
yields a refined estimator for $D_j$ with additive error
\[
O\!\left(\varepsilon\max\left\{\frac{n}{W},\,\min\{D_j,n-D_j\}\right\}\right).
\]

We then apply the same noisy binary-search framework as in the distance case,
but the interval design is more intricate because the error bounds now behave
differently in different regimes of $D_j$. This leads to a partition of the
range $[0,n-1]$ into five subranges, each handled by either an arithmetic or a
geometric progression of breakpoints. The resulting algorithm runs in
$
\widetilde{O}\!\left(\frac{Wd}{\varepsilon^3}\right)
$
time, which is asymptotically larger than in the distance case.

\paragraph{Lower Bounds and Metric Spaces.}

Our lower bounds build on the same distributional framework as CRT, namely the
problem of distinguishing between two biased coin distributions. We encode these
instances into weighted graph distributions whose single-linkage hierarchy costs
differ by a factor of $1+\varepsilon$, yet cannot be distinguished with too few
queries. This yields nearly matching lower bounds in both the distance and
similarity settings and, in particular, explains the separation between the
$\widetilde{O}(d\sqrt{W})$ complexity of the distance case and the
$\widetilde{O}(dW)$ complexity of the similarity case.

Finally, we extend the framework to metric spaces. Here our starting point is
the representative-point method of Czumaj and Sohler \cite{czumaj2009estimating} for MST estimation in
metrics. We express the hierarchy cost as a quadratic functional of component
counts across $\widetilde{O}(\log n/\varepsilon)$ threshold graphs and combine
this with a refined analysis of the representative-based estimator. This yields
sublinear-time approximation algorithms for both distance and similarity metrics.

\subsection{Related Work}
Hierarchical clustering has been extensively studied in the context of approximation algorithms with polynomial running times. A significant body of work has explored hierarchical clustering in metric and Euclidean spaces, including studies on hierarchical $k$-median \citep{lin2010general} and the use of the largest cluster radius as a cost measure \citep{dasgupta2005performance}. Recently, \cite{dasgupta2016cost} introduced a cost function for similarity-based hierarchical clustering and proposed an approximation algorithm for this cost. Several improvements and generalizations have since emerged \citep{charikar2017approximate,moseley2023approximation,cohen2019hierarchical,charikar2019hierarchical}.
{Other recent work on hierarchical correlation clustering and fitting distances in ultrametrics, such as \cite{an2025handlinglproundinghierarchicalclustering}, has further improved approximation algorithms in the hierarchical setting.}

In terms of sublinear algorithms for hierarchical clustering, \cite{agarwal2022sublinear} studied such algorithms concerning Dasgupta cost \citep{dasgupta2016cost} within dynamic streaming, query, and massively parallel computation models. 
\cite{bateni2017affinity} also considered affinity and single-linkage clustering in the massively parallel computation setting.
\cite{assadi2022hierarchical} focused on streaming algorithms for identifying a hierarchical clustering with low Dasgupta cost and estimating the value of the optimal hierarchical tree. Additionally, \cite{kapralov2023learning} provided a sublinear-time hierarchical clustering oracle for graphs exhibiting significant flat clustering, and \cite{kapralov2025approximatingdasguptacostsublinear} estimated Dasgupta cost for well-clusterable graphs in sublinear-time. Other sublinear algorithms addressing graph clustering with conductance-based measures include local graph clustering \citep{spielman2013local,andersen2006local} and spectral clustering oracles for well-clusterable graphs \citep{peng2020robust,gluch2021spectral,shen2024sublinear}.

Our work is closely related to a series of studies on estimating the number of connected components and the weight of minimum spanning trees. \cite{chazelle2005approximating} pioneered a sublinear-time algorithm for these problems using sampling and truncated BFS. Subsequent investigations have explored these problems in various contexts, including Euclidean space \citep{czumaj2005approximating}, metric space \citep{czumaj2009estimating}, and graph streaming \citep{huang2019dynamic,peng2018estimating}.

\paragraph{Organization}

The rest of the paper is organized as follows.
We present preliminaries in \cref{sec:preliminaries}, and describe our connected component estimation algorithm and give its analysis in \cref{sec:estimateNCC}.
In \cref{sec:distancemeasure}, we express the total clustering cost $\cost(G)$ via connected components.
We then give our binary search strategy for succinctly approximating noisy monotone sequence in \cref{sec:binarysearch}. In \cref{sec:algorithmfordistance}, we prove \cref{thm:cost_distance} for distance-based clustering, and \cref{thm:profile}. 
We extend our approach to the similarity setting in \cref{sec:similaritymeasure}, introducing a new algorithm and proving \cref{thm:cost_similarity} and \cref{thm:profile_sim}.
Lower bounds for both settings are given in \cref{sec:lowerbounds}.
Additional proofs, generalizations, metric space results, and experimental validation are deferred to the appendix for clarity.

\section{Preliminaries}\label{sec:preliminaries}

\paragraph{Model of Computation}
We will assume that the algorithm is given access to an adjacency list representation of the (undirected and connected) input graph $G=(V,E)$, except when $G$ is a metric graph. 
That is, we assume w.l.o.g. that $V=\{1,\dots, n\}$ and
that the parameter $n$ is given to the algorithm. The edges of the graph are stored in an array of size $n$
whose $i$-th entry is a pointer to the list of neighbors of vertex $i$. With each neighbor $j$ that appears in the adjacency list of vertex $i$, we store the weight $w((i,j))$.
In particular, computing the degree $\deg(v)$ of a vertex $v$ requires scanning through all of its neighbors and so this can be done in $O(\deg(v))$ time.
We will use $d$ to denote the average vertex degree of $G$. We emphasize that the algorithm has the ability to query the $i$-th neighbor of a specified vertex $v$ in $\Theta(i)$ time, for any given $i\leq \deg(v)$.

\paragraph{Probability Amplification}
Consider an algorithm $A$ which outputs $\mathrm{out}_A$, with the assumption that $\mathrm{out}_A$ 
is correct with constant probability greater than $2/3$. We can amplify
the success probability by constructing a new algorithm $A^*$
that runs $C\cdot \log (1/\delta)$ independent instances of $A$, for some sufficiently large constant $C>0$. The output $\mathrm{out}_{A^*}$
is then set to be the median of the output values from these instances. By Chernoff bound it follows
that $\mathrm{out}_{A^*}$ is correct with probability at least $1-\delta$,
where $\delta$ can be made arbitrarily small (see e.g. \cite{motwani1995randomized}).

\paragraph{Chernoff--Hoeffding bound} We will make use of the following Chernoff--Hoeffding bound (see Theorem 4.4 and Theorem 4.5 in \cite{mitzenmacher2017probability}).
\begin{theorem}[The Chernoff--Hoeffding bound]\label{thm:chernoff}
	Let $t\geq 1$. Let $X=\sum_{1\leq i\leq t}X_i$, where $X_i, 1\leq i\leq t$, are independently distributed in $[0,1]$. Then for all $0<\varepsilon\leq 1$, 
	\[
	\Pr[X\geq(1+\varepsilon)\E[X]] \leq e^{-\E[X]\varepsilon^2/3},\text{ }
         \Pr[X\leq(1-\varepsilon)\E[X]] \leq e^{-\E[X]\varepsilon^2/2}.
         \]
\end{theorem}

\section{Estimating the Number of Connected Components}\label{sec:estimateNCC}
To estimate the cost of single-linkage clustering, we first need an
algorithm that efficiently approximates the number of connected
components in a subgraph $H$ of a graph $G$. This approximation
serves as a crucial step in our overall methodology and builds upon an
algorithm to approximate the number of connected components and the cost of a minimum spanning tree from
\citep{chazelle2005approximating}. We give the formal description in  \cref{alg:appncc}. We start by considering a slight modification of an algorithm from \citep{chazelle2005approximating}
to approximate the number of connected components and we give an improved analysis for the case that
their number is relatively small. 
We remark that the edges of subgraph $H$ are implicitly defined (as a function of the weight of the edge in $G$ and its vertices; for example, it could be the set of edges above or below a certain weight), so we must
inspect \emph{all} edges incident to a vertex $v$ in $G$ to find its neighbors in $H$.

\begin{algorithm}[h!]
    \DontPrintSemicolon
    \caption{\textsc{ApproximateConnectedComponents}($G, H, \varepsilon,k,d$)}
    \label{alg:appncc}
    \SetKwInOut{Input}{input}
    \SetKwInOut{Output}{output}
    \Input{graph $G$, implicit subgraph $H$, approx. parameter $\varepsilon$, threshold parameter
    $k$, avg. degree $d$}
    \Output{an estimate $\hat{c}$ of the number of connected components in $H$}
    \BlankLine
    choose $r=\lceil64k/\varepsilon^2\rceil$ vertices $u_1,\dots,u_r$ uniformly at random\;
    set threshold $\Gamma=\lceil4k/\varepsilon\rceil$ and $d^{(G)}=d\cdot\Gamma$\;
    
    \For{each sampled vertex $u_i$}{
    set $\beta_i=0$\;
    
    take the first step of BFS: identify neighbors of $u_i$ in $H$ by examining all incident edges in $G$\;
    let $d_{u_i}^{(G)}$ be the degree of $u_i$ in $G$\;
    \eIf{$u_i$ is isolated in $H$}
    {set $\beta_i=1$\;
    }
    {
        (*) flip a coin\;
        
        \If{(heads) \& (\# vertices visited $\in H<\Gamma$ during the BFS) \& (no visited vertex $\in H$
        has degree in $G>d^{(G)}$ during the BFS)}{
        resume BFS on $H$, doubling the number of visited edges in $G$\;
            \eIf{this allows BFS on $H$ to explore the whole connected component of $u_i$}
            {
            set $\beta_i=d_{u_i}^{(G)}\cdot 2^{\text{\# coin flips}}/\text{\# edges visited in $G$}$\;
            }
            {go to (*)\;}
        }
    }
    }
    \Return $\hat{c}=\frac{n}{r}\sum_{i=1}^r\beta_i$\; 
\end{algorithm}

The underlying intuition is as follows. For a vertex $u$ in a
connected component $C_u$ of subgraph $H$,
let $\vol(C_u)=\sum_{v\in C_u}d_v^{(G)}$,
where $d_v^{(G)}$ is degree of $v$ in graph $G$.
The sum $\sum_{u\in C_u} \frac{d_u^{(G)}}{\vol(C_u)}=1$ for each component $C_u$.
Therefore, across all vertices in $V$, we have:
$\sum_{u\in V}\frac{d_u^{(G)}}{\vol(C_u)}=c$,
where $c$ is the total number of connected components in subgraph $H$.
By sampling and averaging estimates $\beta_i$
for sampled vertices, we approximate the value of $c$.

At this point, we provide a lemma to establish a theoretical guarantee for \cref{alg:appncc}.

\begin{lemma}\label{lem:appncc}
Let $1>\varepsilon>0$ and $k\ge 1$.
Suppose (an upper bound on) the average degree 
$d$ of graph $G$ is known.
Given access to the graph $G$ and an implicit subgraph $H\subseteq G$ (where $H$ shares 
the same vertex set as $G$, and the edges of $H$ are determined by evaluating conditions 
on the edges of $G$), \Cref{alg:appncc} outputs $\hat{c}$ such that
\[
|\hat{c}-c| \leq \varepsilon\cdot\max\left\{\frac{n}{k}, c\right\},
\]

with probability at least $7/8$,
where $c$ is the number of connected components in $H$.
The query complexity and running time of the algorithm are $O(\frac{k}{\varepsilon^2}d\log(\frac{k}{\varepsilon}d))$ in expectation.
\end{lemma}

Before proving \cref{lem:appncc}, consider the following lemma, 
adapted from \cite{chazelle2005approximating}. 

\begin{lemma}[\cite{chazelle2005approximating}]\label{lem:onbetai}
Let $U$ be the set of vertices that lie in components in subgraph $H$
with fewer than $\Gamma$ vertices; and all of these vertices in 
the original graph are of degree at most $d^{(G)}$.  
The random variable $\beta_i$ in \cref{alg:appncc} is given by:
\begin{displaymath}
\beta_i = \left\{ \begin{array}{ll}
0, & \textrm{if $u_i\notin U$}\\
2^{\left\lceil\log\left(\frac{\vol(C_{u_i})}{d_{u_i}^{(G)}}\right)\right\rceil}
\frac{d_{u_i}^{(G)}}{\vol(C_{u_i})}, & 
w.p.\ 2^{-\left\lceil\log\left(\frac{\vol(C_{u_i})}{d_{u_i}^{(G)}}\right)\right\rceil}
\textrm{, if } \vol(C_{u_i})\neq0\\
1, & w.p. 1 \textrm{, if } \vol(C_{u_i})=0\\
0, & \textrm{otherwise}
\end{array} \right.
\end{displaymath}

The expectation and variance of $\beta_i$ is $\E[\beta_i]=c_U$ and
$\Var[\beta_i]\leq\frac{2c_U}{n}$.
\end{lemma}

\begin{proof}
    We let $H_{\nis}$ be the subgraph of $H$ induced by all non-isolated vertices w.r.t. $H$. Note that if a vertex 
    $u\in H\setminus H_{\nis}$, then $u$ is isolated in $H$.
    The number of vertices in $H_{\nis}$ is defined as $n'$, and let $c'_U$ denote the number of connected components in $H_{nis}[U]$.
    For each $1\leq i\leq r$, the expectation of $\beta_i$
    is given by conditional expectation:
    \begin{align*}
        \E[\beta_i]&=\Pr[u_i\in H_{nis}]\E[\beta_i|u_i\in H_{nis}]+\Pr[u_i\in H\setminus H_{nis}]\E[\beta_i|u_i\in H\setminus H_{nis}]\\
        &=\frac{n'}{n}\cdot\frac{1}{n'}\left(\sum_{u\in H_{nis}\setminus U}0+
        \sum_{u\in H_{nis}\cap U}2^{-\left\lceil\log\left(\frac{\vol(C_{u})}{d_{u}^{(G)}}\right)\right\rceil}
    \cdot2^{\left\lceil\log\left(\frac{\vol(C_{u})}{d_{u}^{(G)}}\right)\right\rceil}
    \frac{d_u^{(G)}}{\vol(C_{u})}\right)
    +\frac{n-n'}{n}\cdot 1\\
    &=\frac{1}{n}\cdot c'_U+\frac{1}{n}(n-n')
    \tag{since $\sum_{u\in U}\frac{d_u^{(G)}}{\vol(C_{u})}=c_U$}\\
    &=\frac{1}{n}\cdot c'_U+\frac{1}{n}(c_U-c'_U)
    \tag{since $n-n'=c_U-c'_U$ is the number of isolated vertices}\\
    &=\frac{c_U}{n}
    \end{align*}

Since when $\vol(C_{u_i})\neq 0$, $\beta_i\leq 2^{\log\left(\frac{\vol(C_{u_i})}{d_{u_i}^{(G)}}\right)+1} \frac{d_{u_i}^{(G)}}{\vol(C_{u_i})}\leq 2$, 
we have $\Var[\beta_i]\leq\E[\beta_i^2]\leq2\cdot\E[\beta_i]\leq\frac{2c_U}{n}$.

\end{proof}

Now we are ready to prove \Cref{lem:appncc}.

\begin{proof}[Proof of \Cref{lem:appncc}]
    Let $U$ be the set of vertices defined in \cref{lem:onbetai},
    and let $c_U$ be the number of connected components in the vertex-induced 
    subgraph $G[U]$.
    The number of connected components containing vertices with degree
    greater than $d^{(G)}$, is at most 
    $\frac{n\cdot d}{d^{(G)}}=\frac{n}{\Gamma}$.
    Furthermore, the number of connected components with size greater
    than $\Gamma$ is at most $\frac{n}{\Gamma}$. Then we have that, 
    \[c-\frac{2 n}{\Gamma}\leq c_U\leq c.\]

    By \cref{lem:onbetai}, $\E[\beta_i]=\frac{c_U}{n}$ and $\Var[\beta_i]\leq\frac{2 c_U}{n}$.
    Thus, $\E[\hat{c}]=\frac{n}{r}\cdot r\cdot \E[\beta_i]=c_U$, 
    leading to: $c-\frac{2 n}{\Gamma}\leq \E[\hat{c}]\leq c$.
    Furthermore, $\Var[\hat{c}]=\frac{n^2}{r^2}\cdot r\cdot \Var[\beta_1]\leq 
    \frac{n^2}{r}\cdot \frac{2c_U}{n} \leq \frac{2n c}{r}$.
    
    We bound the error of $\hat{c}$ by two cases: $c<\frac{n}{k}$, 
    and $c\geq\frac{n}{k}$. If $c< \frac{n}{k}$, by Chebyshev's inequality, 
    \[
    \Pr[|\hat{c}-c_U|\geq \frac{\varepsilon n}{2k}]\leq\frac{\Var[\hat{c}]
    \cdot 4k^2}{\varepsilon^2n^2}
    \leq\frac{2nc\cdot 4k^2}{r\cdot\varepsilon^2n^2}
    \leq\frac{8c\cdot k^2}{64k\cdot n}
    \leq\frac{1}{8}
    \]
    
    The last inequality holds as $c<\frac{n}{k}$. Therefore, the error of $\hat{c}$ is
    \[
    |\hat{c}-c|\leq|\hat{c}-c_U|+|c_U-c|\leq\frac{\varepsilon n}{2k}+\frac{2n}{\Gamma}
    \leq \frac{\varepsilon n}{2k}+\frac{\varepsilon n}{2k}
    =\frac{\varepsilon n}{k}
    \]
    
    Else, when $c\geq \frac{n}{k}$,
    \[
    \Pr[|\hat{c}-c_U|\geq\frac{\varepsilon c}{2}]\leq\frac{4\Var[\hat{c}]}{\varepsilon^2c^2}
    \leq\frac{8nc}{r\cdot\varepsilon^2c^2}
    \leq\frac{8n}{64k\cdot c}
    \leq\frac{1}{8}
    \]
    
    The last inequality holds as $c\geq\frac{n}{k}$. And the error of $\hat{c}$ is
    \[
    |\hat{c}-c|\leq|\hat{c}-c_U|+|c_U-c|\leq
    \frac{\varepsilon c}{2}+\frac{2n}{\Gamma}= \frac{\varepsilon c}{2}
    +\frac{\varepsilon n}{2k} \leq\varepsilon c
    \]

    Combining both cases, we have that 
    $\abs{\hat{c}-c}\leq \varepsilon\max\{\frac{n}{k},c\}$, 
    with probability more than $\frac{7}{8}$.  

    \paragraph{Running time analysis.} 
    Denote $\E[T]$ as the expected running time of \cref{alg:appncc}, and
    $\E[T(u)]$ as the expected running time of BFS when vertex $u$ is sampled.
    Then,
    \[
    \E[T] = O(r)\frac{1}{n}\sum_{u\in V}\E[T(u)],
    \]

    as we sampled $r$ vertices, each vertex $u\in V$ is sampled with probability $\frac{1}{n}$.
    
    Note that by our truncation, there are at most $\Gamma$ vertices visited in BFS, 
    and each of them has degree less than $d^{(G)}=O(d\cdot\Gamma)$. 
    Therefore, at most $O(\Gamma^2 d)$ edges are visited in BFS, the number of coin flips is at most
    $O(\log(\Gamma^2 d))=O(\log(\Gamma d))$.

    If a sampled vertex $u$ flips coins for $t$ times, the running time of BFS is 
    $(d_u+2d_u+\dots+2^{t-1}\cdot d_u)\leq 2^t d_u$. The coin is flipped at most $O(\log(\Gamma d))$ times,
    so the expected time of BFS is
    \[
    \E[T(u)]\leq\sum_{t=1}^{O(\log(\Gamma d))}2^{-t}\cdot 2^t d_u=O(d_u\log(\Gamma d))
    \]

    So the expected running time of \cref{alg:appncc} is 
    $O(r\cdot\frac{1}{n}\sum_{u\in V}[d_u\log(\Gamma d)])=O(r\cdot d\log(\Gamma d))=O(\frac{k}{\varepsilon^2}d\log(\frac{k}{\varepsilon}d))$.
\end{proof}

Note that \cite{chazelle2005approximating} achieves an additive error of $\varepsilon\cdot n$,
with a running time that is quadratic in $\varepsilon$. 
Utilizing their approach to attain an additive error of $\varepsilon\frac{n}{k}$ would require a
running time that is quadratic in $k$.
In contrast, our method achieves this with a running time that scales linearly with $k$ in the regime when $n/k \ge c$.
By amplifying\cref{alg:appncc} we obtain \cref{alg:median-trick} that satisfies the following corollary.

\begin{algorithm}[h!]
    \DontPrintSemicolon
    \caption{\textsc{AppNCCMedianTrick}($G, H, \varepsilon,k,d, \delta$)}
    \label{alg:median-trick}
    Let $C>0$ be a sufficiently large constant\;
    \For{$i=1,\dots,C\cdot \log (1/\delta)$
}
    {
        invoke $i$-th independent instance of \cref{alg:appncc} and get output $\mathrm{out}_{i}$\;
    }
    \Return $\hat{c}$, which is the median of $\mathrm{out}_{1},\dots,\mathrm{out}_{k}$\;
\end{algorithm}

\begin{corollary}\label{cor:numberccappdelta}
Let $\varepsilon,\delta\in(0,1)$ and $k\ge 1$.
Suppose (an upper bound on) the average degree 
$d$ of graph $G$ is known.
Given access to the graph $G$ and an implicit subgraph $H\subseteq G$ (where $H$ shares 
the same vertex set as $G$, and the edges of $H$ are determined by evaluating conditions 
on the edges of $G$), 
\cref{alg:median-trick}  outputs $\hat{c}$ satisfying
\[
|\hat{c}-c| \leq \varepsilon\cdot\max\left\{\frac{n}{k}, c\right\},
\]    
with probability at least $1-\delta$.
Here, $c$ is the number of connected components in the subgraph $H$. 
The algorithm runs in time $O(\frac{k\log(1/\delta)}{\varepsilon^2}d\log(\frac{k}{\varepsilon}d))$ in expectation.
\end{corollary}

\section{Cost Formula for Distance-Based Clustering
}\label{sec:distancemeasure}

Recall that the cost of a hierarchical SLC is
$\cost(G)=\sum_{i=1}^{n-1}(n-i)\cdot w_i $, where $w_i$ is the $i$-th smallest edge weight in the MST. 
Similarly to the work on approximating the cost of an MST \citep{chazelle2005approximating}, we will express $\cost(G)$ as a function of the number of connected components in certain subgraphs of $G$. 
For this purpose, let $G_j$ denote the subgraph induced by edges of weight at most $j$ and let $c_j$ denote the number of connected components in $G_j$. In particular, we let 
$G_0$ denote the graph consisting of $n$ singleton vertices, so $c_0=n$. We assume the graph $G$ is connected, so that the top cluster (corresponding to $1$-SLC) contains all vertices.

We assume that the edge weights of $G$ are integers. If this is not the case, one may rescale the weights by a factor of $\approx 1/\varepsilon$ and
round them to their closest integer. We refer to \Cref{sec:appendix-rounding-weights} for details.

\begin{theorem}\label{thm:formula_cost}
Let $G$ be a connected graph on $n$ vertices, with edge weights from 
$\{1,\dots,W\}$. Then
\[
\cost(G)=\frac{n(n-1)}{2}+\frac12\cdot \sum_{j=1}^{W-1}(c_j^2-c_j).
\]
\end{theorem}
\begin{proof}
Let $n_j$ be the number of edges in the MST with weight $j$. 
Observe that the number of edges in the MST with
weight $1$ is $n_1=n-c_1=c_0-c_1$ and the number of edges with weight $2$ is $n_2=c_1-c_2$, $\dots$, the number of edges with weight $j$ is 
$n_j=c_{j-1}-c_j$ for any integer $j\in \{1,\dots,W\}$. By our assumption that the graph $G$ is connected, $c_W=1$. 
Then we have,
\begin{align*}
\cost(G)&=\sum_{i=1}^{n-1}(n-i)\cdot w_i 
\tag{by \cref{eqn:costG_MST}}\\
& =\sum_{i=1}^{n_1}(n-i)\cdot 1+\sum_{i=n_1+1}^{n_1+n_2}(n-i)\cdot 2+...+\sum_{i=n_1+...+n_{W-1}+1}^{n_1+...+n_W}(n-i)\cdot W
\tag{reorganize the sum by grouping the contributions by edge weights}\\
&=\sum_{i=1}^{n-c_1}(n-i)\cdot 1 + \sum_{i=n-c_1+1}^{n-c_2}(n-i)\cdot 2 + \dots + \sum_{i=n-c_{W-1}+1}^{n-c_W}
(n-i)\cdot {W}
\tag{since $n_j=c_{j-1}-c_j$}\\
&=\sum_{i=1}^{n-c_W}(n-i) + \sum_{i=n-c_1+1}^{n-c_2}(n-i)\cdot 1 + \dots + \sum_{i=n-c_{W-1}+1}^{n-c_W}
(n-i)\cdot (W-1)
\tag{factoring out $(n-i)$ from each summation and combining the terms}\\
&=\sum_{i=1}^{n-c_W}(n-i) + \sum_{i=n-c_1+1}^{n-c_W}(n-i) + \dots + \sum_{i=n-c_{W-1}+1}^{n-c_W}(n-i)
\tag{repeating this decomposition until all summations end with $i=n-c_W$}\\
&=\sum_{i=1}^{n-1}i + \sum_{i=1}^{c_1-1}i + \dots + \sum_{i=1}^{c_{W-1}-1}i 
\tag{substituting $(n-i)$ with $i$, and noting that $c_W=1$}\\
&=\sum_{j=0}^{W-1}\frac{(1+c_j-1)\cdot(c_j-1)}{2}
\tag{noting that $c_0=n$}\\
&=\frac{n(n-1)}{2}+\frac12\cdot \sum_{j=1}^{W-1}(c_j^2-c_j)
\end{align*}
\end{proof}

\section{Problem Abstraction and Binary Search}\label{sec:binarysearch}

As mentioned earlier, our goal is to efficiently approximate the total clustering cost $\cost(G)$, which reduces to estimating the sum $\sum_{j=1}^{W-1}(c_j^2 - c_j)$. A key observation is that the sequence $(c_1, \dots, c_W)$ is non-increasing: as $j$ increases, the subgraph $G_j$ includes more edges, leading to fewer connected components. This monotonicity suggests the possibility of summarizing the sequence succinctly using only a few representative values. However, we only have access to estimated values $\hat{c}_j$, and estimation errors may disrupt the monotonicity. This motivates us to study an abstract version of the problem: how to efficiently approximate the sum over a non-increasing sequence when only noisy estimates are available. This abstraction also arises in our algorithm for similarity-based clustering, and we believe it could be of independent interest in other sublinear or local computation settings.

Now, let us describe this abstract problem in a bit more formal manner: Given a noisy version of a \emph{non-increasing} sequence $X$ with 
elements $(x_1,x_2,\dots,x_W)$, and each element $x_j$ from $[L,R]$ (where $L$ and $R$ are lower and upper bounds for the smallest and largest element of $X$),
we want to find a 
succinct approximation of every $x_j$
without approximating each $x_j$ separately.

Initially, assume the exact values of $X$ are known.
We divide the range $[L,R]$ into $t-1$ intervals, for some integer $t$ to be chosen later. The intervals are specified by some numbers $L=B_{t}<B_{t-1}<\dots<B_{1}=R$.
Specifically, we divide $[L,R]$ into: 
\[
[B_{t},B_{t-1}],(B_{t-1},B_{t-2}],\dots,(B_{2},B_1]
\]

Then to get our estimates we do the following. For each $i \in \{1,\dots, t\}$ we perform binary search on the sequence $X' = (x_1,\dots, x_W, {L})$
to find the smallest index $j_i\in \{1,\dots,W+1\}$,
$1\le i \le t-1$
such that $x_{j_i}\leq B_i$
and where we define $j_t=W+1$. 
This results in a non-decreasing sequence
$(j_1,j_2,\dots, j_t)$ that
partitions $\{1,\dots,W\}$ sets
into $J_i = \{j_{i},\dots, j_{i+1}-1\}, 1\le i \le t-1$ 
(where the set is empty when $j_i = j_{i+1})$. 
Then every $x_j$ with $j \in J_i$ is estimated by $B_i$.

Now consider that we only have access to estimates $\hat{x}_j \in[L,R]$ of the $x_j$ in $X$. 
In particular, we assume that there is a sequence $T=(T_1,\dots, T_W)$ of error bounds $T_j$, 
such that $|\hat{x}_j-x_j|\leq T_j$. 
To find succinct approximation of all entries in the sequence $X$,
we can still apply the approach sketched above and perform a variant of the standard binary search on the sequence
$\widehat{X}'=(\hat{x}_1,\dots,\hat{x}_W,{L})$ corresponding to the estimates, searching for each interval endpoint $B_i$.
For completeness, we give the pseudocode in
\cref{alg:binarysearch}. The algorithm is initialized with $\ell=1$, $r=W+1$,  $B=B_i$, and  sequence $\hat X'$.
\begin{algorithm}[h!]
    \DontPrintSemicolon
    \caption{\textsc{BinarySearch}($\widehat{X}',\ell,r,B$)}
    \label{alg:binarysearch}
    \SetKwInOut{Input}{input}
    \SetKwInOut{Output}{output}
    \Input{Estimated array $\widehat{X}'$, leftmost and rightmost
    indices $\ell$ and $r$, search key $B$}
    \Output{Index $i$}
    \BlankLine
    \uIf{$\ell=r$}{\Return $\ell$}
    \Else{
    $m=\left\lfloor\frac{\ell+r}{2}\right\rfloor$\;
    \uIf {\label{alg:binary_m}$\hat x_{m} \le B$}{\Return\textsc{BinarySearch}($\widehat{X}',\ell,m,B$)} 
        \Else{
        \Return \textsc{BinarySearch}($\widehat{X}',m+1,r,B$)}
    }
\end{algorithm}
We note that, since we only have access to the estimates $\hat x_j$, the sequence $\widehat{X}'=(\hat x_1,\dots,\hat x_W,{L})$ may not be non-increasing. Furthermore, we note that \cref{alg:binarysearch} always
returns the index that is reached at the end of the recursion when
$\ell$ and $r$ become equal. 
 We show that our binary search approach satisfies a certain monotonicity property, namely, that the indices $\hat j_i$ found by the searches for the interval endpoints $B_i$ 
are in non-decreasing
order. This is stated formally in the following lemma.

\begin{lemma}\label{lem:foundindices}
Let $a,b \in [L,R]$ with $a<b$. Let $\hat j_a$ and $\hat j_b$ be the
two indices returned by 
\cref{alg:binarysearch} on a (possibly unsorted) input sequence $\widehat X'$ with $B=a$ and $B=b$, respectively. Then we have $ \hat j_a \ge  \hat j_b$.
\end{lemma}
\begin{proof}
As a thought experiment for the proof,  perform a binary search for $B=a$ and $B=b$ in parallel.
During binary search, the only place where we are using $B$ is \cref{alg:binary_m} where we compare to the value $\hat x_m$.
If $\hat{x}_m\leq a < b$ or $\hat{x}_m>b>a$, the comparison $\hat x_m \le B$ in \cref{alg:binary_m} has the same outcome for both search keys $a$ and $b$.
The two searches behave differently only if $a<\hat{x}_m\leq b$.
For $a$, the algorithm continues to search in the range $[m+1,r]$ which ensures that the returned index 
$\hat j_a$ is at least $m+1$.
For $b$, setting the range to be $[\ell,m]$ ensures the returned 
index $\hat j_b$ is at most $m$. This guarantees that regardless of the cases encountered during the binary search,
we always end up with $\hat j_a\geq \hat j_b$.
\end{proof}

We can apply the above lemma directly on the $B_i$ to obtain the following corollary.

\begin{corollary}\label{cor:index-order}
    For $i\in \{1,\dots,t-1\}$ let $\hat j_i$ be the index returned when \cref{alg:binarysearch} is run on a (possibly unsorted) input sequence $\widehat X'$ with $B=B_i$
    and let $\hat j_t=W+1$.
 Then $\hat j_1\le \hat j_2\le \dots \le \hat j_t$.
\end{corollary}

Given the above corollary, 
our plan is to simply search for the interval endpoints $B_i$ and denote the corresponding resulting indices $\hat j_i$ and we always set $\hat j_t=W+1$. If we then invoke \cref{alg:binarysearch} 
with parameter $B= \hat x_j$  that satisfies $B_{i+1} \le \hat x_j \le B_i$
and let $i^*$ be the index returned by 
our binary search procedure, then we know by \Cref{lem:foundindices} that $\hat j_i \le i^* \le \hat j_{i+1}$.
Thus, it seems that we can simply estimate $\hat x_j$ by $B_i$ and the error resulting from the binary search is $T_j$ plus $B_{i}-B_{i+1}$. 
However, we may encounter a scenario where several consecutive $\hat j_i$ values are 
mapped to the same value, in which case it is not immediately clear that this approach works. 
We can define  $\hat J_i = \{\hat j_{i},\dots, \hat j_{i+1}-1\}, 1\le i \le t-1$,  as in the case of binary search with error.
This ensures that the $\hat J_i$ 
form a partition of $\{1,\dots,W\}$, but in the case that many $\hat j_i$ are mapped to the same value we could still potentially get a large error. In the following we show that this will not be the case.
We will now define the notion of a valid discretization with respect to (w.r.t.) $X$ and $T$ that is a partition of $[L,R]$ into intervals, such that for every $x_i$,
the error of the interval containing $x_i$ is at most the minimum length of the neighboring intervals. 

\begin{definition}
Let $X=(x_1,\dots,x_W)$ be a non-increasing sequence with $x_j\in [L,R]$ for all $j\in \{1,\dots,W\}$
and let $T=(T_1,\dots, T_W)$ be a sequence of error bounds.
Let $\widehat X =(\hat x_1,\dots, \hat x_W)$, $\hat x_j\in [L,R]$
with $|\hat x_j - x_j| \le T_j$.
We say that a sequence of interval endpoints
$L=B_t < B_{t-1}<\dots < B_1=R$ is a \emph{valid discretization} of $[L,R]$ w.r.t. $X$ and $T$, 
if for every $j\in \{1,\dots,W\}$ and $i \in\{1,\dots, t-1\}$ such that $B_{i+1}\le x_{j} 
\le B_{i}$ we have 
that 
\begin{itemize}
\item
$T_j<B_{i-1}-B_{i}$, if $i\ge 2$, and
\item 
$
T_j< B_{i+1} - B_{i+2}
$, if $i \le t-2$.
\end{itemize}
\end{definition}

Intuitively, this ensures that whenever we search for a key $\hat x_j$ with $B_{i+1} \le x_j \le B_i$
we will end up in the set of indices $\hat J_i$ or a neighboring one.

For every $j\in \hat J_i$ we now define the estimator for $x_j$ to be $\bar x_j = B_i$.
In the following lemma, we formalize our intuition and
show that the value of $x_j$ is between $B_{i+2}$ and $B_{i-1}$ and thus 
$\bar x_j$ is close to $x_j$ if our discretization is sufficiently fine. 

\begin{lemma}\label{lem:bound_of_xj}
Let $X=(x_1,\dots,x_W)$ be a non-increasing sequence with $x_j\in [L,R]$ for all $j\in \{1,\dots,W\}$
and let $T=(T_1,\dots, T_W)$ be a sequence of error bounds.
Let $\widehat X =(\hat x_1,\dots, \hat x_W)$, where $\hat x_j\in [1,n]$
with $|\hat x_j - x_j| \le T_j$ for any $j\in[W]$.
Let $L= B_t< \dots < B_1 =R$ be a sequence of interval endpoints that is a valid discretization of $[L,R]$ w.r.t. $X$ and $T$.
Let $\hat j_i$ be the index returned by \cref{alg:binarysearch} on $\widehat X'=(\hat x_1,\dots, \hat x_W,{L})$ for search key $B=B_i$ and let $\hat J_i = \{\hat j_{i},\dots, 
\hat j_{i+1}-1\}$ and define $\hat j_t = W+1$.
For any $1\le i\le t-1$ and every $j\in \hat J_i$ 
we have 
$$
B_{i+2} \le x_{j} \le B_{i-1},
$$
where we define $B_0 = \infty$ and $B_{t+1} = - \infty$.
\end{lemma}

\begin{proof}
Let $i \in \{1,\dots,t-1\}$. We will show that if $x_j > B_{i-1}$
then $x_j \notin \hat J_i$ and if $x_j < B_{i+2}$ then $x_j\notin \hat J_i$.
This implies the theorem.
Now
define $j^*$ to be the largest index such that $x_{j^*} > B_{i-1}$ (if no such $x_{j^*}$ exists there is nothing to prove).
Since the $B_i$ form a valid discretization, we have for every $j\le j^*$ with $B_{i'+1} \le x_j \le B_{i'}$ that
$$
\hat x_{j} \ge x_{j}-T_{j} \ge B_{i'+1} - T_{j} > B_{i'+1} - (B_{i'+1}-B_{i'+2}) = B_{i'+2} \ge B_{i}
$$
since $x_j\ge x_{j^*}>B_{i-1}$ and thus $i'+1 \le i-1$.
Now consider an invocation of 
\cref{alg:binarysearch} with $B= B_{i}$. Whenever we compare an element $\hat x_m$, $m \le j^*$,
with $B=B_i$ we have $\hat x_m > B$ and we recurse with the sequence $\hat x_{m+1},\dots, \hat x_r$. Since we have $\hat x_{W+1} = {L \le} B_i$ this implies that we always have $\hat j_i> j^*$ and so $j^* \notin \hat J_i$.

Now consider the smallest index $j^*$ such that $x_{j^*} < B_{i+2}$ (again, if such an $x_{j^*}$ does not exist, there is nothing to prove). Since the $B_i$'s form a valid discretization, we have for every $j\ge j^*$ 
with $B_{i'+1} \le x_j \le B_{i'}$
that 
$$
    \hat x_{j} \le x_{j}+T_{j} \le B_{i'}+ T_{j} < B_{i'} + (B_{i'-1}-B_{i'})= B_{i'-1}\le B_{i+1}
$$
since $x_j\le x_{j^*}<B_{i+2}$ and thus $i'\ge i+2$.
Now consider an invocation of \cref{alg:binarysearch} with $B= B_{i+1}$. Whenever we compare an element $\hat x_m$, $m \ge j^*$,
with $B=B_{i+1}$ we have $\hat x_m \le B$ and we recurse with the sequence $\hat x_{\ell},\dots, \hat x_m$. This implies that $j_{i+1} \le j^*$. Combining both cases, we observe that there is no $j\in \hat J_i$ with $x_j>B_{i-1}$ or $x_j < B_{i+2}$. Hence the lemma follows.
\end{proof}

\section{Sublinear Algorithms in Distance Case}\label{sec:algorithmfordistance}

We now give a sublinear time algorithm to estimate the SLC cost $\cost(G)$ in the distance graph, 
with a $(1+\varepsilon)$-approximation factor.
The straightforward approach is to estimate each $c_j$ using \cref{alg:appncc}, 
and it can be easily verified that this results in a $(1+\varepsilon)$-approximation. 
However, this method results in a running time that is linear in $W$. 
To improve on this, 
we leverage the non-increasing nature of $c_j$ with the binary search technique given in \cref{sec:binarysearch} to achieve our goal with a more efficient running time. Once we obtain the algorithm for the SLC cost, we then build upon it to give a sublinear time algorithm for estimating the  vector $(\cost_1,\dots,\cost_n)$, which we define as the \emph{SLC profile vector in the distance graph}. 

\subsection{From Binary Search to Clustering Cost Estimation}

We will apply binary search to the sequence of estimates $ \hat{c}_j $ of the number of connected components, where $ \hat{c}_j $ and its error bound are defined in the following lemma. 

\begin{lemma}\label{lem:hatcj}
For any $1\leq j\leq W$, let $\hat c_j=\min\{\max\{\hat{c}_j',1\},n\}$, where $\hat{c}_j'$ is the output of \cref{alg:median-trick} with input $G$, $H=G_j$, $\varepsilon/8$, 
$k=\sqrt{W}$, $d$, $\delta=1/(4W)$. 
Then with probability at least $3/4$, it holds that for \emph{all} $1\leq j\leq W$, $\hat{c}_j\in[1,n]$, and
\[
\abs{\hat{c}_j-c_j}\leq T_j:=\frac{\varepsilon}{8}\cdot\max\left\{\frac{n}{\sqrt{W}},c_j\right\}.
\]
\end{lemma}
\begin{proof}

By \cref{cor:numberccappdelta}, for any fixed $j\in[W]$, it holds that with probability at least
$1-1/(4W)$, 
$\abs{\hat{c}_j'-c_j}\leq \frac{\varepsilon}{8}\cdot\max\left\{\frac{n}{\sqrt{W}},c_j\right\}=T_j$. 
Besides, rounding $\hat{c}_j'$ to $\hat{c}_j=\min\{\max\{\hat{c}_j',1\},n\}$ guarantees that $\hat{c}_j\in [1,n]$. Furthermore,
if $\hat{c}_j<1$ or $\hat{c}_j>n$, this rounding process decreases the gap between $\hat{c}_j$ and $c_j$, and ensures the error remains within the error bound $T_j$.

Therefore, by the union bound, with probability at least $\frac{3}{4}$, 
for \emph{all} $j\in [W]$, $\abs{\hat{c}_j-c_j}\leq T_j$, and $\hat{c}_j\in[1,n]$.
\end{proof}

\emph{Throughout the following, we will assume that for \emph{all} $j$, the inequality $\abs{\hat c_j-c_j}\leq T_j$ holds, and $\hat{c}_j\in[1,n]$}. According to \cref{lem:hatcj}, this occurs with probability at least $3/4$. 

Now we define the endpoints of intervals which partition $[1,n]$.

\begin{definition}\label{def:interval_distance}
    Let $0<\varepsilon<1$, $t_1$ be the largest integer such that
    $\frac{n}{(1+\varepsilon)^{t_1-1}}\ge\frac{n}{\sqrt{W}}$,
    and $t_2$ be the largest integer such that
    $\frac{n}{\sqrt{W}}(1-\varepsilon\cdot t_2)\ge\frac{\varepsilon n}{\sqrt{W}}$. 
    Note that $t_1=\lfloor\log_{1+\varepsilon}\sqrt{W}+1\rfloor$ and $t_2=\lfloor\frac{1}{\varepsilon}-1\rfloor$.
    Define $B_i$ such that
\[
B_i=
\begin{cases}
    \frac{n}{(1+\varepsilon)^{i-1}} & \quad \text{if $1\le i\le t_1$}\\
    \frac{n}{\sqrt{W}}(1-\varepsilon\cdot (i-t_1)) & \quad \text{if $t_1<i\le t_1+t_2$}\\     
    1 & \quad \text{if $i=t:=t_1+t_2+1$}
\end{cases}
\]
\end{definition}

According to the above definition, we have the following fact.
\begin{fact}\label{fact:distanceinterval} 
It holds that
\begin{enumerate}
\item $t=t_1+t_2+1=O(\log_{1+\varepsilon}\sqrt{W}+1/\varepsilon)=O(\log W/\varepsilon)$;
\item $\frac{n}{(1+\varepsilon)^{t_1}}<\frac{n}{\sqrt{W}}$; and thus, $B_{t_1}=\frac{n}{(1+\varepsilon)^{t_1-1}}=(1+\varepsilon)\frac{n}{(1+\varepsilon)^{t_1}}<(1+\varepsilon)\frac{n}{\sqrt{W}}$;
\item $B_{t_1}-B_{t_1+1}\ge\frac{n}{\sqrt{W}}-(1-\varepsilon)\frac{n}{\sqrt{W}}=\frac{\varepsilon n}{\sqrt{W}}$, and $B_{t_1}-B_{t_1+1}<(1+\varepsilon)\frac{n}{\sqrt{W}}-(1-\varepsilon)\frac{n}{\sqrt{W}}=2\frac{\varepsilon n}{\sqrt{W}}$;
\item when $i\le t_1$, $B_{i-1}-B_i=(1+\varepsilon)B_i-B_i=\varepsilon B_i$;
and when $i>t_1$, $B_{i-1}-B_i\ge\frac{\varepsilon n}{\sqrt{W}}$.
\item as for \emph{all} $j\in[W]$, $\hat{c}_j\le n$, when we invoke \textsc{BinarySearch($\widehat{C},1,W,B_1$)}, we will never access $\hat{c}_m>B_1$.
Thus, $\ell$ remains to be $1$, and when $\ell=r$, the returned index $j_1=1$.
\end{enumerate}
\end{fact}

We then prove that the endpoints $B_i$'s defined in this way form a valid discretization of $[1,n]$ w.r.t. $C$ and $T=(T_1,\dots,T_W)$.

\begin{lemma}\label{lem:discretization_distance}
    Assume that the event stated in \cref{lem:hatcj} holds. The sequence of endpoints $(B_1,\dots,B_t)$ defined in \cref{def:interval_distance} is
    a valid discretization of $[1,n]$ w.r.t. $C$ and error bound 
    $T_j=\frac{\varepsilon}{8}\cdot\max\left\{\frac{n}{\sqrt{W}},c_j\right\}$ on each $c_j,j\in[W]$.
\end{lemma}
\begin{proof}
We define $B_0=\infty$ and $B_{t+1}=-\infty$. 
We consider the sequence $(\hat{c}_1, \dots, \hat{c}_W)$ as the approximation of $C=(c_1,\dots,c_W)$. According to the previous assumption, it holds that for all $j$, $\abs{\hat c_j - c_j}\leq T_j$, and $\hat{c}_j\in [1,n]$.

Thus, to show that $(B_1,\dots,B_t)$ is a valid discretization of $[1,n]$ w.r.t. $C$ and $T=(T_1,\dots,T_W)$, it suffices to prove that for any $1\leq j\leq W, 1\leq i\leq t-1$ such that $B_{i+1}\leq c_j\leq B_i$, it holds that 
\[
T_j<\min\{B_{i-1}-B_i, B_{i+1}-B_{i+2}\}.
\]

Now consider any fixed $ j \in [W] $ and the corresponding interval $ [B_{i+1}, B_i] $ that contains $ c_j $ for some $ i \in [t-1] $. We analyze the following cases by using properties given in \cref{fact:distanceinterval}.  

Case (I): $i\le t_1-2$. 
In this case, we have $B_{i-1}-B_{i}=\varepsilon B_i$, for any $2\leq i\leq t_1-2$ and $B_{i-1}-B_i=\infty$ if $i=1$. Furthermore,  $B_{i+1}-B_{i+2}=\varepsilon B_{i+2}=\frac{\varepsilon}{(1+\varepsilon)^2}B_i$. Thus,
\[
\min\{B_{i-1}-B_i,B_{i+1}-B_{i+2}\}=\frac{\varepsilon}{(1+\varepsilon)^2}B_i > \frac{\varepsilon}{8}B_i,
\]
where the last inequality follows from the fact that $(1+\varepsilon)^2\leq 4$. 

Now note that in this case, it holds that $B_{i}\geq c_j\ge B_{i+1}>\frac{n}{\sqrt{W}}$. Therefore, $T_j=\frac{\varepsilon}{8}\cdot \max\{\frac{n}{\sqrt{W}},c_j\}=\frac{\varepsilon}{8}\cdot c_j\leq \frac{\varepsilon}{8}\cdot B_i<\min\{B_{i-1}-B_i,B_{i+1}-B_{i+2}\}$. 

Case (II): $i=t_1-1$. Here, we have that $B_{i-1}-B_i=B_{t_1-2}-B_{t_1-1}=\varepsilon B_{t_1-1}$,
    and $B_{i+1}-B_{i+2}=B_{t_1}-B_{t_1+1}\ge\frac{\varepsilon n}{\sqrt{W}}>\frac{\varepsilon}{1+\varepsilon}B_{t_1} = \frac{\varepsilon}{(1+\varepsilon)^2}B_{t_1-1}$. 
    Furthermore, $B_i=B_{t_1-1}\geq c_j\ge B_{i+1}=B_{t_1}\ge\frac{n}{\sqrt{W}}$, and thus
    \[
    T_j=\frac{\varepsilon}{8}\cdot \max\{\frac{n}{\sqrt{W}},c_j\}=\frac{\varepsilon}{8}c_j \leq \frac{\varepsilon}{8}B_{t_1-1}
    <\frac{\varepsilon}{(1+\varepsilon)^2}B_{t_1-1} = \min\{B_{i-1}-B_i,B_{i+1}-B_{i+2}\}
    \]

Case (III): $i=t_1$. Here, we have $B_{i-1}-B_i=\varepsilon B_i=\varepsilon B_{t_1}$,
    and 
    $B_{i+1}-B_{i+2}=\frac{\varepsilon n}{\sqrt{W}}>\frac{\varepsilon}{1+\varepsilon}B_{t_1}$. Furthermore, $c_j\leq B_{t_1} <(1+\varepsilon)\frac{n}{\sqrt{W}}$. 
   Thus, 
    \[
    T_j=\frac{\varepsilon}{8}\cdot\max\left\{\frac{n}{\sqrt{W}},c_j\right\}
    \le\frac{(1+\varepsilon)\varepsilon}{8}\frac{n}{\sqrt{W}}<\frac{\varepsilon}{1+\varepsilon}B_{t_1}
    =\min\{B_{i-1}-B_i,B_{i+1}-B_{i+2}\}
    \]

Case (IV): $i>t_1$. In this case, we have $B_{i-1}-B_i\ge\frac{\varepsilon n}{\sqrt{W}}$ 
    and $B_{i+1}-B_{i+2}\ge\frac{\varepsilon n}{\sqrt{W}}$ for any $t_1<i\leq t-2$ and $B_{i+1}-B_{i+2}=\infty$ for $i=t-1$. Thus, 
\[
\min\{B_{i-1}-B_i,B_{i+1}-B_{i+2}\} \geq \frac{\varepsilon n}{\sqrt{W}} 
\]

Furthermore, $c_j\le B_i\le B_{t_1+1}<\frac{n}{\sqrt{W}}$. Thus, $T_j=\frac{\varepsilon}{8}\cdot\max\left\{\frac{n}{\sqrt{W}},c_j\right\}=\frac{\varepsilon}{8}\cdot\frac{n}{\sqrt{W}} <\min\{B_{i-1}-B_i,B_{i+1}-B_{i+2}\}$.

This completes the proof of the lemma. 
\end{proof}

Now we are ready to exploit the binary search based algorithm for succinctly approximating the sequence $C=(c_1,\dots,c_W)$ to estimate the single-linkage clustering cost $\cost(G)$. The algorithm simply performs binary search \cref{alg:binarysearch} on the estimated array 
$\widehat{C}=\{\hat{c}_1,\dots,\hat{c}_W,0\}$ with search keys $B_1,\dots,B_t$ 
as defined in \cref{def:interval_distance}, where $\hat{c}_j$'s are estimators of $c_j$'s as defined in \cref{lem:hatcj}. Then we obtain the corresponding indices $\hat j_1,\dots,\hat j_t$. For $i\in\{1,\dots,t-1\}$, let $\hat J_i=\{\hat j_i,\dots,\hat j_{i+1}-1\}$. Note that $\hat J_1,\dots, \hat J_{t-1}$ form a partition of the index set $[W]=\{1,\dots, W\}$. 
For any $j\in\hat J_i$, we define $\bar c_j=B_i$ as the estimate for $c_j$ and then use  $\bar c_j$'s to estimate the cost $\cost(G)$.

\begin{algorithm}[h!]
    \DontPrintSemicolon
    \caption{\textsc{AppCost}($G, \varepsilon, W, d$)}
    \label{alg:appcost}
    \SetKwInOut{Input}{input}
    \SetKwInOut{Output}{output}
    \Input{graph $G$, approximation parameter $\varepsilon$, 
    maximum weight $W$, average degree $d$}
    \Output{$\widehat{\cost}(G)$, which estimates $\cost(G)$}
    \BlankLine
    set $t$ according to \cref{def:interval_distance} and
    set parameters $k=\sqrt{W}$,
    $d^{(G)}=d\cdot\lceil\frac{4k}{\varepsilon}\rceil$\;
    \For{$1\leq i\leq t$}{
        set $B_i$ according to \cref{def:interval_distance}\;
        invoke \textsc{BinarySearch$(\widehat{C},1,W,B_i)$} and get output index $\hat{j}_i$, where $\widehat{C}=(\hat{c}_1,\dots,\hat{c}_W,1)$ and each $\hat{c}_j=\min\{\max\{\hat{c}_j',1\},n\}$ and $\hat{c}_j'$
        is the output of \cref{alg:median-trick} with parameters $G$, $H=G_j$, $\varepsilon/8$, 
        $k=\sqrt{W}$, $d$, $\delta=1/(4W)$
        \Comment{for consistency, reuse stored value of $\hat{c}_j$, if previously estimated}\label{alg:appcostline4}\;
    }
for each $i\in\{1,\dots,t-1\}$, and $j\in \hat J_i:=\{\hat j_i,\dots,\hat j_{i+1}-1\}$, define $\bar c_j=B_i$ 
\Comment{$\bar c_j$'s are defined for the analysis only}\label{alg:appline5}\;
let $\widehat{\cost}(G)=\frac{1}{2}\sum_{i=1}^{t-1}(\hat j_{i+1}-\hat j_i)\cdot(B_i^2-B_i)$\label{alg:appcostline6}\;
output $\widehat{\cost}(G)$ and the sequence $\{\hat j_1,\dots,\hat j_t\}$
\end{algorithm}

The pseudocode of the algorithm is Described in \cref{alg:appcost}. It is important to note that we do not need to access all the values in the estimated array 
$\widehat{C}$. 
Instead, we only need to access each $\hat{c}_m$ in the search path corresponding to the binary search process. 
Besides, to ensure consistency in $\hat C$, we reuse $\hat{c}_m$'s stored value, if it was accessed previously. Intuitively, this means that we only need to approximate the number of connected components for $ O(\poly\log W) $ subgraphs $ G_j $.
The interval endpoints are set according to \cref{def:interval_distance}, which ensures that the sequence $ (B_1, \dots, B_t) $ forms a valid discretization of $[1,n]$ w.r.t. $ C $ and an appropriately chosen error bound as described in \cref{lem:discretization_distance}. Furthermore, we note that given the values $\hat{j}_i$'s and $B_i$'s , \cref{alg:appline5,alg:appcostline6} can be implemented in $ O(t) $ time, as we only need to use the right-hand side (RHS) of the estimator $ \widehat{\cost}(G) $ to estimate $ \cost(G) $, and the definitions for $ \bar{c}_j $'s are primarily used for the analysis.

We first show the following guarantee of the estimate $\bar c_j$, for $1\leq j\leq W$.

\begin{lemma}\label{lem:bound_of_bar_cj}
Assume that for \emph{all} $j\in [W]$, the inequality $\abs{\hat c_j-c_j}\leq T_j$ holds. Then  for \emph{all} $j\in[W]$, 
    \[
    |\bar c_j-c_j|\le 4\varepsilon\cdot\max\left\{\frac{n}{\sqrt{W}},c_j\right\}.
    \]
\end{lemma}
\begin{proof}
    By \cref{lem:discretization_distance}, 
    the preconditions of \cref{lem:bound_of_xj} are satisfied. 
   By \cref{lem:bound_of_xj}, we have $B_{i+2}\le c_j\le B_{i-1}$. By the definition of $\bar c_j$, we have that 
    \begin{itemize}
        \item When $i=1$, we have $\abs{\bar c_j-c_j}\le B_i-B_{i+2}
        =((1+\varepsilon)^2-1)B_{i+2}\le3\varepsilon B_{i+2}\le3\varepsilon c_j$.
        
        \item When $1<i\le t_1-2$, we have 
            if  $c_j>\bar c_j$,
        \text{then } $0<c_j-\bar c_j\le B_{i-1}-B_i=\varepsilon B_i\le\varepsilon c_j$; and if  $c_j\leq \bar c_j$, then $0\leq \bar c_j-c_j\le B_i-B_{i+2}
            =((1+\varepsilon)^2-1)B_{i+2}\leq 3\varepsilon\cdot c_j$.
                
        \item When $i=t_1-1$, we have 
        if $c_j>\bar c_j$, {then } $0<c_j-\bar c_j\le B_{i-1}-B_i=\varepsilon B_i\le\varepsilon c_j$; and if  $c_j\leq \bar c_j$, {then }  $0\leq \bar c_j-c_j\le B_{t_1-1}-B_{t_1+1} <((1+\varepsilon)^2-(1-\varepsilon))\frac{n}{\sqrt{W}} \le 4\varepsilon\frac{n}{\sqrt{W}}$, in which we make use of the fact that $B_{t_1-1}<(1+\varepsilon)^2\frac{n}{\sqrt{W}}$, which in turns follows from the fact that $B_{t_1}<(1+\varepsilon)\frac{n}{\sqrt{W}}$.
        
        \item When $i=t_1$, we have that 
       {if } $c_j>\bar c_j$, then
     $0\leq c_j-\bar c_j\le B_{i-1}-B_i=\varepsilon B_i\le\varepsilon c_j$; and if $c_j\leq \bar c_j$, {then } $0\leq \bar c_j-c_j\le B_{t_1}-B_{t_1+2}
    <((1+\varepsilon)-(1-2\varepsilon))\frac{n}{\sqrt{W}} = 3\varepsilon\frac{n}{\sqrt{W}}.$
        
        \item When $i=t_1+1$, we have that 
 \text{if } $c_j>\bar c_j$, then $0<c_j-\bar c_j\le B_{t_1}-B_{t_1+1}
<((1+\varepsilon)-(1-\varepsilon))\frac{n}{\sqrt{W}}
=2\varepsilon\frac{n}{\sqrt{W}}$; and if $c_j\leq \bar c_j$, {then } $0\leq \bar c_j-c_j\le B_{t_1+1}-B_{t_1+3}
            ((1-\varepsilon)-(1-3\varepsilon))\frac{n}{\sqrt{W}} = 2\varepsilon\frac{n}{\sqrt{W}}$.

        \item When $t_1+2\le i\le t-2$, $B_{i-1}-B_i = \frac{\varepsilon n}{\sqrt{W}}$,
        and $B_{i}-B_{i+2}=B_i-B_{i+1}+B_{i+1}-B_{i+2} = 2\frac{\varepsilon n}{\sqrt{W}}$,
        and so
        \[
        |\bar c_j-c_j|\leq\max\{B_{i-1}-B_i,B_i-B_{i+2}\}\leq2\frac{\varepsilon n}{\sqrt{W}}
        \]

        \item When $i=t-1$, $B_{i-1}-B_i=\frac{\varepsilon n}{\sqrt{W}}$ and 
        $B_i-B_{i+1}\le2\frac{\varepsilon n}{\sqrt{W}}$.
        Since $B_{t}\le c_j\le B_{t-2}$, we have $\abs{\bar c_j - c_j}\le \max\{B_{t-2}-B_{t-1}, B_{t-1}-B_{t}\}\le2\frac{\varepsilon n}{\sqrt{W}}$.

        \item When $i=t$, $\abs{\bar c_j-c_j}\le B_{t-1}-B_t=\frac{\varepsilon n}{\sqrt{W}}$.
    \end{itemize}
    
 Therefore, for \emph{all} $1\le j\le W$, we have 
    $\abs{\bar c_j-c_j}\le 4\varepsilon\cdot\max\left\{\frac{n}{\sqrt{W}},c_j\right\}$
\end{proof}

Now we analyze the performance of \cref{alg:appcost} and prove \cref{thm:cost_distance}.
\distance*

\begin{proof}[Proof of \cref{thm:cost_distance}]
    By \cref{lem:bound_of_bar_cj}, with probability at least $\frac{3}{4}$, we have 
    $\abs{\bar c_j-c_j}\le 4\varepsilon\cdot\max\left\{\frac{n}{\sqrt{W}},c_j\right\}$ for all $j\in[W]$.
    Then, $\abs{\bar c_j^2 - c_j^2} = \abs{\bar c_j - c_j}\cdot\abs{\bar c_j + c_j}
    \leq\abs{\bar c_j - c_j}\cdot(2c_j+|\bar c_j - c_j|)
    \le 4\varepsilon\cdot\max\left\{\frac{n}{\sqrt{W}},c_j\right\} \cdot 6\max\left\{\frac{n}{\sqrt{W}},c_j\right\}
    \leq 24\varepsilon\left(\frac{n^2}{W}+c_j^2\right)$.
    
    As $\widehat{\cost}(G)=\frac{1}{2}\sum_{i=1}^{t-1}(\hat j_{i+1}-\hat j_i)\cdot(B_i^2-B_i)$, and by the definition of $\bar c_j$, this is equivalent to $\widehat{\cost}(G)=\frac{n(n-1)}{2}+\frac{1}{2}\sum_{j=1}^{W-1}(\bar c_j^2-\bar c_j)$.
    We have that 
    \begin{align*}
        \abs{\cost(G)-\widehat{\cost}(G)}&=\frac{1}{2}
        \tag{by definitions of $\cost(G)$ and $\widehat{\cost}(G)$}
        \abs{\sum_{j=1}^{W-1}[(c_j^2-c_j)-(\bar c_j^2-\bar c_j)]}\\
        &\leq\frac{1}{2}\sum_{j=1}^{W-1}(\abs{c_j^2-\bar c_j^2}+\abs{\bar c_j-c_j})\\
        &\leq\frac{1}{2}\sum_{j=1}^{W-1} (24\varepsilon\frac{n^2}{W}+
        24\varepsilon c_j^2+4\varepsilon\frac{n}{\sqrt{W}}+4\varepsilon c_j)
        \tag{applying the error bound for $\bar c_j$}\\
        &\le\frac{24}{2}\varepsilon n^2+\frac{24}{2}\varepsilon\sum_{j=1}^{W-1} (c_j^2-c_j)
        +2\varepsilon n\sqrt{W}+14\varepsilon\sum_{j=1}^{W-1} c_j
    \end{align*}
    
Since $\sqrt{W}\leq n$, it holds that $W\leq n^2$.
By \cite{chazelle2005approximating}, we also know that the weight of minimum spanning tree is
    $\cost(\MST)=n-W+\sum_{j=1}^{W-1}c_j$, and $\cost(\MST)=\sum_{i=1}^{n-1}w_i=\cost_1\leq\cost(G)$.
    Thus, $\sum_{j=1}^{W-1}c_j=\cost(\MST)-n+W\leq\cost(G)-n+W\leq\cost(G)+n^2$,
    and we have
    \[
    \abs{\cost(G)-\widehat{\cost}(G)}\leq\frac{24}{2}\varepsilon n^2+
    \frac{24}{2}\varepsilon\sum_{j=1}^{W-1}(c_j^2-c_j)+2\varepsilon^2n^2
    +14\varepsilon(\cost(G)+n^2)\leq 102\varepsilon\cost(G)
    \]
    
    The last inequality follows as $\cost(G)=\frac{n(n-1)}{2}+\frac{1}{2}\sum_{j=1}^{W-1}(c_j^2-c_j)\geq \frac{n(n-1)}{2}$.
    When $n\geq 2$, $\cost(G)\geq\frac{1}{4}n^2$. Replacing $\varepsilon$ with $\varepsilon/102$
    achieves a $(1+\varepsilon)$ approximation factor.
    
    \paragraph{Running time analysis.}
 Note that \cref{alg:appcost} invokes \textsc{BinarySearch} for 
    $t=O(\log W/\varepsilon)$ search keys, and each invocation of \textsc{BinarySearch}
    takes $O(\log W)$ times. 
    Thus, the algorithm accesses at most $O(\log^2 W/\varepsilon)$ estimates $\hat c_j$ in $\widehat{C}$.
    According to \cref{cor:numberccappdelta},
    each estimate of $\hat{c}_j$ can be obtained in
    $O(\frac{\sqrt{W}}{\varepsilon^2}d\log(\frac{\sqrt{W}}{\varepsilon}d)\cdot\log(W))$ time. 
    Thus, the running time of \cref{alg:appcost} is 
    $O(\frac{\log^2 W}{\varepsilon}\cdot\frac{\sqrt{W}}{\varepsilon^2}d\log\frac{\sqrt{W}d}{\varepsilon}\cdot\log W)=
    O(\frac{\sqrt{W}d}{\varepsilon^3}\log^4\frac{Wd}{\varepsilon})
    =\tilde{O}(\frac{\sqrt{W}d}{\varepsilon^3})$.
\end{proof}

\subsection{Estimating the Profile Vector} \label{subsec:cost_k}

Recall that $\cost(G)=\sum_{k=1}^{n}\cost_k=\sum_{k=1}^{n}\sum_{i=1}^{n-k}w_i$, 
where 
$\cost_k$ is the cost of the $k$-SLC and $w_1,\dots, w_{n-1}$ are the weights of the minimum spanning tree in non-decreasing order. Now we give an algorithm for approximating the SLC profile vector $(\cost_1,\dots,\cost_n)$ and prove \cref{thm:profile}.

\distanceprofile*

In the following, we first derive a formula for the quantity $ \cost_k $ for each $ k = c_j $, where $ 1 \leq j \leq W $. Using this formula, we define a corresponding quantity $ \overline{\cost}_{B_i} $ for each $ B_i $, where $ 1 \leq i \leq t $, as introduced in \cref{def:interval_distance}. These quantities constitute our succinct representation. Finally, we define a profile oracle that, given any specified $ k $, outputs the estimator $ \widehat{\cost}_k $ of $ \cost_k $.
\paragraph{Formulas of Profile.}

By definition, we can equivalently express $\cost_k$ as follows:
\begin{align}
    \cost_k&=\sum_{i=1}^{n-k}w_i \nonumber\\
    &=(n-c_1)+(c_1-c_2)\cdot 2+\cdots+(c_{w_{n-k}-2}-c_{w_{n-k}-1})\cdot (w_{n-k}-1)+(c_{w_{n-k}-1}-k)\cdot w_{n-k} \nonumber\\
    &=n+c_1+c_2+\cdots+c_{w_{n-k}-1}-k\cdot w_{n-k} \nonumber\\
    &=n+\sum_{j=1}^{w_{n-k}-1}c_j-k\cdot w_{n-k}   \label{eqn:cost_k}
\end{align}

Now we observe that for any given integer $j\in \{1,\dots,W\}$, we can determine the number of edges in the MST that have weights at most $j$.
This allows us to compute $\cost_k$, where $k$ corresponds to the rank of the first edge in the MST with weight $j$.
We have the following lemma.

\begin{lemma}\label{lem:costk}
Given any integer $j\in\{1,...,W\}$ and define $k=c_j$, where $c_j$ is the number of connected components in $G_j$. 
Here, $G_j$ is a subgraph of $G$, and contains all edges with weights at most $j$.
Then we have,
\[
\cost_k=n+\sum_{i=1}^{j-1}c_i-c_j\cdot j.
\]
\end{lemma}
\begin{proof}
Let $n_j$ denote the number of edges in the $\MST$ with weight $j$. From our earlier analysis in \Cref{thm:formula_cost}, we have that 
    $n_j=c_{j-1}-c_j$ for any $1\leq j\leq W$, where $c_0=n$.
    Then the number of edges in the $\MST$ with weights at most $j$ is
\[n_1+n_2+\dots+n_j=(n-c_1)+(c_1-c_2)+\dots+(c_{j-1}-c_j)=n-c_j\]

    As $k=c_j$, the $(n-k)$-th smallest edge in the MST has weight $w_{n-k}=w_{n-c_j}=j$.
    Then the statement of lemma follows from \Cref{eqn:cost_k}.
\end{proof}

\paragraph{Algorithm to Estimate Profile.}

Our main idea of approximating the profile of clustering is to first define $\overline{\cost}_{B_i}$ for
non-integer $B_i,\ 1\leq i\leq t$, where $B_i$'s are interval endpoints defined in 
\cref{def:interval_distance}. 
We will make use of \cref{alg:appcost}, which performs a binary search using the search key $ B_i $ and parameters as described in \cref{alg:appcost}, ultimately returning the index $ \hat{j}_i $. By substituting $ j $ with $ \hat{j}_i $ and $ c_j $ with $ \bar{c}_j = B_i $ in the formula given in \cref{lem:costk}, we obtain:
\begin{align}
    \overline{\cost}_{B_i}=n+\sum_{j=1}^{\hat j_i-1}\bar c_j-B_i\cdot \hat j_i
    =n+\sum_{k=1}^{i-1}(\hat j_{k+1}-\hat j_k)\cdot B_k-B_i\cdot \hat j_i.
    \label{eqn:cost_Bi}
\end{align}

This approach gives a $t$-dimensional vector $(\overline{\cost}_{B_1},\dots,\overline{\cost}_{B_t})$ and effectively summarizes the histogram of profile vector.
The detailed algorithm is described in \cref{alg:appprofile}.

\begin{algorithm}[h!]
    \DontPrintSemicolon
    \caption{\textsc{AppProfile}($G, \varepsilon,W,d$)}
    \label{alg:appprofile}
    get sequence $\{\hat j_1,\dots,\hat j_t\}$ from \cref{alg:appcost} \textsc{AppCost($G,\varepsilon,W,d$)}\;
    set $t$ according to \cref{def:interval_distance}\;
    {set $\overline{\cost}_{B_1}=0$}\;
    \For{$i\in\{{2},\dots,t\}$}{
        set $B_i$ according to \cref{def:interval_distance}\;
        set $\overline{\cost}_{B_i}=n-B_i\cdot \hat j_i+\sum_{k=1}^{i-1}
        (\hat j_{k+1}-\hat j_k)\cdot B_{k}$\;
    }
    output the vector $(\overline{\cost}_{B_1},\dots,\overline{\cost}_{B_t})$
\end{algorithm}

Given $\{B_1,\dots,B_t\}$, the vector $(\overline{\cost}_{B_1},\dots,\overline{\cost}_{B_t})$, and any specified integer $k\in[1,n]$ such that $B_{i+1}\le k< B_{i}$, we define $\widehat{\cost}_k=\overline{\cost}_{B_{i+1}}$ to be the estimate for $\cost_k$, as described in \cref{alg:profile_oracle}. 

\begin{algorithm}[h!]
    \DontPrintSemicolon
    \caption{\textsc{ProfileOracle}($k, \{B_1,\dots,B_t\}, (\overline{\cost}_{B_1},\dots,\overline{\cost}_{B_t})$)}
    \label{alg:profile_oracle}
    define $B_0=\infty$\;
    use binary search over $(B_0,B_1,\dots,B_t)$, and find the index $i$ such that $B_{i+1}\le k<B_i$\;
    output $\widehat{\cost}_k:=\overline{\cost}_{B_{i+1}}$\;
\end{algorithm}

We call the vector $(\overline{\cost}_{B_1},\dots,\overline{\cost}_{B_t})$ a \emph{succinct representation} of the vector $(\widehat{\cost}_1,\dots,\widehat{\cost}_{n})$. We call \cref{alg:profile_oracle} a \emph{profile oracle}, as it can take any index $k$ as input and answer $\widehat{\cost}_k$. 

\paragraph{Analysis of \Cref{alg:appprofile,alg:profile_oracle}.}
Now we analyze \cref{alg:appprofile,alg:profile_oracle}. In particular, we will show that the vector $(\widehat{\cost}_1,\dots,\widehat{\cost}_{n})$ is a good approximation of the profile vector $({\cost}_1(G),\dots,{\cost}_{n}(G))$.
We first give an error bound for each individual $\widehat{\cost}_k$, and then bound the sum of the error and give the proof of \cref{thm:profile}.

\begin{fact}\label{fact:B_i-B_i+3}
    For any $1\le i\le t-3$, $B_i-B_{i+3}\le 8\varepsilon\cdot\max\{B_{i+1},\frac{n}{\sqrt{W}}\}$.
\end{fact}

\begin{proof}
    According to \cref{def:interval_distance}, when $i+3\le t_1$, i.e., $i\le t_1-3$, we have $B_i=\frac{n}{(1+\varepsilon)^{i-1}}$ and $B_{i+3}=\frac{n}{(1+\varepsilon)^{i+2}}$. Thus, since $\frac{1}{1+\varepsilon}\ge 1-2\varepsilon$,
    \[
    B_i-B_{i+3}=((1+\varepsilon)-\frac{1}{(1+\varepsilon)^2})\frac{n}{(1+\varepsilon)^i}\le (1+\varepsilon-(1-2\varepsilon)^2)B_{i+1} \le 5\varepsilon B_{i+1}
    \]

    When $i=t_1-2$, $B_i-B_{i+3}=(1+\varepsilon)^2B_{t_1}-B_{t_1+1}=(2\varepsilon+\varepsilon^2)B_{t_1}+(B_{t_1}-B_{t_1+1})$. Since $B_{t_1}-B_{t_1+1}<2\frac{\varepsilon n}{\sqrt{W}}$ and $B_{t_1}<(1+\varepsilon)\frac{n}{\sqrt{W}}<2\frac{n}{\sqrt{W}}$, $B_i-B_{i+3}< 3\varepsilon B_{t_1}+2\frac{\varepsilon n}{\sqrt{W}}<6\frac{\varepsilon n}{\sqrt{W}}+2\frac{\varepsilon n}{\sqrt{W}}= 8\varepsilon\frac{n}{\sqrt{W}}$.

    When $i=t_1-1$, $B_i-B_{i+3}=(1+\varepsilon)B_{t_1}-B_{t_1+2}=\varepsilon B_{t_1}+(B_{t_1}-B_{t_1+1})+\frac{\varepsilon n}{\sqrt{W}}<2\frac{\varepsilon n}{\sqrt{W}}+2\frac{\varepsilon n}{\sqrt{W}}+\frac{\varepsilon n}{\sqrt{W}} = 5\varepsilon\frac{n}{\sqrt{W}}$.

    When $i=t_1$, $B_i-B_{i+3}=B_{t_1}-B_{t_1+3}=(B_{t_1}-B_{t_1+1})+2\frac{\varepsilon n}{\sqrt{W}}<2\frac{\varepsilon n}{\sqrt{W}}+2\frac{\varepsilon n}{\sqrt{W}}=4\varepsilon\frac{n}{\sqrt{W}}$.

    When $i\ge t_1+1$, $B_i-B_{i+3}=3\varepsilon\frac{n}{\sqrt{W}}$, which finishes the proof of the fact.
\end{proof}

The following lemma provides an error bound for individual $\widehat{\cost}_k$, demonstrating that the error grows as $k$ increases. 
Observe that $|\widehat{\cost}_k-\cost_k|=|\overline{\cost}_{B_{i+1}}-\cost_k|\le|\overline{\cost}_{B_{i+1}}-\cost_{c_{j_{i+1}}}|+|\cost_{c_{j_{i+1}}}-\cost_k|$, {where $j_{i+1}$ is the smallest index such that $\hat{c}_{j_{i+1}}\le B_{i+1}$}. We bound these two terms separately.
First, for each interval, $B_{i+1}$ is close to $c_{j_{i+1}}$, so the gap $|\overline{\cost}_{B_{i+1}}-\cost_{c_{j_{i+1}}}|$ is small.
Second, by definition, $\cost_k=\sum_{j=1}^{n-k}w_j$, {where $w_j$ is the $j$-th smallest edge weight in the MST}. Since both $k$ and $c_{j_{i+1}}$ are constrained by nearby interval endpoints, the gap $|\cost_{c_{j_{i+1}}}-\cost_k|$ is also small. Consequently, the error $|\widehat{\cost}_k-\cost_k|$ is bounded.

\begin{lemma}\label{lem:single_costk_error}
    Assume that $\sqrt{W}\le n$ and let $0<\varepsilon<1$ be a parameter. With probability at least 3/4, for any integer $k\in\{1,\dots,n\}$, \cref{alg:profile_oracle} returns an estimate $\widehat{\cost}_k$ for the $k$-clustering cost, i.e., $\cost_k$, such that
    \[
    |\widehat{\cost}_k-\cost_k| \leq 4\varepsilon\cdot\cost_k+48\varepsilon\cdot\max\{k,\frac{n}{\sqrt{W}}\}\cdot \max\{w_{n-k}, j_{i+1}\}\le 4\varepsilon\cdot\cost_k+48\varepsilon\cdot\max\{k,\frac{n}{\sqrt{W}}\} W
    \]

    Given the succinct representation $(\overline{\cost}_{B_1},\dots,\overline{\cost}_{B_t})$, the running time of the algorithm is $O(\log (\frac{\log W}{\varepsilon}))$.
\end{lemma}

\begin{proof}
We first note that in the case that $\varepsilon<\frac{\sqrt{W}}{n}$, we can simply find the exact minimum spanning tree in $\tilde{O}(nd)\ll \tilde{O}(\frac{n^3}{W}d)=\tilde{O}(\frac{\sqrt{W}}{\varepsilon^3}d)$, and calculate the exact SLC profile vector in linear time. Then the theorem trivially holds. 
Thus in the following, we assume that $\varepsilon\geq \frac{\sqrt{W}}{n}$.

Recall from \cref{lem:bound_of_bar_cj} that with probability at least $\frac{3}{4}$, we have 
$|\bar c_j-c_j|\le 4\varepsilon\cdot\max\left\{\frac{n}{\sqrt{W}},c_j\right\}$ for all $j\in[W]$. In the following, we will assume that this event holds. 

For simplicity of notation, we use $j_i$  to denote  $\hat{j}_i$, the value returned by the binary search with key $B_i$ as invoked in \cref{alg:appcost}. 
By definitions of $\overline{\cost}_{B_{i}}$ and $\cost_{c_{j_{i}}}(G)$, we have that
\begin{align*}
|\overline{\cost}_{B_{i}}-\cost_{c_{j_{i}}}|
    &=\abs{(n+\sum_{j=1}^{j_{i}-1}\bar c_j-B_{i}\cdot j_{i})-(n+\sum_{j=1}^{j_{i}-1}c_i-c_{j_{i}}\cdot j_{i})}\\
    &=|\sum_{j=1}^{j_{i}-1}(\bar c_j-c_j)-B_{i}\cdot j_{i}+c_{j_{i}}\cdot j_{i}|\\
    &\leq\sum_{j=1}^{j_{i}-1}|\bar c_j-c_j|+|B_{i}-c_{j_{i}}|\cdot j_{i}\\
    &\leq 4\varepsilon\sum_{j=1}^{j_{i}-1}c_j+4\varepsilon\cdot c_{j_{i}}\cdot j_{i} \tag{applying the error bound for $\bar c_j$}\\
    &= 4\varepsilon(\cost_{c_{j_i}}-n+c_{j_{i}}\cdot j_{i})+4\varepsilon\cdot c_{j_{i}}\cdot j_{i} \tag{according to \cref{lem:costk}}\\
    &\le 4\varepsilon\cdot\cost_{c_{j_i}} + 8\varepsilon\cdot c_{j_i}\cdot j_i
\end{align*}

When $k=B_1=n$, as $\widehat{\cost}_{B_1}=0$ and $\cost_n=0$, we have $|\widehat{\cost}_k-\cost_k|=0$.

For $B_{i+1}\le k< B_i$ where $1\le i\le t-1$, we estimate $\cost_k$ using $\cost_{B_{i+1}}$, and thus we have
\begin{align*}
    |\widehat{\cost}_k-\cost_k|&\leq
    |\overline{\cost}_{B_{i+1}}-\cost_{c_{j_{i+1}}}|
    +|\cost_{c_{j_{i+1}}}-\cost_k|\\
    &\le 4\varepsilon(\cost_{c_{j_{i+1}}}+2 c_{j_{i+1}}\cdot j_{i+1})+|\cost_{c_{j_{i+1}}}-\cost_k|\\
    &\le 4\varepsilon(\cost_k+|\cost_{c_{j_{i+1}}}-\cost_k|+2B_i\cdot j_{i+1})+|\cost_{c_{j_{i+1}}}-\cost_k| \tag{since $c_{j_{i+1}}\le B_i$}\\
    &\le 4\varepsilon(\cost_k+2B_i\cdot j_{i+1})+5|\cost_{c_{j_{i+1}}}-\cost_k| \tag{since $\varepsilon<1$}\\
    &\le 4\varepsilon\cdot\cost_k+8\varepsilon\cdot\max\{k,\frac{n}{\sqrt{W}}\}\cdot j_{i+1}+|\cost_{c_{j_{i+1}}}-\cost_k| \tag{since $B_i\le B_{i+1}+\varepsilon\cdot\max\{B_{i+1},\frac{n}{\sqrt{W}}\}\le 2\max\{k,\frac{n}{\sqrt{W}}\}$}
\end{align*}

Thus, we only need to bound $|\cost_{c_{j_{i+1}}}-\cost_k|$.

Case (I): $c_{j_{i+1}}\le k$. In this case, $\cost_k\le \cost_{c_{j_{i+1}}}$ and $w_{n-k}\le j_{i+1}$. Since $\cost_k=\sum_{j=1}^{n-k}w_{j}$, 
\[
|\cost_{c_{j_{i+1}}}-\cost_k|=\cost_{c_{j_{i+1}}}-\cost_k=\sum_{j=n-k+1}^{n-c_{j_{i+1}}}w_j
\le(k-c_{j_{i+1}})\cdot w_{n-c_{j_{i+1}}} = (k-c_{j_{i+1}})\cdot j_{i+1}
\]

When $t-1\le i\le t-2$, $c_{j_{i+1}}\ge B_t=1$ and $k\le B_i\le B_{t-2}=2\frac{\varepsilon n}{\sqrt{W}}$, and thus $|\cost_{c_{j_{i+1}}}-\cost_k|\le 2\varepsilon\frac{n}{\sqrt{W}}\cdot j_{i+1}$.

When $i+3\le t$, i.e., $i\le t-3$, from \cref{lem:bound_of_xj}, $B_{i+3}\le c_{j_{i+1}}\le B_i$, and thus $k-c_{j_{i+1}}\le B_i-B_{i+3}$. According to \cref{fact:B_i-B_i+3}, $B_i-B_{i+3}\le 8\varepsilon\cdot\max\{B_{i+1},\frac{n}{\sqrt{W}}\} \le 8\varepsilon\cdot\max\{k,\frac{n}{\sqrt{W}}\}$. 
Thus, $|\cost_{c_{j_{i+1}}}-\cost_k|\le 8\varepsilon\cdot\max\{k,\frac{n}{\sqrt{W}}\}\cdot j_{i+1}$. Therefore
\[
|\widehat{\cost}_k-\cost_k| \leq 4\varepsilon\cdot\cost_k+(8\varepsilon+5\cdot8\varepsilon)\max\{k,\frac{n}{\sqrt{W}}\}\cdot j_{i+1}=4\varepsilon\cdot\cost_k+48\varepsilon\cdot\max\{k,\frac{n}{\sqrt{W}}\}\cdot j_{i+1}
\]

Case (II): $k<c_{j_{i+1}}$. In this case, $\cost_{c_{j_{i+1}}}\le\cost_k$ and $j_{i+1}\le w_{n-k}$. We have,
\[
|\cost_{c_{j_{i+1}}}-\cost_k|=\cost_k-\cost_{c_{j_{i+1}}}=\sum_{j=n-c_{j_{i+1}}+1}^{n-k}w_j
\le(c_{j_{i+1}}-k)\cdot w_{n-k} \le (B_i-B_{i+1})\cdot w_{n-k}
\]

When $i<t_1$, we have $B_i-B_{i+1}=\varepsilon B_{i+1}\le \varepsilon k$; when $i\ge t_1$, we have $B_i-B_{i+1}=\varepsilon\frac{n}{\sqrt{W}}$. Thus, $B_i-B_{i+1}\le\varepsilon\cdot\max\{k,\frac{n}{\sqrt{W}}\}$. Then,
\[
|\widehat{\cost}_k-\cost_k| \leq 4\varepsilon\cdot\cost_k+8\varepsilon\cdot\max\{k,\frac{n}{\sqrt{W}}\}\cdot j_{i+1}+8\varepsilon\cdot\max\{k,\frac{n}{\sqrt{W}}\}\cdot w_{n-k}
\]

In a conclusion, as $\max\{w_{n-k}, j_{i+1}\}\le W$, we have
\[
|\widehat{\cost}_k-\cost_k| \leq 4\varepsilon\cdot\cost_k+48\varepsilon\cdot\max\{k,\frac{n}{\sqrt{W}}\}\cdot \max\{w_{n-k}, j_{i+1}\}\le 4\varepsilon\cdot\cost_k+48\varepsilon\cdot\max\{k,\frac{n}{\sqrt{W}}\} W
\]

\paragraph{Running time analysis.}
For any $k$, given the succinct representation, one can simply perform binary search to find the first $i$ with $B_{i+1}\leq k<B_i$, and thus the running time of \cref{alg:profile_oracle} is $O(\log t)=O(\log(\frac{\log W}{\varepsilon}))$.
\end{proof}

Note that when $k$ is very close to $n$, it becomes impossible to estimate $ \cost_k $ within a constant factor in sublinear time. For instance, when $k = n - 1$, estimating $ \cost_k $ amounts to finding the minimum edge weight in the graph -- an inherently hard task to perform in sublinear time. Consequently, no algorithm can provide an estimator $ \widehat{\cost}_k $ with a constant-factor approximation guarantee for arbitrary values of $k$ under sublinear time constraints. However, in \Cref{thm:profile} we prove that the \textit{average} error bound for $\widehat{\cost}_k$ is $\varepsilon\cdot\cost_k$, that is, $\sum_{k=1}^{n}|\widehat{\cost}_k-\cost_k|\le\varepsilon\sum_{k=1}^{n}\cost_k$. 

\begin{proof}[Proof of \Cref{thm:profile}]

For $i<t_1$, $B_i-B_{i+1}=\varepsilon B_{i+1}\ge\varepsilon B_{t_1}\ge\frac{\varepsilon n}{\sqrt{W}}$,
and for $i\ge t_1$, $B_i-B_{i+1}\ge\frac{\varepsilon n}{\sqrt{W}}$.
Since $\varepsilon \geq \frac{\sqrt{W}}{n}$, we have that 
for every $i$, $B_i-B_{i+1}\ge 1$, then by the error bound for individual $\widehat{\cost}_k$ in \cref{lem:single_costk_error}, we have that \begin{align*}
    \sum_{k=1}^{n}|\widehat{\cost}_k-\cost_k|
    &\le\sum_{i=1}^{t-1} \sum_{k\in[B_{i+1},B_i)}\left(
    4\varepsilon\cdot\cost_k+48\varepsilon\cdot\max\{k,\frac{n}{\sqrt{W}}\}\cdot \max\{w_{n-k}, j_{i+1}\} \right)\\
    &=4\varepsilon\sum_{k=1}^{n}\cost_k+48\varepsilon\sum_{i=1}^{t}\sum_{k\in[B_{i+1},B_i)}\max\{k,\frac{n}{\sqrt{W}}\}\cdot\max\{w_{n-k},j_{i+1}\}\\
    &\le 4\varepsilon\cdot\cost(G)+48\varepsilon\sum_{i=1}^{t_1}\sum_{k\in[B_{i+1},B_i)}k\cdot\max\{w_{n-k},j_{i+1}\}+48\varepsilon\sum_{i=t_1+1}^{t}\sum_{k\in[B_{i+1},B_i)}\frac{n}{\sqrt{W}}\cdot W \tag{by the definition of $\cost(G)$, and $\max\{w_{n-k},j_{i+1}\}\le W$}\\
    &\le 4\varepsilon\cdot\cost(G)+48\varepsilon\frac{1}{\varepsilon}\cdot\frac{\varepsilon n}{\sqrt{W}}\cdot\frac{n}{\sqrt{W}}W+48\varepsilon\sum_{i=1}^{t_1}\sum_{k\in[B_{i+1},B_i)}k\cdot\max\{w_{n-k},j_{i+1}\} \tag{since $t-t_1=t_2\le\frac{1}{\varepsilon}$, and when $i\ge t_1$, $B_i-B_{i+1}=\frac{n}{\sqrt{W}}$}\\
    &= 4\varepsilon\cdot\cost(G)+48\varepsilon n^2+48\varepsilon\sum_{i=1}^{t_1}\sum_{k\in[B_{i+1},B_i)}k\cdot\max\{w_{n-k},j_{i+1}\}\\
    &\le 196\varepsilon\cdot\cost(G)+48\varepsilon\sum_{i=1}^{t_1}\sum_{k\in[B_{i+1},B_i)}k\cdot\max\{w_{n-k},j_{i+1}\} \tag{since $n^2\le4\cost(G)$}
\end{align*}

Define $(*)=48\varepsilon\sum_{i=1}^{t_1}\sum_{k\in[B_{i+1},B_i)}k\cdot\max\{w_{n-k},j_{i+1}\}$, now we only need to bound $(*)$.

Case (I): $j_{i+1}\le w_{n-k}$. In this case, $(*)\le 48\varepsilon\sum_{k=1}^{{n-1}} k\cdot w_{n-k}=48\varepsilon\cdot\cost(G)$, by the fact that $\cost(G)=\sum_{i=1}^{n-1}(n-i)w_i=\sum_{k=1}^{n-1}k\cdot w_{n-k}$.

Case (II): $w_{n-k}\le j_{i+1}$. In this case, since $k<B_i$, we have that
\[
(*)\le 48\varepsilon\sum_{i=1}^{t_1}(B_i-B_{i+1})B_i\cdot j_{i+1} = 48\varepsilon^2\sum_{i=1}^{t_1}B_{i+1}B_i\cdot j_{i+1}
\]

Next we bound $\sum_{i=1}^{t_1} B_{i+1} B_i\cdot j_{i+1}$. On one hand, $\sum_{i=1}^{t_1} B_{i+1} B_i\cdot j_{i+1} =(1+\varepsilon)\sum_{i=1}^{t_1-1}B_{i+1}^2 j_{i+1}
=(1+\varepsilon)\sum_{i'=2}^{t_1} B_{i'}^2 j_{i'}$;
on the other hand, 
$\sum_{i=1}^{t_1} B_{i+1} B_i\cdot j_{i+1}=\frac{1}{1+\varepsilon}\sum_{i=1}^{t_1-1}B_i^2 j_{i+1}$.

Therefore, $(1+\varepsilon-\frac{1}{1+\varepsilon})\sum_{i=1}^{t_1} B_{i+1} B_i\cdot j_{i+1}=\sum_{i=1}^{t_1-1}B_i^2j_{i+1}-\sum_{i=2}^{t_1}B_i^2 j_i
=\sum_{i=1}^{t_1}B_i^2(j_{i+1}-j_i)-B_{t_1}^2 j_{t_1+1}+B_1^2 j_1$. 
Note that according to \cref{fact:distanceinterval}, $j_1=1$.

Moreover, when $i\le t_1$, $B_i\ge B_{t_1}\ge\frac{n}{\sqrt{W}}$. 
If $\frac{n}{\sqrt{W}}\ge 2$, then $B_i^2-B_i\ge\frac{B_i}{2}$;
and if $n\ge 2$, $\frac{n(n-1)}{2}\ge\frac{n^2}{4}$.
Then $\widehat{\cost}(G)=\frac{n(n-1)}{2}+\frac{1}{2}\sum_{i=1}^{t-1}(j_{i+1}-j_i)\cdot(B_i^2-B_i)
\ge \frac{n^2}{4}+\frac{1}{4}\sum_{i=1}^{t_1}B_i^2 (j_{i+1}-j_i)$.
Since $\widehat{\cost}(G)$ approximates $\cost(G)$ by $(1+\varepsilon)$ factor, we have
\[
\frac{(1+\varepsilon)^2-1}{1+\varepsilon}\sum_{i=1}^{t_1} B_{i+1} B_i\cdot j_{i+1} = \frac{3\varepsilon}{1+\varepsilon}\sum_{i=1}^{t_1} B_{i+1} B_i\cdot j_{i+1} 
\le \sum_{i=1}^{t_1}B_i^2(j_{i+1}-j_i)+n^2
\le 4\cdot\widehat{\cost}(G) \le 4(1+\varepsilon)\cost(G)
\]

That is, $\varepsilon \sum_{i=1}^{t_1} B_{i+1} B_i\cdot j_{i+1}\le\frac{4}{3}(1+\varepsilon)^2\cost(G) \le 6\cost(G)$. Then,
\[
(*)=48\varepsilon^2 \sum_{i=1}^{t_1} B_{i+1} B_i\cdot j_{i+1} \le 48\varepsilon\cdot 6\cost(G)=288\varepsilon\cdot\cost(G)
\]

Combining the above analysis, we have $\sum_{k=1}^{n}\abs{\widehat{\cost}_k-\cost_k}\leq (196+288)\varepsilon\cdot\cost(G)=484\varepsilon\cdot\cost(G)$.

Replacing $\varepsilon$ with $\varepsilon/484$, we get an succinct representation of $(\widehat{\cost}_1,\dots,\widehat{\cost}_{n})$ such that each $\widehat{\cost}_k$ is a $(1+\varepsilon)$-estimator \emph{on average}.

\paragraph{Running time analysis.}
Note that we first invoke \cref{alg:appcost} to obtain the sequence $\{j_0,j_1,\dots,j_t\}$,
which takes 
$O(\frac{\sqrt{W}d}{\varepsilon^3}\log^4\frac{Wd}{\varepsilon})$ time.
Given the sequence, to estimate each $\overline{\cost}_{B_i}$, 
we need to calculate $\sum_{k=1}^{i-1}(j_{k+1}-j_k)\cdot B_k$
in $O(t)$ time, for all $i\in[1,t]$.
Thus, getting the estimated vector $(\overline{\cost}_{B_1},\dots,\overline{\cost}_{B_t})$
can be done in $O(t^2)=O(\log^2(\sqrt{W}/\varepsilon)/\varepsilon^2)$ time.
Thus, \cref{alg:appprofile} runs in time $O(\frac{\sqrt{W}d}{\varepsilon^3}\log^4\frac{Wd}{\varepsilon})$.
\end{proof}

\section{Sublinear Algorithms in Similarity Case}\label{sec:similaritymeasure}
Now, we consider single-linkage clustering in similarity graphs. In this case, the agglomerative clustering algorithm merges the two clusters with the highest similarity at each step. The similarity between two clusters is defined as the maximum similarity between any pair of their members. Recall that $w_1^{(s)},w_2^{(s)},...,w_{n-1}^{(s)}$ are the weights of a
maximum spanning tree (MaxST) of $G$ in non-increasing order, and the cost of the SLC in the similarity graph is defined as 
\[
\cost^{(s)}(G)= \sum_{k=1}^{n} \cost_k^{(s)}
=\sum_{i=1}^{n-1}(n-i)\cdot w_i^{(s)},
\]
where $\cost_k^{(s)}= \sum_{i=1}^{n-k} w_i^{(s)}$ is the cost of the corresponding $k$-clustering. 

\subsection{Cost Formula for Similarity-based Clustering}

We first derive an equivalent formula for the clustering cost. 
We let $n_j$ denote the number of edges of 
weight $j$ in the MaxST.
For each weight $j$, 
$G_j^{(s)}$ contains all edges of weight $\geq j$.
We let $c_j^{(s)}$ be the number of 
connected components in $G_j^{(s)}$.  We observe that $c_1^{(s)}=1$ and $c_{W+1}^{(s)}=n$,
if we assume the whole graph $G$ is connected.
Furthermore, it holds that  $\sum_{i<j}n_i=c_j^{(s)}-1$, which gives $n_j=c_{j+1}^{(s)}-c_j^{(s)}$.

\begin{theorem}\label{thm:costsimformula}
Let $G$ be a connected graph on $n$ vertices, with edge weights from $\{1,\dots, W\}$. Then
\[
\cost^{(s)}(G) = \sum_{j=1}^W\frac{(c_j^{(s)}+n-1)(n-c_j^{(s)})}{2}
\]
\end{theorem}

\begin{proof}
The cost of clustering in the similarity graph can be derived as:
\begin{align*}
    \cost^{(s)}(G)&=\sum_{i=1}^{n-1}(n-i)\cdot w_i^{(s)}
    \tag{by \cref{eqn:costG_similarity}}\\
    &=\sum_{i=1}^{n_W}(n-i)\cdot W
    +\sum_{i=n_W+1}^{n_W+n_{W-1}}(n-i)\cdot(W-1)
    +\dots+\sum_{i=n_W+\dots+n_2+1}^{n_W+\dots+n_2+n_1}(n-i)\cdot1
    \tag{reorganizing the sum by grouping the terms according to edge weights}\\
    &=\sum_{i=1}^{n-c_W^{(s)}}(n-i)\cdot W+\dots+\sum_{i=n-c_3^{(s)}+1}^{n-c_2^{(s)}}(n-i)\cdot2+\sum_{i=n-c_2^{(s)}+1}^{n-c_1^{(s)}}(n-i)\cdot1
    \tag{since $n_j=c_{j+1}^{(s)}-c_j^{(s)}$}\\
    &=\sum_{i=c_W^{(s)}}^{n-1}i\cdot W+\dots+\sum_{i=c_2^{(s)}}^{c_3^{(s)}-1}i\cdot2+\sum_{i=c_1^{(s)}}^{c_2^{(s)}-1}i\cdot1
    \tag{substituting $(n-i)$ with $i$}\\
    &=\sum_{i=c_W^{(s)}}^{n-1}i\cdot (W-1)+\dots+\sum_{i=c_2^{(s)}}^{c_3^{(s)}-1}i\cdot(2-1)+\sum_{i=c_1^{(s)}}^{n-1}i\cdot1
    \tag{factoring out $i$ from each summation and combining the terms}\\
    &=\sum_{i=c_W^{(s)}}^{n-1}i+\dots+\sum_{i=c_2^{(s)}}^{n-1}i+\sum_{i=c_1^{(s)}}^{n-1}i
    \tag{repeating this decomposition until all summations end with $i=n-1$}\\
    &=\sum_{j=1}^W\frac{(c_j^{(s)}+n-1)(n-c_j^{(s)})}{2}
\end{align*}
\end{proof}

We remark that the total weight of the MaxST can be derived in an analogous way:
$\cost(\mathrm{MaxST})= \sum_{j=1}^{W} j \cdot (c_{j+1}^{(s)}  - c_j^{(s)} )
= 1 \cdot (c_2^{(s)} - c_1^{(s)}) + 2 \cdot (c_3^{(s)} - c_2^{(s)}) + 3 \cdot (c_4^{(s)} - c_3^{(s)}) + \cdots + W \cdot (c_{W+1}^{(s)} - c_W^{(s)}) 
=nW -\sum_{j=1}^{W} c_j^{(s)}=\sum_{j=1}^W (n-c_j^{(s)})$.

By the above theorem and the fact that $c_j^{(s)}\ge c_1^{(s)}=1$ for any $j$, we have the following facts.
\begin{fact}\label{fact:lbcostsim-n2}
It holds that $\cost^{(s)}(G)\geq \frac{(c_1^{(s)}+n-1)(n-c_1^{(s)})}{2}=\frac{(1+n-1)(n-1)}{2}=\frac{n(n-1)}{2}$.
\end{fact}

\begin{fact}\label{fact:lbcostsim-maxst}
It holds that $\cost^{(s)}(G)\geq \frac{n}{2}\cdot\cost(\mathrm{MaxST})$.
\end{fact}

\subsection{Estimating the Clustering Cost}
The formula of $\cost^{(s)}(G)$ in \Cref{thm:costsimformula} includes a sum of products $(c_j^{(s)}+n-1)\cdot (n-c_j^{(s)})$. Our approach is to estimate each product in $\tilde{O}(Wd\poly(1/\varepsilon))$ time with small additive error. 
To do so, we note that by using \cref{alg:appncc}, the first term $c_j^{(s)}+n-1$ can be estimated with an additive error of
$O(\varepsilon n)$, which is sufficient for our purpose. For the second term $D_j:=n-c_j^{(s)}$, 
it is also tempting to directly applying \Cref{alg:appncc} to approximate it. However, the resulting additive error will be too large when $c_j^{(s)}$ is large, i.e., close to $n$. 

\subsubsection{Approximating $D_j$}

\paragraph{High Level Idea of Approximating $D_j$.} To resolve the above issue, we relate $D_j$ to the number of \emph{non-isolated} vertices and the connected components in the graph induced by such vertices in the threshold graph $G_j^{(s)}$. 

Specifically, for a given graph $G$ and its subgraph $H$ with the same vertex set, we define $H_{\nis}$ to be the subgraph of $H$ induced by all non-isolated vertices w.r.t. $H$.
The number of vertices in $H_{\nis}$ is defined as $n'$,
and the number of connected components is $c'$. It is important to note that the number of connected components $\textrm{cc}(H)$ of $H$ is exactly $n-n'+c'$. Thus, the quantity $n-\textrm{cc}(H)=n'-c'$, and one can approximate $n-\textrm{cc}(H)$ by approximating $n'$ and $c'$ respectively. We remark that the edges of subgraph $H$ are implicitly defined, so we must
inspect \emph{all} edges incident to a vertex $v$ in $G$ to find its neighbors in $H$. Based on this, we give an algorithm for approximating the $n-\textrm{cc}(H)=n'-c'$ with small enough additive error in \cref{alg:appncc_sim}. 

Now we give a bit more details of \cref{alg:appncc_sim}, we sample vertices to simultaneously estimate
the number $\textrm{cc}(H)$ of connected components, 
the number $c'$ of connected components in $H_{\nis}$ (each of which is a component with at least $2$ vertices in $H$) and
the number $n'$ of non-isolated vertices. If there are only a few isolated vertices, we estimate $ n - \textrm{cc}(H) $ using $ n - \hat{c} $, where $ \hat{c} $ is an estimate of $ \textrm{cc}(H) $. Otherwise, we estimate $ n - \textrm{cc}(H) = n'_j - c'_j $ using $ \hat{n'} - \hat{c'} $, where $ \hat{n'} $ and $ \hat{c'} $ are estimates of $ n' $ and $ c' $, respectively.

Then to approximate $D_j=n-c_j^{(s)}$, we let $H$ be the  threshold graph $G^{(s)}_j$, and let 
and $n'_j, c'_j$ be the number of vertices, and the number of connected components of $H_{\nis}$, respectively. 
Then, one can observe that $\textrm{cc}(H)=c^{(s)}_j=n-n'_j+c'_j$, so $D_j=n'_j-c'_j$, which can be approximated by \cref{alg:appncc_sim} with appropriate input parameters.

\begin{algorithm}[h!]
    \DontPrintSemicolon
    \caption{\textsc{AppNCCSim}($G, H,\varepsilon,k,d$)}
    \label{alg:appncc_sim}
    \SetKwInOut{Input}{input}
    \SetKwInOut{Output}{output}
    \Input{graph $G$, subgraph $H$, approximation parameter $\varepsilon$, threshold 
    $k$, avg. degree $d$}
    \Output{an estimate $\widehat{D}$ of $n-c$, where $c$ is the number of 
    connected component in $H$}
    \BlankLine
    choose $r=\lceil64\frac{k}{\varepsilon^2}\rceil$ vertices $v_1,\dots,v_r$ uniformly at random\;
    \For{each sampled vertex $v_i$}{
        identify neighbors of $v_i$ in $H$ by examining all incident edges in $G$\;
        \eIf{$v_i$ is isolated in $H$}{
        set $x_i=0$\;
        }
        {set $x_i=1$\;}
    }
    set $\hat{n'}=\frac{n}{r}\sum_{i=1}^r x_i$\;
    choose another $r=\lceil64\frac{k}{\varepsilon^2}\rceil$ vertices 
    $u_1,\dots,u_r$ uniformly at random\;
    set threshold $\Gamma=\lceil4\frac{k}{\varepsilon}\rceil$ and $d^{(G)}=d\cdot\Gamma$\;
    \For{each sampled vertex $u_i$}{
        set $\alpha_i=0$, $\beta_i=0$\;
        
        identify neighbors of $u_i$ in $H$ by examining all incident edges in $G$\;
        let $d_{u_i}^{(G)}$ be the degree of $u_i$ in $G$\;
        \eIf{$u_i$ is isolated in $H$}
        {set $\beta_i=1$\;}
        {
            (*) flip a coin\;
            
            \If{(heads) \& (\# vertices visited $\in H<\Gamma$ during BFS) \& (no visited vertex $\in H$ has degree in $G>d^{(G)}$ during BFS)}{
            resume BFS on $H$, doubling the number of visited edges in $G$\;
            
            \eIf{this allows BFS on $H$ to complete}
            {set $\alpha_i=\beta_i=d_{u_i}^{(G)}\cdot 2^{\text{\# coin flips}}/\text{\#edges visited in $G$}$\;
            }{go to (*)\;}
            }
        }    
    }
    set $\hat{c}=\frac{n}{r}\sum_{i=1}^r\beta_i$,
    $\hat{c'}=\frac{n}{r}\sum_{i=1}^{r}\alpha_i$\;
    \eIf{$\hat{n'}< 0.5n$}{
    \Return $\widehat{D}=\hat{n'}-\hat{c'}$\;
    }
    {
    \Return $\widehat{D}=n-\hat{c}$\;
    }
\end{algorithm}

\paragraph{Analysis of \cref{alg:appncc_sim}.} We first provide an approximation guarantee of \cref{alg:appncc_sim} in the following lemma. 

\begin{lemma}\label{lem:appncc_sim}
Let $0<\varepsilon<1$, and $k$ is an integer greater than $10$.
Let $r=\lceil\frac{64k}{\varepsilon^2}\rceil$, and $\Gamma=\lceil\frac{{4k}}{\varepsilon}\rceil$.
Suppose the average degree $d$ of graph $G$ is known, and set $d^{(G)}=d\cdot\Gamma=d\cdot\lceil\frac{{4k}}{\varepsilon}\rceil$.
Given access to the graph $G$ and an implicit subgraph $H\subseteq G$ 
(where $H$ shares the same vertex set as $G$, and the edges of $H$ are determined by evaluating 
conditions on the edges of $G$), \Cref{alg:appncc_sim} runs in time $O(\frac{{k}}
{\varepsilon^2}d\log(\frac{k}{\varepsilon} d))$ and outputs a value $\widehat{D}$ such that
\[
\abs{\widehat{D}-D}\leq\varepsilon\cdot\max\left\{\frac{n}{k},\min\{D,n-D\}\right\}
\]

with error probability at least $\frac{3}{4}$,
where $D=n-c$ and $c$ is the number of connected components in $H$.
\end{lemma}

Consider the \cref{alg:appncc_sim}, we have the following lemma on the estimators $\hat{n'}$ and $\hat{c'}$.

\begin{lemma}\label{lem:on_n'_c'} 
    Let $r,\Gamma,d^{(G)}$ be the same parameters set in \cref{lem:appncc_sim}.
    Let $U$ be the set of vertices that lie in components in subgraph $H_{nis}$ with fewer 
    than $\Gamma$ vertices; and all of these vertices in the original graph 
    are of degree at most $d^{(G)}$. Let $c'_U$ denote the number of connected 
    components in $H_{nis}[U]$. Then we have,
    \[
    \E[\hat{n'}]=n',\ \Var[n']=\frac{n'(n-n')}{r}
    \]

    And $\abs{\hat{c'}-c'}\leq\max\{\frac{n}{k},c'\}$, with probability at least $\frac{7}{8}$.
\end{lemma}
\begin{proof}
\textbf{Analysis on $\hat{n'}$.}
For each sampled vertex $v_i$, $x_i=1$ if $v_i\in H_{nis}$; and $x_i=0$ otherwise. Then,
\[
\E[x_i^2]=\E[x_i]=\Pr[v_i\in H_{nis}]=\frac{1}{n}\sum_{v_i\in V}\frac{n'}{n}
=\frac{n'}{n},\ 
\Var[x_i]=\frac{n'}{n}(1-\frac{n'}{n})
\]

As $\hat{n'}=\frac{n}{r}\sum_{i=1}^r x_i$, so
$\E[\hat{n'}]=\frac{n}{r}\cdot r\cdot \E[x_1]=n'$, and 
$\Var[\hat{n'}]=(\frac{n}{r})^2\cdot r\cdot \Var[x_1]=\frac{n'(n-n')}{r}$.

\textbf{Analysis on $\hat{c'}$.} By the definition of $c'_U$, we have, 
$c'-\frac{2n}{\Gamma}\leq c'_U\leq c'$.

Since $\alpha_i$ has the same value as $\beta_i$ in \cref{lem:onbetai}, 
except that when $u_i$ is isolated in $H$, $\alpha_i=0$, so 

\begin{align*}
    \E[\alpha_i]&=\Pr[u_i\in H_{nis}]\E[\alpha_i|u_i\in H_{nis}]+\Pr[u_i\in H\setminus H_{nis}]\E[\alpha_i|u_i\in H\setminus H_{nis}]\\
    &=\frac{n'}{n}\cdot\frac{1}{n'}\left(\sum_{u\in H_{nis}\setminus U}0+
    \sum_{u\in H_{nis}\cap U}2^{-\left\lceil\log\left(\frac{\vol(C_{u})}{d_{u}^{(G)}}\right)\right\rceil}
    \cdot2^{\left\lceil\log\left(\frac{\vol(C_{u})}{d_{u}^{(G)}}\right)\right\rceil}
    \frac{d_u^{(G)}}{\vol(C_{u})}\right)
    +\frac{n-n'}{n}\cdot 0\\
    &=\frac{c'_U}{n}
    \tag{since $\sum_{u\in U}\frac{d_u^{(G)}}{\vol(C_{u})}=c_U$}
\end{align*}

And $\Var[\alpha_i]\leq\E[\alpha_i^2]\leq2\cdot\E[\alpha_i]=\frac{2c'_U}{n}$. Then,
\[
\E[\hat{c'}]=\frac{n}{r}\cdot\sum_{i=1}^r\E[\alpha_i]
=\frac{n}{r}\cdot r\cdot\frac{c'_U}{n}=c'_{U},\ 
\Var[\hat{c'}]\leq\frac{n^2}{r^2}\cdot r\cdot\frac{2c'_U}{n}=
\frac{2c'_U n}{r}\leq\frac{2c'n}{r}
\]

If $c'< \frac{n}{k}$, by Chebyshev's inequality, 
\[
\Pr[|\hat{c'}-c'_U|\geq \frac{\varepsilon n}{2k}]\leq\frac{\Var[\hat{c'}]
\cdot 4k^2}{\varepsilon^2n^2}
\leq\frac{2nc'\cdot 4k^2}{r\cdot\varepsilon^2n^2}
\leq\frac{8c'\cdot k^2}{64k\cdot n}
\leq\frac{1}{8}
\]

The last inequality holds as $c'<\frac{n}{k}$. Therefore, the error of $\hat{c'}$ is
\[
|\hat{c'}-c'|\leq|\hat{c'}-c'_U|+|c'_U-c'|\leq\frac{\varepsilon n}{2k}+\frac{2n}{\Gamma}
\leq \frac{\varepsilon n}{2k}+\frac{\varepsilon n}{2k}
=\frac{\varepsilon n}{k}
\]

Else, when $c'\geq \frac{n}{k}$,
\[
\Pr[|\hat{c'}-c'_U|\geq\frac{\varepsilon c'}{2}]\leq\frac{4\Var[\hat{c'}]}{\varepsilon^2c'^2}
\leq\frac{8nc}{r\cdot\varepsilon^2c'^2}
\leq\frac{8n}{64k\cdot c'}
\leq\frac{1}{8}
\]

The last inequality holds as $c'\geq\frac{n}{k}$. And the error of $\hat{c'}$ is
\[
|\hat{c'}-c'|\leq|\hat{c'}-c'_U|+|c'_U-c'|\leq
\frac{\varepsilon c'}{2}+\frac{2n}{\Gamma}= \frac{\varepsilon c'}{2}
+\frac{\varepsilon n}{2k} \leq\varepsilon c'
\]

Combining both cases, we have that 
$\abs{\hat{c'}-c'}\leq \varepsilon\max\{\frac{n}{k},c'\}$, 
with probability more than $\frac{7}{8}$.
\end{proof}

Now we are ready to prove \cref{lem:appncc_sim}.

\begin{proof}[Proof of \cref{lem:appncc_sim}] 
In the graph $H_{\nis}$, each connected component has at
least two vertices, implying $2c'\leq n'$. 
Therefore, $D=n'-c'\geq 2c'-c'$, which simplifies to $0\leq c'\leq D$.
Consequently, $D\leq n'\leq2D$.
According to \cref{lem:on_n'_c'}, $\abs{\hat{c'}-c'}\leq\varepsilon\cdot\max\{\frac{n}{k}, c'\}$.
Therefore, $\abs{\hat{c'}-c'}\leq\varepsilon\cdot\max\left\{\frac{n}{k}, D\right\}$,
with probability at least $\frac{7}{8}$.

\begin{itemize}
    \item Case (I): $0\leq D\leq\frac{n}{k}$. Then $n'\leq2D\leq2\frac{n}{k}$.
    By Chebyshev's inequality,
    \[
    \Pr[\abs{\hat{n'}-n'}\geq \varepsilon\frac{n}{k}]
    \leq\frac{\Var[\hat{n'}]k^2}{\varepsilon^2 n^2}
    =\frac{n'(n-n')k^2}{r\cdot\varepsilon^2 n^2}
    \leq\frac{2\frac{n}{k}(n-n')k^2}{64\cdot kn^2}
    \leq\frac{1}{32}
    \]

    Therefore, $\hat{n'}\leq n'+\varepsilon\frac{n}{k}<(2+1)\frac{n}{k}<0.5n$, 
    as $k\geq 10$ and $\varepsilon<1$.
    In this case, we let $\widehat{D}=\hat{n'}-\hat{c'}$, then
    \[
    \abs{\widehat{D}-D}\leq\abs{\hat{n'}-n}+\abs{\hat{c'}-c'}
    \leq\varepsilon\frac{n}{k}+\varepsilon\cdot\max\left\{\frac{n}{k}, D\right\}
    \leq2\varepsilon\cdot\frac{n}{k}
    \]

    By union bound over $\hat{n'}$ and $\hat{c'}$, the above inequality is true
    with probability more than $1-(\frac{1}{32}+\frac{1}{8})\geq\frac{3}{4}$.
    
    \item Case (II): $\frac{n}{k}<D\leq0.2n$. Since $D\leq n'\leq2D$,
    by Chebyshev's inequality,
    \[
    \Pr[\abs{\hat{n'}-n'}\geq \frac{\varepsilon}{2}D]
    \leq\frac{4\Var[\hat{n'}]}{\varepsilon^2 D^2}
    =\frac{4n'(n-n')}{r\cdot\varepsilon^2D^2}
    \leq\frac{8\cdot D\cdot(n-D)}{64\cdot k\cdot D^2}
    \leq\frac{n-\frac{n}{k}}{8\cdot k\cdot\frac{n}{k}}
    \leq\frac{1}{8}
    \]

    Therefore, we have $\hat{n'}\leq n'+\frac{\varepsilon}{2}D
    \leq(2+\frac{\varepsilon}{2})D<0.5n$.
    In this case, we let $\widehat{D}=\hat{n'}-\hat{c'}$. Then,
    \[
    \abs{\widehat{D}-D}\leq\abs{\hat{n'}-n}+\abs{\hat{c'}-c'}
    \leq\frac{\varepsilon}{2}D+\varepsilon\cdot\max\left\{\frac{n}{k}, D\right\}
    \leq2\varepsilon D
    \]

    The above inequality is true with probability at least 
    $1-(\frac{1}{8}+\frac{1}{8})=\frac{3}{4}$.

    \item Case (III): ${0.6n<D\leq n}$. Then $n-n'\leq n-D\leq 0.4n$. By Chebyshev's inequality,
    \[
    \Pr[\abs{\hat{n'}-n'}\geq \varepsilon\cdot0.1n]\leq\frac{\Var[\hat{n'}]}{\varepsilon^2 \cdot0.01n^2}
    =\frac{n'(n-n')}{r\cdot\varepsilon^2\cdot0.01n^2}
    \leq\frac{n-n'}{r\cdot\varepsilon^2\cdot0.01n}
    \leq\frac{0.4n}{64k\cdot 0.01n}
    \leq\frac{1}{16}
    \]
    
    The last inequality follows from the assumption that $k\geq 10$.
    
    Then with probability at least $1-\frac{1}{16}$, $\hat{n'}\geq n'-\varepsilon\cdot0.1n\geq0.6n-\varepsilon\cdot0.1n\geq 0.5n$. 
Conditioned on this event, we let $\widehat{D}=n-\hat{c}$.
    From the proof of \cref{lem:appncc}, we know that 
    $\abs{\widehat{D}-D}\leq\abs{\hat{c}-c}\leq\varepsilon\cdot
    \max\{\frac{n}{k},c\} = \varepsilon\cdot\max\{\frac{n}{k},n-D\}$, with probability at least $\frac{7}{8}$. Thus, with probability 
    at least $1-(\frac{1}{16}+\frac{1}{8})\geq\frac{3}{4}$, $\abs{\widehat{D}-D}\leq\varepsilon\cdot\max\{\frac{n}{k},n-D\}$.

    \item Can (IV): $0.2n<D\leq 0.6n$. If the algorithm uses
    $\widehat{D}=n-\hat{c}$ to estimate $D$, then with probability at least 
    $\frac{7}{8}\geq\frac{3}{4}$, we have $\abs{\widehat{D}-D}
    \leq\abs{\hat{c}-c}\leq \varepsilon c\leq4\varepsilon D$ (as when $D>0.2n$, $c=n-D<0.8n<4D$).
    
    On the other hand, if the algorithm uses $\widehat{D}=\hat{n'}-\hat{c'}$ to estimate $D$,
    we have that $0.2n\leq n'<n$. By Chebyshev's inequality,
    \[
    \Pr[\abs{\hat{n'}-n'}\geq \varepsilon\cdot0.2n]
    \leq\frac{\Var[\hat{n'}]}{\varepsilon^2\cdot0.04 n^2}
    =\frac{n'(n-n')}{r\cdot\varepsilon^2\cdot0.04 n^2}
    \leq\frac{n'0.8n}{r\cdot\varepsilon^2\cdot0.04 n^2}
    \leq\frac{2}{64k\cdot0.1}
    \leq\frac{1}{16}
    \]

    Therefore, with probability more than $1-(\frac{1}{16}+\frac{1}{8})\geq\frac{3}{4}$, we have
    \[
    \abs{\widehat{D}-D}\leq\abs{\hat{n'}-n}+\abs{\hat{c'}-c'}
    \leq\varepsilon\cdot0.2n+\varepsilon\cdot\max\left\{\frac{n}{k}, D\right\}
    \leq\varepsilon D+\varepsilon D
    \leq2\varepsilon D.
    \]
\end{itemize}

Combining the above case analysis, we have that with probability more than $\frac{3}{4}$,
\begin{displaymath}
\abs{\widehat{D}-D} \leq \left\{ \begin{array}{ll}
4\varepsilon\cdot\max\left\{\frac{n}{k},D\right\} & \textrm{if $0\leq D\leq 0.5n$}\\
2\varepsilon\cdot\max\{\frac{n}{k},n-D\} & \textrm{if $0.5n<D\leq n$}\\
\end{array} \right.
\end{displaymath}

Replacing $\varepsilon$ with $4\varepsilon$, the statement of the lemma is satisfied.

\paragraph{Running time analysis.} 
Similar to analysis of running time of \cref{alg:appncc},
consider our truncation approach. 
It ensures that there are at most $O(\Gamma^2 d)$ edges visited during BFS, starting from a sampled vertex $u$.
Consequently, the number of coin flips is at most $O(\log(\Gamma d)$, and the expected running time associated with the sampled vertex $u$ is $O(d_u\log(\Gamma d)$.
Since we sample $r$ vertices and each vertex $u\in V$ is sampled with probability $\frac{1}{n}$, the overall running time of \cref{alg:appncc_sim} is $O(r\cdot\frac{1}{n}\sum_{u\in V} d_u\log(\Gamma d))$.
Given that we set $r=\frac{k}{\varepsilon^2}$, and $\Gamma=\frac{k}{\varepsilon}$, so the running time is $O(\frac{k}{\varepsilon^2}d\log(\frac{k}{\varepsilon}d))$.

\end{proof}

Similar to \cref{alg:median-trick}, by applying median trick to \cref{alg:appncc_sim}, 
we have the following corollary.
\begin{corollary}\label{cor:appncc_sim_delta}
For every $1\leq i \leq W$, an appropriately amplified version of \cref{alg:appncc_sim}
accepts the same inputs as \cref{alg:appncc_sim}, except for an additional 
success parameter $\delta$;
and it outputs $\widehat{D}$ satisfying
\[
\abs{\widehat{D}-D}\leq\varepsilon\cdot\max\left\{\frac{n}{k},\min\{D,n-D\}\right\}
\]

with probability at least $1-\delta$.
Here, $D=n-c$ where $c$ is the number of connected components in the subgraph $H$. 
The algorithm runs in time $O(\frac{k\log(1/\delta)}{\varepsilon^2}d\log(\frac{k}{\varepsilon}d))$.
\end{corollary}

\subsubsection{Adapting Binary Search for Clustering Cost Estimation}
Note that $(D_1, \dots, D_W)$ is a non-increasing array with values in the range $[0, n-1]$. 
Therefore, we can utilize binary search as described in \cref{sec:binarysearch}, similar to the approach used for the distance case, to obtain effective estimators for all elements in this sequence. Specifically, we will apply binary search on the sequence $\widehat{D}=(\widehat{D}_1,\dots,\widehat{D}_W)$ with some newly defined search keys $B_i$'s. Now we define $\widehat{D}_j$'s in the following lemma.

\begin{lemma}\label{lem:hatDj}
For any $1\leq j\leq W$, let $\widehat{D}_j=\min\{\max\{\widehat{D}_j',0\},n-1\}$, where $\widehat{D}_j'$ is the output of \cref{cor:appncc_sim_delta} with input $G$, $H=G_j^{(s)}$, $\varepsilon/8$, 
$k=W$, $d$, $\delta=1/(8W)$. 
Then with probability at least $7/8$,
it holds that for \emph{all} $1\leq j\leq W$, $\widehat{D}_j\in[0,n-1]$, and
\[
\abs{\widehat{D}_j-D_j}\leq T_j:=\frac{\varepsilon}{8}\cdot\max\left\{\frac{n}{W},\min\{D_j,n-D_j\}\right\}.
\]
\end{lemma}
\begin{proof}
According to \cref{cor:appncc_sim_delta}, for any $j\in[W]$, we know that with probability of $1-\varepsilon/(8W)$, the estimator $\widehat{D}_j'$ satisfies
$\abs{\widehat{D}_j'-D_j}\leq
\frac{\varepsilon}{8}\cdot\max\left\{\frac{n}{W},\min\{D_j,n-D_j\}\right\}=T_j$. 
Besides, by rounding $\widehat{D}_j'$ to $\widehat{D}_j=\min\{\max\{\widehat{D}_j',0\},n-1\}$, we can guarantee that $D_j\in[0,n-1]$.
If $\widehat{D}_j'$ is out of the range $[0,n-1]$, then the rounding process decreases the gap between $\widehat{D}_j$ and $D_j$, and ensures the error remains within the error bound $T_j$.

By the union bound on all estimates
$\widehat{D}_1,\dots,\widehat{D}_W$, with probability at least $\frac{7}{8}$, 
for \emph{all} $j\in [W]$, $\widehat{D}_j\in[0,n-1]$, and $\abs{\widehat{D}_j-D_j}\leq T_j$.
\end{proof}

\emph{Throughout the following, we will assume that for \emph{all} $j$, the inequality $\abs{\widehat{D}_j-D_j}\leq T_j$ holds, and $\widehat{D}_j\in[0,n-1
]$}. According to \cref{lem:hatDj}, this occurs with probability at least $7/8$. 

Now we define the endpoints of intervals which partition $[0,n-1]$.

\begin{definition}\label{def:interval_similarity}
    Let $0<\varepsilon<1$, $t_1$ be the largest integer such that $n-{t_1}\frac{\varepsilon n}{W}\ge n-\frac{n}{W}$, $t_2$ be the largest integer such that $n-(1+\varepsilon)^{t_2}\frac{n}{W}\ge 0.5n$ (i.e. the largest integer such that $\frac{n}{2(1+\varepsilon)^{t_2}}\ge\frac{n}{W}$), and let $t_3$ be the largest integer such that $\frac{n}{W}(1-t_3\varepsilon)\ge\frac{\varepsilon n}{W}$. Note that $t_1=\lfloor\frac{1}{\varepsilon}\rfloor$, $t_2=\lfloor\log_{1+\varepsilon}\frac{W}{2}\rfloor$, and $t_3=\lfloor\frac{1-\varepsilon}{\varepsilon}\rfloor$.
    Then we define $B_i$ as
    \[
    B_i=
    \begin{cases}
        n-1 & \quad \text{if $i=1$}\\
        n-i\frac{\varepsilon n}{W} & \quad \text{if $1< i\le t_1$}\\
        n-(1+\varepsilon)^{i-t_1}\frac{n}{W} & \quad \text{if $t_1<i\le t_1+t_2$}\\     
        \frac{n}{2(1+\varepsilon)^{i-(t_1+t_2)}} & \quad \text{if $t_1+t_2<i\le t_1+2t_2$}\\
        \frac{n}{W}(1-(i-t_1-2t_2)\varepsilon) & \quad \text{if $t_1+2t_2<i\le t_1+2t_2+t_3$}\\
        0 & \quad \text{if $i=t:=t_1+2t_2+t_3+1$}
    \end{cases}
    \]
\end{definition}

We have the following fact.

\begin{fact}\label{fact:interval_similarity}
    It holds that
    \begin{enumerate}
        \item $t=t_1+2t_2+t_3+1=O(\frac{2}{\varepsilon}+2\log_{1+\varepsilon}W)=O(\log W/\varepsilon)$.
        
        \item When $1\le i\le t_1$, $n-\frac{n}{W}\le B_i\le n$;
        
        when $t_1<i\le t_1+t_2$, $0.5n\le B_i<n-\frac{n}{W}$;
        
        when $t_1+t_2<i\le t_1+2t_2$, $\frac{n}{W}<B_i\le 0.5n$;
        
        and when $t_1+2t_2<B_i\le t_1+2t_2+t_3$, $\frac{\varepsilon n}{W}\le B_i<\frac{n}{W}$.
        
        \item the gap between neighboring endpoints is,
        \[
        B_{i-1}-B_i=
        \begin{cases}
            \frac{\varepsilon n}{W} & \quad \text{if $1<i\le t_1$}\\
            \varepsilon(1+\varepsilon)^{i-1}\frac{n}{W}=\varepsilon(n-B_{i-1}) & \quad \text{if $t_1+2\le i\le t_1+t_2$}\\     
            (1+\varepsilon)B_i-B_i=\varepsilon B_i & \quad \text{if $t_1+t_2+2\le i\le t_1+2t_2$}\\
            \frac{\varepsilon n}{W} & \quad \text{if $t_1+2t_2+2\le i\le t_1+2t_2+t_3$}\\
        \end{cases}
        \]
        \item $B_{t_1}-B_{t_1+1}\ge n-\frac{n}{W}-(n-(1+\varepsilon)\frac{n}{W})=\frac{\varepsilon n}{W}$;

        $B_{t_1+t_2}-B_{t_1+t_2+1}\ge 0.5n-\frac{n}{2(1+\varepsilon)} =\frac{\varepsilon n}{2(1+\varepsilon)}=\varepsilon B_{t_1+t_2+1}$;

        $B_{t_1+2t_2-1}-B_{t_1+2t_2}=\varepsilon B_{t_1+2t_2}\ge\frac{\varepsilon n}{W}$;

        $B_{t_1+2t_2}-B_{t_1+2t_2+1}\ge\frac{n}{W}-\frac{n}{W}(1-\varepsilon)=\frac{\varepsilon n}{W}$

        \item 
     $n-(t_1+1)\frac{\varepsilon n}{W}<n-\frac{n}{W}$, and thus $B_{t_1}=n-t_1\frac{n}{W}<n-\frac{n}{W}(1-\varepsilon)$;
     
$n-(1+\varepsilon)^{t_2+1}\frac{n}{W}<\frac{n}{2}$, 
        and thus $n-B_{t_1+t_2}>\frac{n}{2(1+\varepsilon)}$;
        
 $\frac{n}{2(1+\varepsilon)^{t_2+1}}<\frac{n}{W}$,
        and thus $B_{t_1+2t_2}<(1+\varepsilon)\frac{n}{W}$.
    \end{enumerate}
\end{fact}

For each $j\in[W]$, we define the error bound as $T_j=\frac{\varepsilon}{8}\max\{\frac{n}{W},\min\{D_j,n-D_j\}\}$.
Then, we prove that the sequence of $B_i$'s defined above form a valid discretization of $[0,n-1]$ w.r.t.
$(D_1,\dots,D_W)$ and error bound $T=(T_1,\dots,T_W)$.

\begin{lemma}\label{lem:discretization_similarity}
    The sequence of endpoints $(B_1,\dots,B_t)$ defined in \cref{def:interval_similarity} is
    a valid discretization of $[0,n-1]$ w.r.t. $(D_1,\dots,D_W)$ and error bound 
    $T_j=\frac{\varepsilon}{8}\max\{\frac{n}{W},\min\{D_j,n-D_j\}\}$. \end{lemma}
\begin{proof}
We consider the sequence $(\widehat{D}_1, \dots, \widehat{D}_W)$ as the approximation of $D=(D_1,\dots,D_W)$, where each $\widehat{D}_j$ is defined as in \cref{lem:hatDj}. According to the previous assumption, it holds that for all $j$, $\widehat{D}_j\in[0,n-1]$ and $\abs{\widehat{D}_j - D_j}\leq T_j$.

 Now we consider any $j\in [W]$. Suppose that $D_j\in[B_{i+1},B_i]$.

When $i=1$, $B_{i+1}-B_{i+2}=\frac{\varepsilon n}{W}>T_j$.
    
    When $2\le i\le t_1-1$, we have $B_{i-1}-B_{i}=\frac{\varepsilon n}{W}$
    and $B_{i+1}-B_{i+2}\ge\frac{\varepsilon n}{W}$.
    In this case, $D_j\ge B_{i+1}>n-\frac{n}{{W}}$, and as $(1+\varepsilon)^2\leq4$, 
    we have $T_j=\frac{\varepsilon n}{8 W} < \min\{B_{i-1}-B_i,B_{i+1}-B_{i+2}\}$.

    When $i=t_1$, we have $B_{i-1}-B_i=\frac{\varepsilon n}{W}$,
    and $B_{i+1}-B_{i+2}=\varepsilon(n-B_{i+1})=\varepsilon(1+\varepsilon)\frac{n}{W}$.
    In this case,
    \[
    T_j=\frac{\varepsilon}{8}\max\{\frac{n}{W},n-D_j\}
    \le\frac{\varepsilon}{8}\max\{\frac{n}{W},n-B_{i+1}\}
    =\frac{\varepsilon}{8}(1+\varepsilon)\frac{n}{W}
    \le\min\{B_{i-1}-B_i,B_{i+1}-B_{i+2}\}
    \]

    When $i=t_1+1$, we have $B_{i-1}-B_i\ge\frac{\varepsilon n}{W}$,
    and $B_{i+1}-B_{i+2}=\varepsilon(n-B_{i+1})=\varepsilon(1+\varepsilon)^2\frac{n}{W}$.
    In this case,
    \[
    T_j=\frac{\varepsilon}{8}\max\{\frac{n}{W},n-D_j\}
    \le\frac{\varepsilon}{8}\max\{\frac{n}{W},n-B_{i+1}\}
    =\frac{\varepsilon}{8}(1+\varepsilon)^2\frac{n}{W}
    \le\max\{B_{i-1}-B_i,B_{i+1}-B_{i+2}\}
    \]

    When $t_1+2\le i\le t_1+t_2-2$, we have $B_{i-1}-B_i=\varepsilon(n-B_{i-1})=
    \frac{\varepsilon}{(1+\varepsilon)^2}(n-B_{i+1})$, and $B_{i+1}-B_{i+2}=\varepsilon(n-B_{i+1})$. 
    In this case, $D_j\le B_{t_1+2}<n-\frac{n}{W}$, and $D_j\ge B_{t_2-2}>0.5n$. Thus,
    \[
    T_j=\frac{\varepsilon}{8}(n-D_j)\le\frac{\varepsilon}{8}(n-B_{i+1})
    <\frac{\varepsilon}{(1+\varepsilon)^2}(n-B_{i+1})=\min\{B_{i-1}-B_i,B_{i+1}-B_{i+2}\}
    \]

    When $i=t_1+t_2-1$, we have $B_{i-1}-B_i=\frac{\varepsilon}{(1+\varepsilon)^2}(n-B_{i+1})$.
    Since $B_{i+1}\ge0.5n$, then $B_{i+1}-B_{i+2}=B_{t_1+t_2}-B_{t_1+t_2+1}\ge\frac{\varepsilon n}{2(1+\varepsilon)}\ge\frac{\varepsilon}{1+\varepsilon}(n-B_{i+1})$.
    In this case,
    \[
    T_j=\frac{\varepsilon}{8}(n-D_j)\le\frac{\varepsilon}{8}(n-B_{i+1})
    \le\frac{\varepsilon}{(1+\varepsilon)^2}(n-B_{i+1})
    =\min\{B_{i-1}-B_i,B_{i+1}-B_{i+2}\}
    \]

    When $i=t_1+t_2$, we have $B_{i-1}-B_i=B_{t_1+t_2-1}-B_{t_1+t_2}=\frac{\varepsilon}{1+\varepsilon}(n-B_{t_1+t_2})>\frac{\varepsilon n}{2(1+\varepsilon)^2}$,
    and $B_{i+1}-B_{i+2}=\varepsilon B_{i+2}=\frac{\varepsilon n}{2(2+\varepsilon)^2}$.
    In this case,
    \[
    T_j=\frac{\varepsilon}{8}\min\{n-D_j,D_j\}\le\frac{\varepsilon}{16}n\le\frac{\varepsilon}{2(1+\varepsilon)^2}n
    =\min\{B_{i-1}-B_i,B_{i+1}-B_{i+2}\}
    \]

    When $t_1+t_2+1\le i\le t_1+2t_2-2$, we have $B_{i-1}-B_i\ge\varepsilon B_i$, and
    $B_{i+1}-B_{i+2}=\varepsilon B_{i+2}=\frac{\varepsilon}{(1+\varepsilon)^2}B_i$. 
    In this case, $D_j\le B_{t_1+t_2+2}<0.5n$, and $D_j\ge B_{t_1+t_2+t_3-2}>\frac{n}{W}$. Thus, 
    \[
    T_j=\frac{\varepsilon}{8}D_j\le\frac{\varepsilon}{8}B_i
    <\frac{\varepsilon}{(1+\varepsilon)^2}B_i=\min\{B_{i-1}-B_i,B_{i+1}-B_{i+2}\}
    \]

    When $i=t_1+2t_2-1$, we have $B_{i-1}-B_i\ge\varepsilon B_i$.
    Since $\frac{B_{i+1}}{1+\varepsilon}=\frac{B_{t_1+2t_2}}{1+\varepsilon}<\frac{n}{W}$,
    $B_i=(1+\varepsilon)B_{i+1}\le(1+\varepsilon)^2\frac{n}{W}$.
    Then, $B_{i+1}-B_{i+2}\ge\frac{\varepsilon n}{W}\ge\frac{\varepsilon}{(1+\varepsilon)^2}B_i$.
    Thus, 
    \[
    T_j=\frac{\varepsilon}{8}D_j\le\frac{\varepsilon}{8}B_i
    <\frac{\varepsilon}{(1+\varepsilon)^2}B_i=\min\{B_{i-1}-B_i,B_{i+1}-B_{i+2}\}
    \]

    When $t_1+2t_2\le i$, we have $B_{i-1}-B_i\ge\frac{\varepsilon n}{W}$ and $B_{i+1}-B_{i+2}\ge\frac{\varepsilon n}{W}$. 
    In this case, $D_j\le B_i\le B_{t_1t_2+t_3+2}<\frac{n}{W}$. Thus, 
    $T_j=\frac{\varepsilon n}{8W}<\min\{B_{i-1}-B_i,B_{i+1}-B_{i+2}\}$.

    Therefore, the sequence $(B_1,\dots,B_t)$ is a valid discretization of $[0,n-1]$ w.r.t. $(D_1,\dots,D_W)$
    and the error bound $T=(T_1,\dots,T_W)$.
\end{proof}

\paragraph{The Algorithm} Now we are ready to exploit the binary search based algorithm for succinctly approximating the sequence $D=(D_1,\dots,D_W)$ to estimate the clustering cost $\cost^{(s)}(G)$. The algorithm simply performs binary search \cref{alg:binarysearch} on the estimated array 
$\widehat{D}=\{\widehat{D}_1,\dots,\widehat{D}_W,0\}$ with search keys $B_1,\dots,B_t$ 
as defined in \cref{def:interval_similarity}, where $\widehat{D}_j$'s are estimators of $D_j$'s as defined in \cref{lem:hatDj}. Then we obtain the corresponding indices $\hat j_1,\dots,\hat j_t$. For $i\in\{1,\dots,t-1\}$, let $\hat J_i=\{\hat j_i,\dots,\hat j_{i+1}-1\}$. Note that $\hat J_1,\dots, \hat J_{t-1}$ form a partition of the index set $[W]=\{1,\dots, W\}$. 
For any $j\in\hat J_i$, we define $\overline{D}_j=B_i$ as the estimate for $D_j$ and then use  $\overline{D}_j$'s to estimate the cost $\cost^{(s)}(G)$. The algorithm is described in \cref{alg:appcost_sim}. Note that we estimate $c_j^{(s)}$ and $D_j$ separately,
and then estimate $\cost^{(s)}(G)$ by multiplying the estimates of $c_j^{(s)}$ and $D_j$. Furthermore, for the same reason as in distance case, we do not need to access the whole array 
$\widehat{D}$, but only access $\poly(\log W)$ values $\widehat{D}_i$ in it. We also store previously estimated
value of $\widehat{D}_i$ to maintain consistency. 

\begin{algorithm}[h!]
    \DontPrintSemicolon
    \caption{\textsc{AppCostSim}($G, \varepsilon, d, W$)}
    \label{alg:appcost_sim}
    \SetKwInOut{Input}{input}
    \SetKwInOut{Output}{output}
    \Input{graph $G$, approximation parameter $\varepsilon$}
    \Output{$\widehat{\cost^{(s)}}(G)$, which estimates $\cost^{(s)}(G)$}
    \BlankLine
    set $t$ according to 
    \cref{def:interval_similarity}\;
    get estimates $\{\hat{c}_1^{(s)},\dots,\hat{c}_W^{(s)}\}$ from 
    \cref{cor:numberccappdelta} with parameters
    $G,H=G_j,\varepsilon,k=1,d^{(G)}=d\cdot\lceil\frac{4}{\varepsilon}\rceil,\delta=1/8W$\;
    \For{$1\leq i\leq t$}{
        set $B_i$ according to \cref{def:interval_similarity}\;
        invoke \textsc{BinarySearch$(\widehat{D},1,W,B_i)$} and get output index $\hat j_i$, where
        $\widehat{D}=(\widehat{D}_1,\dots,\widehat{D}_W,0)$ and each $\widehat{D}_j=\min\{\max\{\widehat{D}_j',0\},n-1\}$
        and $\widehat{D}_j'$ is the output of the algorithm in \cref{cor:appncc_sim_delta} with input $G, H=G^{(s)}_j, \varepsilon/8, k=W$, $\delta=1/(8W)$ \Comment{for consistency, reuse stored value of $\widehat{D}_j$, if previously estimated}\;
    }
    for any $j\in\hat J_i:=\{\hat j_i,\dots,\hat j_{i+1}-1\}$, we let $\overline{D}_j=\overline{D}_{\hat j_i}=B_i$\;
    set $\widehat{\cost^{(s)}}(G)=\frac{1}{2}
    \sum_{j=1}^{W}(\hat{c}_j^{(s)}+n-1)\cdot\overline{D}_j$\;
    output $\widehat{\cost^{(s)}}(G)$ and the sequence $(\hat j_1,\dots,\hat j_t)$
\end{algorithm}

Now we first give the following guarantee on the estimates $\overline{D}_j$'s.

\begin{lemma} \label{lem:bound_of_bar_D}
Assume that for \emph{all} $j\in [W]$, the inequality $\abs{\widehat{D}_j-D_j}\leq T_j$ holds. Then it holds that for all $j\in [W]$, 
\[
\abs{\overline{D}_j-D_j} \leq 5\varepsilon\cdot\max\left\{\frac{n}{W},\min\{D_j,n-D_j\}\right\}.
\]
\end{lemma}
\begin{proof}
By \cref{lem:discretization_similarity}, the sequence of endpoints $(B_1,\dots,B_t)$ defined in \cref{def:interval_similarity} is a valid discretization of $[0,n-1]$ w.r.t. $(D_1,\dots,D_W)$ and error bound 
$T=(T_1,\dots,T_W)$. 
Thus, by \cref{lem:bound_of_xj}, for every $i\in[1,t-1]$, and every $j\in\hat J_i$,
we have $B_{i+2}\le D_j\le B_{i-1}$. By case distinction,

\begin{itemize}
    \item When $i=1$, $B_{i+2}\le D_j\le B_1=n$,
    then $\abs{\overline{D}_j-D_j}\le B_1-B_{i+2}=2\frac{\varepsilon n}{W}$.
    
    \item When $i\le t_1-2$, we have $B_{i-1}-B_i=\frac{\varepsilon n}{W}$,
    and $B_i-B_{i+2}=2\frac{\varepsilon n}{W}$.
    
    Thus, $\abs{\overline D_j-D_j}\le\max\{B_{i-1}-B_i,B_i-B_{i+2}\}=2\frac{\varepsilon n}{W}$.

    \item When $i=t_1-1$, since $B_{t_1}<n-\frac{n}{W}(1-\varepsilon)$,
    then $B_{t_1-1}-B_{t_1+1}=(B_{t_1}+\frac{\varepsilon n}{W})
    -(n-(1+\varepsilon)\frac{n}{W})<3\frac{\varepsilon n}{W}$.
    Besides, $B_{i-1}-B_i=\frac{\varepsilon n}{W}$.
    Thus, $\abs{\overline{D}_j-D_j}\le\max\{B_{i-1}-B_i,B_i-B_{i+2}\}
    <3\frac{\varepsilon n}{W}\le3\varepsilon\max\{\frac{n}{W}, n-D_j\}$.

    \item When $i=t_1$, we have $B_i-B_{i+2}=B_{t_1}-B_{t_1+2}<(n-\frac{n}{W}(1-\varepsilon))-(1-(1+\varepsilon)^2)\frac{n}{W}
    \le4\frac{\varepsilon n}{W}$.
    Besides, $B_{i-1}-B_i=\frac{\varepsilon n}{W}$.
    Thus, $\abs{\overline{D}_j-D_j}
    <4\frac{\varepsilon n}{W}\le4\varepsilon\max\{\frac{n}{W}, n-D_j\}$.

    \item When $i=t_1+1$, we have
    $\begin{cases}
        \text{if } D_j>\overline D_j,
        &D_j-\overline D_j\le B_{i-1}-B_i<(n-\frac{n}{W}(1-\varepsilon))-(n-(1+\varepsilon)\frac{n}{W})=2\frac{\varepsilon n}{W}\\
        \text{else, } &\overline D_j-D_j\le B_i-B_{i+2}
        =((1+\varepsilon)^2-1)(n-B_i) \le 3\varepsilon(n-D_j)
    \end{cases}$

    \item When $t_1+2\le i \le t_1+t_2-2$, since $D_j\le B_{i-1}$, then $n-D_j\ge n-B_{i-1}$; 
    and if $D_j\le \overline D_j=B_i$, then $n-D_j\ge n-B_i$. Thus,
    $\begin{cases}
        \text{if } D_j>\overline D_j,
        &D_j-\overline D_j\le B_{i-1}-B_i=\varepsilon (n-B_{i-1})\le\varepsilon(n-D_j)\\
        \text{else, } &\overline D_j-D_j\le B_i-B_{i+2}
        =((1+\varepsilon)^2-1)(n-B_i) \le 3\varepsilon(n-D_j)
    \end{cases}$

    \item When $i=t_1+t_2-1$, since $n-B_{t_1+t_2}>\frac{n}{2(1+\varepsilon)}$,
    then $n-B_{i}=\frac{n-B_{t_1+t_2}}{1+\varepsilon}>\frac{n}{2(1+\varepsilon)^2}$.
    Thus, $B_i-B_{i+2}<n-\frac{n}{2(1+\varepsilon)^2}-\frac{n}{2(1+\varepsilon)}\le\frac{5\varepsilon n}{2(1+\varepsilon)^2}$.
    As $\frac{5\varepsilon n}{2(1+\varepsilon)^2}
    = \frac{5\varepsilon}{1+\varepsilon}B_{i+2}\le 5\varepsilon D_j$,
    and $\frac{5\varepsilon n}{2(1+\varepsilon)^2}< 5\varepsilon(n-B_i)$, we have
    
    $\begin{cases}
        \text{if } D_j>\overline{D}_j,
        &D_j-\overline{D}_j \le \varepsilon(n-B_{i-1}) \le \varepsilon(n-D_j)\\
        \text{else, } & \overline{D}_j-D_j \le B_i-B_{i+2} < \frac{5\varepsilon n}{2(1+\varepsilon)^2}
        < 5\varepsilon\min\{D_j, n-D_j\}
    \end{cases}$

    \item When $i=t_1+t_2$, $B_i-B_{i+2}<n-\frac{n}{2(1+\varepsilon)}-\frac{n}{2(1+\varepsilon)^2}\le\frac{5\varepsilon n}{(1+\varepsilon)^2}$.
    As $\frac{5\varepsilon n}{2(1+\varepsilon)^2}
    = {5\varepsilon}B_{i+2}\le 5\varepsilon D_j$,
    and $\frac{5\varepsilon n}{2(1+\varepsilon)^2}< \frac{5\varepsilon}{1+\varepsilon}(n-B_i)<5\varepsilon(n-B_i)$, we have
    $\begin{cases}
        \text{if } D_j>\overline{D}_j,
        &D_j-\overline{D}_j \le \varepsilon(n-B_{i-1}) \le \varepsilon(n-D_j)\\
        \text{else, } & \overline{D}_j-D_j \le B_i-B_{i+2} < \frac{5\varepsilon n}{2(1+\varepsilon)^2}
        <5\varepsilon\min\{n-D_j,D_j\}
    \end{cases}$

    \item When $i=t_1+t_2+1$, $B_{i-1}-B_i<n-\frac{n}{2(1+\varepsilon)}-\frac{n}{2(1+\varepsilon)}=\frac{\varepsilon n}{1+\varepsilon}$. 
    Note that $\frac{\varepsilon n}{1+\varepsilon}=2\varepsilon B_i$ and 
    $\frac{\varepsilon n}{1+\varepsilon}<2\varepsilon (n-B_{i-1})$.
    Then,
    $\begin{cases}
        \text{if } D_j>\overline{D}_j,
        &D_j-\overline{D}_j \le B_{i-1}-B_i
        < \frac{\varepsilon n}{1+\varepsilon} < 2\varepsilon\min\{B_i,(n-B_{i-1})\}
        \le2\varepsilon\min\{\varepsilon D_j, n-D_j\}\\
        \text{else, } & \overline{D}_j-D_j \le B_i-B_{i+2} = ((1+\varepsilon)^2-1)B_{i+2}\le 3\varepsilon B_{i+2}
        \le 3\varepsilon D_j
    \end{cases}$

    \item When $t_1+t_2+2\le i \le t_1+2t_2-2$, we have $B_{i-1}-B_i=\varepsilon B_i$, 
    and $B_i-B_{i+2}=((1+\varepsilon)^2-1)B_{i+2}\le 3\varepsilon B_{i+2}$.
    Since $D_j\ge B_{i+1}$, then 
    $\begin{cases}
        \text{if } D_j>\overline D_j=B_i,
        &D_j-\overline D_j\le B_{i-1}-B_i\le\varepsilon D_j\\
        \text{else, } &\overline D_j-D_j\le B_i-B_{i+2}\le 3\varepsilon D_j
    \end{cases}$

    \item When $i=t_1+2t_2-1$, since $B_{i+1}=B_{t_1+2t_2}<(1+\varepsilon)\frac{n}{W}$,
    then $B_i-B_{i+2}=(1+\varepsilon)B_{i+1}-B_{i+2} <(1+\varepsilon)^2\frac{n}{W}-\frac{n}{W}(1-\varepsilon)\le 4\frac{\varepsilon n}{W}$.
    Thus,
    $\begin{cases}
        \text{if } D_j>\overline{D}_j,
        &D_j-\overline{D}_j \le B_{i-1}-B_i=\varepsilon B_i\le \varepsilon D_j \\
        \text{else, } & \overline{D}_j-D_j \le B_i-B_{i+2} < 4\frac{\varepsilon n}{W}
    \end{cases}$

    \item When $i=t_1+2t_2$, we have
    $\begin{cases}
        \text{if } D_j>\overline{D}_j,
        &D_j-\overline{D}_j \le B_{i-1}-B_i=\varepsilon B_i\le \varepsilon D_j \\
        \text{else, } & \overline{D}_j-D_j \le B_i-B_{i+2} < (1+\varepsilon)\frac{n}{W}-\frac{n}{W}(1-2\varepsilon)=3\frac{\varepsilon n}{W}
    \end{cases}$

    \item When $i=t_1+2t_2+1$, we have
    $\begin{cases}
        \text{if } D_j>\overline{D}_j,
        &D_j-\overline{D}_j \le B_{i-1}-B_i<(1+\varepsilon)\frac{n}{W}-\frac{n}{W}(1-\varepsilon)=2\frac{\varepsilon n}{W} \\
        \text{else, } & \overline{D}_j-D_j \le B_i-B_{i+2} = 2\frac{\varepsilon n}{W}
    \end{cases}$

    \item When $t_1+2t_2+2\le i$, we have $B_{i-1}-B_i=\frac{\varepsilon n}{W}$, 
    and $B_i-B_{i+2}=2\frac{\varepsilon n}{W}$.

    Thus, $\abs{\overline D_j-D_j}\le\max\{B_{i-1}-B_i,B_i-B_{i+2}\}=2\frac{\varepsilon n}{W}$.

    \item When $i=t$, $\abs{\overline{D}_j-D_j}\le B_{t-1}-B_{t}=B_{t_1+2t_2+t_3}-B_t<2\frac{\varepsilon n}{W}$.
    
\end{itemize}

Therefore, for \emph{all} $j\in[W]$, we have that 
$\abs{\overline D_j-D_j}\le 5\varepsilon\max\{\frac{n}{W},\min\{D_j,n-D_j\}\}$.
\end{proof}

\paragraph{Analysis of \cref{alg:appcost_sim}}

Now we give the analysis on the approximation ratio of \cref{alg:appcost_sim} as stated in \cref{thm:cost_similarity}.

\similarity*

\begin{proof}[Proof of \cref{thm:cost_similarity}]
Setting $\delta=\frac{1}{8W}$ and $k=1$, we obtain estimates
$\{\hat{c}_1^{(s)},\dots,\hat{c}_W^{(s)}\}$ from \cref{cor:numberccappdelta}, such that each estimate 
$\hat{c}_j^{(s)}$ satisfies $\abs{\hat{c}_j^{(s)}-c_j^{(s)}}\leq\varepsilon n$
with probability at least $1-\frac{1}{8W}$.
Using the union bound, all these estimates simultaneously have additive error $\varepsilon n$ with probability at least $\frac{7}{8}$.

Besides, from \cref{lem:hatDj} and \cref{lem:bound_of_bar_D}, with probability at least 
$\frac{7}{8}$, it holds that for \emph{all} $j\in[W]$, $\abs{\overline{D}_j-D_j}\leq 5\varepsilon\cdot\max\left\{\frac{n}{W},\min\{D_j,n-D_j\}\right\}$.
Therefore, for \emph{all} $j\in[W]$, both $\overline{D}_j$ and $\hat{c}_j^{(s)}$ satisfy the above error guarantee with probability at least $1-\frac{1}{8}-\frac{1}{8}=\frac{3}{4}$.
For the remainder of the proof, we assume that this event holds. 

Denote $A_j=c_j^{(s)}+n-1$, and $\hat{A}_j=\hat{c}_j^{(s)}+n-1$.
Since $1\leq c_j\leq n$, we have $n\leq A_j\leq 2n$. Furthermore, it holds that $\abs{\hat{A}_j-A_j}
=\abs{\hat{c}_j^{(s)}-c_j^{(s)}}\leq\varepsilon n\leq\varepsilon A_j$.

Case (I): $0.5n<D_j\leq n$. In this case, we have that $n<2D_j$, and thus 
\[
\abs{\overline{D}_j-D_j}\leq 5\varepsilon\cdot\max\left\{\frac{n}{{W}},n-D_j\right\}
\leq 5\varepsilon\left(\frac{2}{{W}}+2-1\right)D_j\leq 15\varepsilon D_j.
\]
Furthermore, 
\[
(1-\varepsilon)(1-15\varepsilon)A_j\cdot D_j
\leq\hat{A}_j\cdot\overline{D}_j\leq(1+\varepsilon)(1+15\varepsilon)A_j\cdot D_j
\]

That is, 
\[
\abs{\hat{A}_j\cdot\overline{D}_j-A_j\cdot D_j}\leq 
(16\varepsilon+15\varepsilon^2)A_j\cdot D_j \leq 31\varepsilon A_j\cdot D_j.
\]

Case (II): $0\leq D\leq0.5 n$. In this case, we have that $\abs{\overline{D}_j-D_j}
\leq 5\varepsilon\cdot\max\left\{\frac{n}{{W}},D_j\right\}
\leq 5\varepsilon\left(\frac{n}{W}+D_j\right)$.
Thus, 
\[
\hat{A}_j\cdot\overline{D}_j\leq(1+\varepsilon)A_j\left(D_j+5\varepsilon D_j+5\varepsilon\frac{n}{W}\right)
\leq (1+\varepsilon)(1+5\varepsilon)A_j\cdot D_j+5\varepsilon(1+\varepsilon)\frac{2n^2}{W}
\]

The last inequality follows from $A_i\leq 2n$. Furthermore,
\[
\hat{A}_j\cdot\overline{D}_j\geq(1-\varepsilon)A_j\left(D_j-5\varepsilon D_j-5\varepsilon\frac{n}{W}\right)
\geq
(1-\varepsilon)(1-5\varepsilon)A_j\cdot D_j-5\varepsilon(1-\varepsilon)\frac{2n^2}{W}
\]

Therefore, \[
\abs{\hat{A}_j\cdot\overline{D}_j-A_j\cdot D_j}
\leq (6\varepsilon+5\varepsilon^2) A_j\cdot D_j+10\varepsilon(1+\varepsilon)\frac{n^2}{W}
\leq 11\varepsilon A_j\cdot D_j+20\varepsilon\frac{n^2}{W}.
\]

Combining the analysis above,
we have that for any $j\leq W$, it holds that 
\[\abs{\hat{A}_j\cdot\overline{D}_j-A_j\cdot D_j}\leq 31\varepsilon A_j\cdot D_j+
20\varepsilon\frac{n^2}{W}.
\]
By \Cref{fact:lbcostsim-n2}, $\cost^{(s)}(G)\geq\frac{n(n-1)}{2}$, and when $n\geq 2$, 
$\cost^{(s)}(G)\geq\frac{1}{4}n^2$.
Besides, $\cost^{(s)}(G)=\frac12\sum_{j=1}^W (c_j^{(s)}+n-1)(n-c_j^{(s)})=\frac{1}{2}\sum_{j=1}^W A_j\cdot D_j$.
Now we can bound the error of $\widehat{\cost^{(s)}}(G)$ by
\[
\abs{\widehat{\cost^{(s)}}(G)-\cost^{(s)}(G)}
\leq\frac{1}{2}\sum_{j=1}^W\abs{\hat{A}_j\cdot\overline{D}_j-A_j\cdot D_j}
\leq \frac{31\varepsilon}{2}\sum_{j=1}^W A_j\cdot D_j+20\varepsilon n^2
\leq 120\varepsilon\cost^{(s)}(G)
\]

Replacing $\varepsilon$ with $\varepsilon/120$ achieves a $(1+\varepsilon)$ approximation factor.

\paragraph{Running time analysis.}
According to \cref{cor:numberccappdelta}, each $\hat{c}^{(s)}_j$ is obtained in $O(\frac{d}{\varepsilon^2}\log(\frac{d}{\varepsilon})\cdot\log(\frac{1}{\delta}))=O(\frac{d}{\varepsilon^2}\log(\frac{d}{\varepsilon})\cdot\log W)$ time,
and thus the the sequence of estimates $\{\hat{c}^{(s)}_1,\dots,\hat{c}^{(s)}_W\}$ are obtained in $O(\frac{Wd}{\varepsilon^2}\log(\frac{d}{\varepsilon})\cdot\log W)$ time.

Besides, \cref{alg:appcost_sim} invokes \textsc{BinarySearch} for 
$t=O(\log(W/\varepsilon)/\varepsilon)$ search keys, and each invocation of \textsc{BinarySearch} accesses $O(\log W)$ estimates $\widehat{D}_j$.
Thus, the algorithm accesses at most $O(\log^2(\frac{W}{\varepsilon})/\varepsilon)$ estimates $\widehat{D}_j$ in $\widehat{D}$.
According to \cref{cor:appncc_sim_delta}, each estimate of $\widehat{D}_j$ can be obtained in 
$O(\frac{W}{\varepsilon^2}d\log(\frac{W}{\varepsilon}d)\cdot\log(1/\delta))=O(\frac{Wd}{\varepsilon^2}\log^2(\frac{Wd}{\varepsilon}))$ time.

Thus, the running time of \cref{alg:appcost_sim} is:
\[
O(\frac{Wd}{\varepsilon^2}\log(\frac{d}{\varepsilon})\cdot\log W)
+O(\log^2(\frac{W}{\varepsilon})/\varepsilon)\cdot O(\frac{Wd}{\varepsilon^2}\log^2(\frac{Wd}{\varepsilon}))
=O(\frac{Wd}{\varepsilon^3}\log^4(\frac{Wd}{\varepsilon}))
=\Tilde{O}(\frac{Wd}{\varepsilon^3})
\]

\end{proof}

\subsection{Estimating the Profile Vector} \label{subsec:cost_k_sim}
Recall that $\cost^{(s)}(G)=\sum_{i=1}^{n}\cost_k^{(s)}$, where 
$\cost_k^{(s)}=\sum_{i=1}^{n-k}w_i^{(s)}$ and $w_i^{(s)}$ is the $i$-th largest weight in $\mathrm{MaxST}$. We now give a sublinear time algorithm to approximate the SLC profile vector
$(\cost_1^{(s)},\dots,\cost_n^{(s)})$ and prove the following theorem.

\similarityprofile*

We will first give a formula to calculate $\cost_k^{(s)}$, which is slightly
different from the one for $\cost_k$ in the distance case. Next we give an algorithm to estimate the profile,
and the corresponding analysis.

\paragraph{Formulas of Profile}
Recall that $w_i^{(s)}$ is the weight of the $i$-th maximum edge's weight, in $\mathrm{MaxST}$.
We can derive $\cost_k^{(s)}$ as follows:
\begin{align}
    \cost_k^{(s)}&=\sum_{i=1}^{n-k}w_i^{(s)} \nonumber\\
    &=n_W\cdot W+n_{W-1}\cdot(W-1)+\dots 
    +(n-k-n_W-\dots-n_{w_{n-k}^{(s)}+1})\cdot w_{n-k}^{(s)}\nonumber\\
    &=(c_{W+1}^{(s)}-c_W^{(s)})\cdot W+(c_W^{(s)}-c_{W-1}^{(s)})\cdot(W-1)+
    \cdots+ (n-k-c_{W+1}^{(s)}+c_{w_{n-k}^{(s)}+1}^{(s)})\cdot w_{n-k}^{(s)} 
    \tag{since $n_j=c_{j+1}^{(s)}-c_j^{(s)}$} \nonumber\\
    &=(n-c_W^{(s)})\cdot W+(c_W^{(s)}-c_{W-1}^{(s)})\cdot(W-1)+
    \cdots+ (c_{w_{n-k}^{(s)}+1}^{(s)}-k)\cdot w_{n-k}^{(s)} 
    \tag{since $c_{W+1}^{(s)}=n$} \nonumber\\
    &=n\cdot W-c_W^{(s)}-c_{W-1}^{(s)}+\cdots-c_{w_{n-k}^{(s)}+1}^{(s)}
    -k\cdot w_{n-k}^{(s)} \nonumber\\
    &=\sum_{j=w_{n-k}^{(s)}+1}^{W}(n-c_j^{(s)})+(n-k)\cdot w_{n-k}^{(s)} \label{eqn:cost_k_sim}
\end{align}

We can derive the following lemma, which is similar to \cref{lem:costk}.
The main idea of the lemma is, for a given weight $j$, 
we approximate the rank of the first edge with weight $j$ in $\mathrm{MaxST}$. 
Then, we can take $k$ as its rank, 
and calculate $\cost_k^{(s)}$.

\begin{lemma}\label{lem:costk_sim}
Given any integer $j\in\{1,...,W\}$ 
and let $k=c_{j}^{(s)}$, where $c_{j}^{(s)}$ 
is the number of connected components in $G_j^{(s)}$.
Then 
$$
\cost_k^{(s)}=\sum_{i=j+1}^{W}(n-c^{(s)}_i)+(n-c^{(s)}_j)\cdot j.
$$
\end{lemma}
\begin{proof}
    Let $n_j$ be the number of edges in $\mathrm{MaxST}$ with weight $j$, then
    $n_j=c_{j+1}^{(s)}-c_j^{(s)},\ \forall 1\leq j\leq W$, and $c_{W+1}=n$. Thus, the number of edges in $\mathrm{MaxST}$ with weights at least $j$ is:
    \[n_W+n_{W-1}+\dots+n_j=(n-c_W^{(s)})+(c_W^{(s)}-c_{W-1}^{(s)})
    +\dots+(c_{j+1}^{(s)}-c_{j}^{(s)})=n-c_{j}^{(s)}\]

    As $k=c^{(s)}_j$, the edge ordering $n-k=n-c^{(s)}_j$ in $\mathrm{MaxST}$ has weight $w^{(s)}_{n-k}=j$. Then according to \cref{eqn:cost_k_sim}, we have
    \[
    \cost_k^{(s)}=\cost_{c_{j}}^{(s)}(G)=
    \sum_{i=j+1}^{W}(n-c^{(s)}_i)+(n-c^{(s)}_j)\cdot j
    \]
\end{proof}

\paragraph{Algorithm to Estimate Profile}
Similarly to the distance  case, we first round each $ c_j $ for $ 1 \leq j \leq W $ to its nearest interval, such as $ (B_{i+1}, B_i] $. This allows us to estimate $ \cost_k $ for $ k $ corresponding to all the interval endpoints. The detailed algorithm is given in \cref{alg:appprofile_similarity}.

\begin{algorithm}[h!]
    \DontPrintSemicolon
    \caption{\textsc{AppProfileSim}($G, \varepsilon,d,W$)}
    \label{alg:appprofile_similarity}
    get sequence 
    $\{\hat j_1,\dots,\hat j_t\}$ from \textsc{AppCostSim($G,\varepsilon,d,W$)}\;
    set $t$ according to \cref{def:interval_similarity}\;
    {set $\overline{\cost}_{n-B_t}=0$}\;
    \For{$i=\{1,\dots,{t-1}\}$}{
        set $B_i$ according to \cref{def:interval_similarity}\;
        $\overline{\cost^{(s)}}_{n-B_i}=B_i\cdot (\hat j_i-1)+
        \sum_{k=i}^{t-1}(\hat j_{k+1}-\hat j_{k})\cdot B_{k}$\;
    }
    Output estimated vector
    $(\overline{\cost^{(s)}}_{n-B_1},\dots,\overline{\cost^{(s)}}_{n-B_t})$\end{algorithm}

We note that $B_i$ serves as an approximation for $D_{\hat j_i} = n - c_{\hat j_i}$. In \cref{lem:costk_sim}, if we substitute $k$ with $n - \overline{D}_{\hat j_i} = n - B_i$ and $j$ with $\hat{j}_i$, we arrive at the following expression:
\[
\overline{\cost^{(s)}}_{n - B_i} 
= \sum_{j=\hat{j}_i+1}^W \overline D_j + \overline D_{\hat j_i} \cdot \hat j_i
= B_i \cdot (\hat j_i-1) + \sum_{k=i}^{t-1} (\hat j_{k+1} - \hat j_k) \cdot B_{k}.
\]

In the last equation, we have $\hat j_i-1$ because $\overline D_{\hat j_i}$ is considered in $\hat J_i$,
which should be excluded in the summation $\sum_{j=\hat j_i+1}^W \overline{D}_j$.
Once we have such a succinct representation, we can construct an oracle, that given any $k\in[1,n]$, 
and $n-B_{i}\le k< n-B_{i+1}$, we let $\widehat{\cost^{(s)}}_k=\overline{\cost^{(s)}}_{B_{i}}$.

\begin{algorithm}[h!]
    \DontPrintSemicolon
    \caption{\textsc{ProfileOracleSim}($k,\{B_1,\dots,B_t\}, (\overline{\cost^{(s)}}_{n-B_1},\dots,\overline{\cost^{(s)}}_{n-B_t})$)}
    \label{alg:appprofile_profile_similarity}
    {define $B_{t+1}=-\infty$}\;
    use binary search over $(B_1,\dots,B_t)$, and find the index $i$ such that $n-B_{i}\le k<n-B_{i+1}$\;
    output $\widehat{\cost^{(s)}}_k:=\overline{\cost^{(s)}}_{n-B_{i}}$\;
\end{algorithm}

We call the vector $(\overline{\cost^{(s)}}_{n-B_1},\dots,\overline{\cost^{(s)}}_{n-B_t})$ a \emph{succinct representation} of the vector $(\widehat{\cost^{(s)}}_1,\dots,\widehat{\cost^{(s)}}_{n})$, and we call \cref{alg:appprofile_profile_similarity} a \emph{profile oracle}, as it can take as input any index $k$ and answer $\widehat{\cost^{(s)}}_k$.

\paragraph{Analysis of \cref{alg:appprofile_similarity} and \cref{alg:appprofile_profile_similarity}}
Now we analyze \cref{alg:appprofile_similarity} and \cref{alg:appprofile_profile_similarity}. We will show that the vector $(\widehat{\cost^{(s)}}_1,\dots,\widehat{\cost^{(s)}}_{n})$ is a good approximation of the profile vector $(\cost^{(s)}_1,\dots,\cost^{(s)}_n)$. We first prove an error bound for each individual $\widehat{\cost^{(s)}}_k$, and then bound the sum of the error and give the proof of \cref{thm:profile_sim}.

\begin{fact}\label{fact:B_i-1_i+1_sim}
    For any $2\le i\le t_1+2t_2-1$, $B_{i-1}-B_{i+1}\le 4\varepsilon B_{i+1}$.
\end{fact}

\begin{proof}
    We develop the following proof based on \cref{def:interval_similarity} and \cref{fact:interval_similarity}.
    \begin{itemize}
        \item When $2\le i\le t_1-1$, 
        \[
        B_{i-1}-B_{i+1}=(n-(i-1)\frac{\varepsilon n}{W})-(n-(i+1)\frac{\varepsilon n}{W})=2\varepsilon\frac{n}{W},
        \]
        
        Since $B_{i+1}\ge B_{t_1}\ge n-\frac{n}{W}$, and when $W\ge 2$, $n-\frac{n}{W}\ge \frac{n}{W}$, thus $B_{i-1}-B_{i+1}\le 2\varepsilon(n-\frac{n}{W})\le 2\varepsilon B_{i+1}$.
        \item When $i=t_1$, $B_{i-1}-B_{i+1}=B_{t_1-1}-B_{t_1+1}=B_{t_1}+\frac{\varepsilon n}{W}-B_{t_1+1}$. From \cref{fact:interval_similarity}, $B_{t_1}<n-\frac{n}{W}(1-\varepsilon)$. Then,
        \[
        B_{i-1}-B_{i+1}=B_{t_1}+\frac{\varepsilon n}{W}-B_{t_1+1}<(n-\frac{n}{W}(1-\varepsilon))+\frac{\varepsilon n}{W}-(n-(1+\varepsilon)\frac{n}{W})=3\frac{\varepsilon n}{W},
        \]
    
        Since $B_{t_1+1}\ge B_{t_1+2t_2}\ge \frac{n}{W}$, we have $B_{t_1-1}-B_{t_1+1}\le 3\varepsilon B_{t_1+1}$.

        \item When $i=t_1+1$, $B_{i-1}-B_{i+1}=B_{t_1}-B_{t_1+2}<(n-\frac{n}{W}(1-\varepsilon))-(n-(1+\varepsilon)^2\frac{n}{W})=(3\varepsilon+\varepsilon^2)\frac{n}{W}<4\varepsilon\frac{n}{W}$, as $B_{t_1}<n-\frac{n}{W}(1-\varepsilon)$ and $\varepsilon<1$. Since $B_{t_1+2}\ge B_{t_1+2t_2}\ge \frac{n}{W}$, we have $B_{t_1}-B_{t_1+2}\le 4\varepsilon B_{t_1+2}$.
    
        \item When $t_1+2\le i\le t_1+t_2-1$, 
        \[
        B_{i-1}-B_{i+1}=(n-(1+\varepsilon)^{i-1-t_1}\frac{n}{W})-(n-(1+\varepsilon)^{i+1-t_1}\frac{n}{W})=((1+\varepsilon)^2-1)\cdot(1+\varepsilon)^{i-1-t_1}\frac{n}{W},
        \]
    
        Since $\varepsilon<1$ and $B_{i-1}>B_{i+1}$, we have $B_{i-1}-B_{i+1}\le (2\varepsilon+\varepsilon^2)(n-B_{i-1})<3\varepsilon(n-B_{i+1})$. As $B_{i+1}\ge B_{t_1+t_2}\ge\frac{n}{2}$, $B_{i-1}-B_{i+1}\le 3\varepsilon B_{i+1}$.

        \item When $i=t_1+t_2$, $B_{i-1}-B_{i+1}=B_{t_1+t_2-1}-B_{t_1+t_2+1}$. By \cref{def:interval_similarity}, $n-B_{t_1+t_2-1}=(1+\varepsilon)(n-B_{t_1+t_2})$. Besides, from \cref{fact:interval_similarity}, $n-B_{t_1+t_2}>\frac{n}{2(1+\varepsilon)}$. Then we have $B_{t_1+t_2-1}<n-(1+\varepsilon)\frac{n}{2(1+\varepsilon)}=\frac{n}{2}$. Then, $B_{t_1+t_2-1}-B_{t_1+t_2+1}<\frac{n}{2}-\frac{n}{2(1+\varepsilon)}=\varepsilon\frac{n}{2(1+\varepsilon)}=\varepsilon B_{t_1+t_2+1}$. Thus, $B_{i-1}-B_{i+1}<\varepsilon B_{i+1}$.

        \item When $i=t_1+t_2+1$, since $n-B_{t_1+t_2}>\frac{n}{2(1+\varepsilon)}$ and $\varepsilon<1$, we have
        \[
        B_{i-1}-B_{i+1}=B_{t_1+t_2}-B_{t_1+t_2+2}<(n-\frac{n}{2(1+\varepsilon)})-\frac{n}{2(1+\varepsilon)^2}=(3\varepsilon+\varepsilon^2)\frac{n}{2(1+\varepsilon)^2}<4\varepsilon B_{t_1+t_2+2}=4\varepsilon B_{i+1}.
        \]
    
        \item When $t_1+t_2+2\le i\le t_1+2t_2-1$, 
        \[
        B_{i-1}-B_{i+1}=\frac{n}{2(1+\varepsilon)^{i-1-(t_1+t_2)}}-\frac{n}{2(1+\varepsilon)^{i+1-(t_1+t_2)}}=((1+\varepsilon)^2-1)B_{i+1}=(2\varepsilon+\varepsilon^2)B_{i+1}\le 3\varepsilon B_{i+1}.
        \]
    \end{itemize}

    In a conclusion, for any $2\le i\le t_1+2t_2-1$, $B_{i-1}-B_{i+1}\le 4\varepsilon B_{i+1}$.
\end{proof}

The following lemma provides an error bound for each individual $\widehat{\cost^{(s)}}_k$.  

\begin{lemma}\label{lem:single_costk_error_sim}
    Assume that $W\le n$ and let $\varepsilon<1$. With probability at least $3/4$, for any integer $1\le k \le n-1$, \cref{alg:appprofile_profile_similarity} returns an estimate $\widehat{\cost^{(s)}}_k$ for the $k$-clustering cost in similarity case, i.e., ${\cost^{(s)}}_k$, such that \[
    |\widehat{\cost^{(s)}}_k-{\cost^{(s)}}_k|\le 30\varepsilon\cdot\max\{\cost^{(s)}_k,n\}
    \]

    Given the succinct representation $(\overline{\cost^{(s)}}_{B_1},\dots,\overline{\cost^{(s)}}_{B_t})$, the running time of the algorithm is $O(\log (\frac{\log W}{\varepsilon}))$.
\end{lemma}

\begin{proof}
From \cref{lem:bound_of_bar_D}, we have $\abs{\overline D_j-D_j}\leq5\varepsilon\cdot\max\left\{\frac{n}{k},\min\{D_j,n-D_j\}\right\}$ 
for all $j\in[W]$, with probability at least $\frac{7}{8}$. We assume that this event holds in the remaining part of the proof.

For simplicity of notation, we use $j_i$ to denote $\hat{j}_i$ in the following of the proof.
By definition of $\overline{\cost^{(s)}}_{n-B_{i}}$ and $\cost^{(s)}_{c_{j_{i}}}(G)$, the gap between them is,
\begin{align*}
    |\overline{\cost^{(s)}}_{n-B_{i}}-\cost^{(s)}_{c_{j_{i}}}|
    &=\abs{(\sum_{j=j_i+1}^W \overline{D}_j+\overline{D}_{j_i}\cdot j_i)-(\sum_{j=j_i+1}^W D_j+D_{j_i}\cdot j_i)}\\
    &\leq\sum_{j=j_i+1}^{W}|\overline D_j-D_j|+|\overline{D}_{j_i}-D_{j_{i}}|\cdot j_{i}\\
    &\leq 5\varepsilon\sum_{j=j_i+1}^{W}D_j+5\varepsilon\cdot D_{j_{i}}\cdot j_{i}\\
    &=5\varepsilon\cdot\cost^{(s)}_{c_{j_i}}
\end{align*}

For $n-B_{i}\le k< n-B_{i+1}$, where $1\le i \le t-1$, we estimate $\cost_k$ using $\overline{\cost^{(s)}}_{n-B_{i}}$, and thus we have

\begin{align*}
    |\widehat{\cost^{(s)}}_k-{\cost^{(s)}}_k|&\leq
    |\overline{\cost^{(s)}}_{n-B_{i}}-\cost^{(s)}_{c_{j_{i}}}|
    +|\cost^{(s)}_{c_{j_{i}}}-\cost^{(s)}_k|\\
    &\leq 5\varepsilon\cdot\cost^{(s)}_{c_{j_i}} + |\cost^{(s)}_{c_{j_{i}}}-\cost^{(s)}_k|\\
    &\le 5\varepsilon(\cost^{(s)}_k+|\cost^{(s)}_{c_{j_{i}}}-\cost^{(s)}_k|)+|\cost^{(s)}_{c_{j_{i}}}-\cost^{(s)}_k|\\
    &\le 5\varepsilon\cdot\cost^{(s)}_k + 6|\cost^{(s)}_{c_{j_{i}}}-\cost^{(s)}_k| \tag{since $\varepsilon<1$}
\end{align*}

Thus, we only need to bound $|\cost^{(s)}_{c_{j_{i}}}-\cost^{(s)}_k|$.

Case (I): $c^{(s)}_{j_i}\le k$. In this case, $\cost^{(s)}_k\le \cost^{(s)}_{c_{j_{i}}}$ and $j_i\le w_{n-k}$. Since $\cost^{(s)}_k=\sum_{j=1}^{n-k}w_j^{(s)}$,
\[
|\cost^{(s)}_{c_{j_{i}}}-\cost^{(s)}_k|=\cost^{(s)}_{c_{j_{i}}}-\cost^{(s)}_k=\sum_{j=n-k+1}^{n-c^{(s)}_{j_i}}w_j^{(s)}\le(k-c^{(s)}_{j_i})\cdot w^{(s)}_{n-k}
\]

From \cref{lem:bound_of_xj}, when $i\le t-2$, $B_{i+2}\le D_{j_i}=n-c^{(s)}_{j_i}\le B_{i-1}$, then $n-B_{i-1}\le c^{(s)}_{j_i}\le n-B_{i+2}$, and thus $k-c^{(s)}_{j_i}\le B_{i-1}-B_{i+1}$. When $2\le i\le t_1+2t_2-1$, according to \cref{fact:B_i-1_i+1_sim}, $B_{i-1}-B_{i+1}\le 4\varepsilon B_{i+1}\le 4\varepsilon(n-k)$; when $t_1+2t_2\le i\le t-1$, according to \cref{def:interval_similarity}, $B_{i-1}-B_{i+1}=2\varepsilon\frac{n}{W}>2\varepsilon(n-k)$. Thus, $|\cost^{(s)}_{c_{j_{i}}}-\cost^{(s)}_k|\le 4\varepsilon\cdot \max\{n-k,\frac{n}{W}\}\cdot w^{(s)}_{n-k}$. Since $\cost^{(s)}_k=\sum_{j=w^{(s)}_{n-k}+1}^W(n-c_j^{(s)})+(n-k)\cdot w^{(s)}_{n-k}\ge (n-k)\cdot w^{(s)}_{n-k}$ and $w^{(s)}_{n-k}\le n$, we have,
\[
|\widehat{\cost^{(s)}}_k-{\cost^{(s)}}_k|\leq 5\varepsilon\cdot\cost^{(s)}_k+6\cdot4\varepsilon\cdot \max\{n-k,\frac{n}{W}\}\cdot w^{(s)}_{n-k}\le 30\varepsilon\cdot\max\{\cost^{(s)}_k,n\}
\]

Case(II): $k\le c^{(s)}_{j_i}$. In this case, $\cost^{(s)}_{c_{j_{i}}}\le \cost^{(s)}_k$ and $w_{n-k}\le j_i$. Then,
\[
|\cost^{(s)}_{c_{j_{i}}}-\cost^{(s)}_k|=\cost^{(s)}_k-\cost^{(s)}_{c_{j_{i}}}=\sum_{j=n-c^{(s)}_{j_i}+1}^{n-k}w_j^{(s)}\le(c^{(s)}_{j_i}-k)\cdot j_i \le (B_i-B_{i+2})\cdot j_i
\]

When $1\le i\le t_1+2t_2-2$, according to \cref{fact:B_i-1_i+1_sim}, $B_i-B_{i+2}\le 4\varepsilon B_{i+2}<4\varepsilon B_{i-1}\le 4\varepsilon (n-c^{(s)}_{j_i})$; when $t_1+2t_2\le i\le t-2$, according to \cref{def:interval_similarity}, $B_i-B_{i+2}=2\varepsilon\frac{n}{W}>2\varepsilon(n-c^{(s)}_{j_i})$.
From \cref{lem:costk_sim}, we know that $\cost^{(s)}_{c_{j_{i}}}\ge (n-c^{(s)}_{j_i})$, and thus
\[
|\widehat{\cost^{(s)}}_k-{\cost^{(s)}}_k|\leq 5\varepsilon\cdot\cost^{(s)}_k+6\cdot 4\varepsilon\cdot\max\{n-c^{(s)}_{j_i},\frac{n}{W}\}\cdot j_i\le 30\varepsilon\cdot\max\{\cost^{(s)}_k,n\}
\]

In a conclusion, we have $|\cost^{(s)}_k-\widehat{\cost^{(s)}}_k|\leq 30\varepsilon\cdot\max\{\cost^{(s)}_k,n\}$.

\paragraph{Running time analysis.}
Given the succinct representation, \cref{alg:appprofile_profile_similarity} can be implemented in $O(\log t)=O(\log(\frac{\log W}{\varepsilon}))$ time for any $k$, as one can simply perform binary search to find the right index $i$ with $n-B_i\leq k<n-B_{i+1}$. 
\end{proof}

With \cref{lem:single_costk_error_sim} established, we are ready to bound $\sum_{k=1}^{n}|\widehat{\cost^{(s)}}_k-{\cost^{(s)}}_k|$. 

\begin{proof}[Proof of \cref{thm:profile_sim}]
From the error bound for individual $\widehat{\cost^{(s)}}_k$ in \cref{lem:single_costk_error_sim}, we have that \[
\sum_{k=1}^{n}|\widehat{\cost^{(s)}}_k-{\cost^{(s)}}_k|\leq \sum_{k=1}^{n} 30\varepsilon\cdot\max\{\cost^{(s)}_k,n\}\le 30\varepsilon\cdot\max\left\{\sum_{k=1}^{n}\cost^{(s)}_k,n^2\right\}
\]

According to \cref{fact:lbcostsim-n2}, $n^2\le 4\cost^{(s)}(G)$, and thus $\sum_{k=1}^{n}|\widehat{\cost^{(s)}}_k-{\cost^{(s)}}_k|\le 120\varepsilon\cdot \cost^{(s)}(G)$.

Replacing $\varepsilon$ with $\varepsilon/120$, we get an succinct representation of $(\widehat{\cost^{(s)}}_1,\dots,\widehat{\cost^{(s)}}_n)$ such that each $\widehat{\cost^{(s)}}_k$ is a $(1+\varepsilon)$-estimator \emph{on average}.

\paragraph{Running time analysis.}
Note that we first invoke \cref{alg:appcost_sim} to obtain the sequence $\{j_0,j_1,\dots,j_t\}$,
which takes $O(\frac{Wd}{\varepsilon^3}\log^4(\frac{Wd}{\varepsilon}))$ time.
Given the sequence, to estimate each $\overline{\cost^{(s)}}_{n-B_i}$, 
we need to calculate $\sum_{k=i}^{t-1}(j_{k+1}-j_k)\cdot B_k$
in $O(t)$ time, for all $i\in[1,t]$.
Thus, getting the estimated vector $(\overline{\cost^{(s)}}_{n-B_1},\dots,\overline{\cost^{(s)}}_{n-B_t})$
can be done in $O(t^2)=O(\log^2W/\varepsilon^2)$ time.
Thus, the total running time of \cref{alg:appprofile_similarity} is $O(\frac{Wd}{\varepsilon^3}\log^4(\frac{Wd}{\varepsilon}))=\Tilde{O}(\frac{{W}}{\varepsilon^3}d)$. 
\end{proof}

\section{Lower Bounds}\label{sec:lowerbounds}
In this section, we give the proofs of the lower bounds on the query complexities for estimating the SLC costs $\cost(G)$ in the distance case and $\cost^{(s)}(G)$ in the similarity case. The results hold for the model of computation we introduced earlier, which requires $\Theta(i)$ time to access the $i$-th neighbor in the adjacency list, as well as for a model where we can query for the degree of a vertex and for its $i$-th neighbor in $O(1)$ time.
The proofs built upon the query lower bounds for estimating the number of connected components and the weight of the minimum spanning tree in \citep{chazelle2005approximating} and can be seen as an adaptation of their proof to our setting (that is, we need to adjust them at various places to our objective function).

We first introduce a useful tool. For any $ 0 < q \leq 1/2 $ and $ t = 0, 1 $, let $ \mathcal{D}_{q}^t $ denote the distribution over $ \{0, 1\} $ that assigns $ 1 $ with probability $ q_t = q(1 + (-1)^t \varepsilon) $. We define a distribution $ \mathcal{D} $ on $ n $-bit strings as follows: (1) choose $ t = 1 $ with probability $ 1/2 $ and $ t = 0 $ otherwise; (2) generate a random string $b$ from $ \{0, 1\}^n $ by independently selecting each bit $ b_i $ from the distribution $ \mathcal{D}_q^t $. The following result is shown in \citep{chazelle2005approximating}. 
\begin{lemma}[Lemma 9 in \citep{chazelle2005approximating}]\label{lem:CRT_lowerlemma}
Given query access to a random bit string generated from the above distribution $\mathcal{D}$, any algorithm that determines the value of $t$ with a success probability of at least $3/4$ requires $\Omega(\varepsilon^{-2}/q)$ bit queries on average. 
\end{lemma}

Now we use the above lemma to prove our first lower bound.

\distancelowerbound*
\begin{proof}
The proof follows closely a corresponding proof from \cite{chazelle2005approximating}
while adapting it to our setting.
We first define two families of weighted graphs, $ \mathcal{P}_0 $ and $ \mathcal{P}_1 $. Each graph in $ \mathcal{P}_i $ (where $ i=0,1 $) is a path consisting of $ n $ vertices. In the family $ \mathcal{P}_0 $, we generate a random $(n-1)$-bit string $ b_1 \dots b_{n-1} $, with each bit independently drawn from $ \mathcal{D}_q^{0} $, where $ q=1/\sqrt{W-1} $. We assign a weight of $ W $ (or $ 1 $) to the $ i $-th edge along the path if $ b_i=1 $ (or $ 0 $, respectively). The construction for the family $ \mathcal{P}_1 $ is similar, except that each bit is drawn from $ \mathcal{D}_q^{1} $.

Consider a graph $G$ from $\mathcal{P}_0\cup \mathcal{P}_1$. Let $T_W$ be the number of edges of weight $W$ in $G$. We note that the SLC cost of $G$ in distance case is 
\begin{align*}
\cost(G)&=\sum_{i=1}^{n-1}(n-i)\cdot w_i=\sum_{i=1}^{n-1-T_W}(n-i)\cdot 1+\sum_{i=n-T_W}^{n-1}(n-i)\cdot W=\sum_{i=T_W+1}^{n-1}i+\sum_{i=1}^{T_W}i\cdot W\\
&=\frac{(T_W+1+n-1)(n-1-T_W)}{2} + W \cdot \frac{(T_W+1)T_W}{2}\\
&=\frac{n(n-1)-T_W-T_W^2+W\cdot T_W^2+W\cdot T_W}{2}\\
&=\frac{n(n-1)+(W-1)(T_W^2+T_W)}{2}
\end{align*}

Now observe that $\E[T_W]=\frac{n-1}{\sqrt{W-1}}(1+\varepsilon)$ or $\E[T_W]=\frac{n-1}{\sqrt{W-1}}(1-\varepsilon)$, depending on if $G\in \mathcal{P}_0$ or $G\in \mathcal{P}_1$. It further holds that 
$\E[\cost(G)]=\Theta(n^2+(W-1)\cdot \frac{(n-1)^2}{W-1})=\Theta(n^2)$.

\textbf{Case I.} If $G\in \mathcal{P}_0$, $\E[T_W]=\frac{n-1}{\sqrt{W-1}}(1+\varepsilon)$, then 
By Chernoff Bound in \Cref{thm:chernoff}, 
\[
\Pr[T_W<(1-\frac{\varepsilon}{2})\E[T_W]]\leq 
\mathrm{e}^{-\frac{\varepsilon^2(1+\varepsilon)(n-1)}{8\sqrt{W-1}}}
\leq \mathrm{e}^{-\frac{\varepsilon^2 n}{8\sqrt{W}}}.
\]
Since $\varepsilon\leq\frac{1}{2}$, it holds that with probability $1-\mathrm{e}^{-\varepsilon^2 n/8\sqrt{W}}$, $T_W\geq \frac{n-1}{\sqrt{W-1}}(1+\varepsilon)(1-\frac{\varepsilon}{2})
    \geq \frac{n-1}{\sqrt{W-1}}(1+\frac{\varepsilon}{4})$. Now as $\varepsilon>\frac{W^{1/4}}{\sqrt{40 n}}$, we have that with  probability at least $1-\mathrm{e}^{-\frac{40}{24}}\geq 0.8$,
\[\cost(G)\geq \frac{n(n-1)}{2}+\frac{(n-1)^2}{2}(1+\frac{\varepsilon}{4})^2+\frac{(n-1)\sqrt{W-1}}{2}(1+\frac{\varepsilon}{4}).\]

\textbf{Case II.} If $G\in \mathcal{P}_1$, $\E[T_W]=\frac{n-1}{\sqrt{W-1}}(1-\varepsilon)$ and then 
\[
\Pr[T_W>(1+\frac{\varepsilon}{2})\E[T_W]]\leq 
\mathrm{e}^{-\frac{\varepsilon^2(1-\varepsilon) (n-1)}{12\sqrt{W-1}}}
\leq \mathrm{e}^{-\frac{\varepsilon^2 n}{24\sqrt{W}}}.
\]

Similar as above, 
with probability at least $1-\mathrm{e}^{-\varepsilon^2 n/24\sqrt{W}}$, it holds that 
    $T_W\leq \frac{n-1}{\sqrt{W-1}}(1-\varepsilon)(1+\frac{\varepsilon}{2})
    \leq \frac{n-1}{\sqrt{W-1}}(1-\frac{\varepsilon}{2})$. By the fact that $\varepsilon>\frac{W^{1/4}}{\sqrt{40 n}}$, with probability at least $1-\mathrm{e}^{-\frac{40}{24}}\geq 0.8$, it holds that
\[\cost(G)\leq \frac{n(n-1)}{2}+\frac{(n-1)^2}{2}(1-\frac{\varepsilon}{2})^2+\frac{(n-1)\sqrt{W-1}}{2}(1-\frac{\varepsilon}{2}).\]

Since $\varepsilon\leq\frac{1}{2}$, $2-\frac{\varepsilon}{4}\geq\frac{15}{8}\geq\frac{4}{3}$.
So the gap between the above two bounds is
\[
\frac{(n-1)^2}{2}(2-\frac{\varepsilon}{4})\frac{3\varepsilon}{4}
+\frac{(n-1)\sqrt{W-1}}{2}\frac{3\varepsilon}{4}
\geq \varepsilon\frac{(n-1)^2}{2}.
\]

Thus, any algorithm that estimates $\cost(G)$ with a relative error of $\varepsilon/C$ for some large constant $C$, can be used to distinguish if $G\in\mathcal{P}_0$ or $G\in\mathcal{P}_1$. By \Cref{lem:CRT_lowerlemma}, any algorithm that solves the latter problem requires $\Omega(\sqrt{W}/\varepsilon^2)$ queries on average. 

Now we extend our result to the case of graphs with arbitrary average degree $d\ge 2-2/n$. Let $d = 2 m/n$, i.e. our graph is supposed to have $m$ edges.
In order to obtain a graph with average degree $d$ we add edges to the vertices of our path in such a way that every vertex has degree $\lceil d \rceil $ or $\lfloor d \rfloor$ (this can be done by greedily adding edges between the two vertices with smallest degree).
All edges that have been newly added receive a weight of $W$. We observe that this does not change the weight of the minimum spanning tree and that this weight is still given by the weight of the edges of the initial path.
For each vertex $v$ of degree $d_v$
we then choose a random permutation of $d_v$ elements uniformly at random and put the adjacency list in the corresponding order. Let us call the resulting graph $H$ and the initial path $G$.

We will argue that if an algorithm makes $q$ queries to $H$ then in expectation it queries $O(q/d)$ edges from $G$. 
Observe that, conditioned on an arbitrary history of queries and answers, the probability that a query to a previously unqueried neighbor of vertex $v$
yields an edge from $G$ is at most $2/i$, where $i$ denotes the number of unqueried neighbors of $v$. 

To simplify the analysis, we count only the first $\lfloor d/2 \rfloor$ 
queries to a vertex; any remaining edges are revealed for free.
The cost for revealing these additional edges can be charged to the first $\lfloor d/2 \rfloor$ queries. 
As a result, any vertex with unknown neighbors is still in the phase where queries are being charged, and thus must have at least $\lfloor d/2 \rfloor$ unknown neighbors. 
Therefore, conditioned on the query history, the probability that a query reveals an edge from the original path $G$ is at most $2/\lfloor d/2\rfloor$. 
By linearity of expectation, the expected number of such queried edges is at most $2q/\lfloor d/2\rfloor + R$, 
where $R = O(q/d)$ is the number of edges revealed for free. 
Thus, the expected number of revealed edges is at most $C' q/d$ for some constant $C'>0$.

Now consider an arbitrary algorithm that computes a $(1+\varepsilon)$-approximation of
$\cost(G)$ using $q$ queries in expectation. Then this algorithm makes at most $
C' q/ d$ queries to the path on average. We can then use this algorithm to solve the problem on $G$ using $C' q/d$ queries on average (by first building $H$ as described above and then querying $G$ whenever we query an edge of the path in $H$). However,
we know that the latter problem requires $\Omega(\sqrt{W}/\varepsilon^2)$ queries on average as proven above. Thus, the algorithm has to make $\Omega(d\sqrt{W}/\varepsilon^2)$ queries on $H$. 
\end{proof}

Now we consider the lower bound for the similarity case and prove the following theorem. 

\similaritylowerbound*

\begin{proof}
We construct two families of graphs, $\mathcal{P}_0$ and $\mathcal{P}_1$, in the same manner as described in the proof of \cref{thm:cost_distance_lowerbound}, with the key difference that we now set $q = \frac{1}{W - 1}$.

Now suppose that $G\in \mathcal{P}_0\cup \mathcal{P}_1$. Let $T_W$ be the number of edges of weight $W$ in $G$. We note that the SLC cost of the similarity graph $G$ is
\begin{align*}
\cost^{(s)}(G)&=(n-1)W+(n-2)W+\cdots+(n-T_W)W+(n-T_W-1)\cdot 1+\cdots+1\cdot 1\\ 
&=\frac{(n-T_W)(n-T_W-1)}{2} + W\cdot \frac{(2n-T_W-1)T_W}{2}\\
&=\frac{n(n-1)+(2n-1-T_W)T_W(W-1)}{2}\\
&=\frac{n(n-1)}{2}-\frac{W-1}{2}T_W^2+\frac{(W-1)(2n-1)}{2}T_W
\end{align*}

Now since $q=\frac{1}{W-1}$, similar to the proof of \Cref{thm:cost_distance_lowerbound}, we have the following two cases.

\textbf{Case I.} If $G_0\in\mathcal{P}_0$,
$\E[T_W]=\frac{n-1}{W-1}(1+\varepsilon)$. Then by Chernoff bound and the fact that $\varepsilon>\sqrt{\frac{W}{40 n}}$, with probability at least $1-e^{-\varepsilon^2n/24W}\geq 1-e^{-40/24}\geq 0.8$, 
\[
\frac{n-1}{{W-1}}(1+\varepsilon)(1-\frac{\varepsilon}{2})\leq T_W
    \leq \frac{n-1}{{W-1}}(1+\varepsilon)(1+\frac{\varepsilon}{2}).
\]

This further gives that 
\[
\cost^{(s)}(G)\geq \frac{n(n-1)}{2}-\frac{W-1}{2}\cdot
[\frac{n-1}{W-1}(1+\varepsilon)(1+\frac{\varepsilon}{2})]^2
+\frac{(W-1)(2n-1)}{2}\cdot[\frac{n-1}{W-1}(1+\varepsilon)(1-\frac{\varepsilon}{2})].
\]

\textbf{Case II.} If $G\in\mathcal{P}_1$,
$\E[T_W]=\frac{n-1}{W-1}(1-\varepsilon)$. Then by Chernoff bound and the fact that $\varepsilon>\sqrt{\frac{W}{40 n}}$, with probability at least $1-e^{-\varepsilon^2n/24W}\geq 1-e^{-40/24}\geq 0.8$, it holds that for $G\in \mathcal{P}_1$, 
\[
  \frac{n-1}{{W-1}}(1-\varepsilon)(1-\frac{\varepsilon}{2})\leq T_W
    \leq \frac{n-1}{{W-1}}(1-\varepsilon)(1+\frac{\varepsilon}{2}).
\]

This further gives that 
\[
\cost^{(s)}(G)\leq \frac{n(n-1)}{2}-\frac{W-1}{2}\cdot
[\frac{n-1}{W-1}(1-\varepsilon)(1-\frac{\varepsilon}{2})]^2
+\frac{(W-1)(2n-1)}{2}\cdot[\frac{n-1}{W-1}(1-\varepsilon)(1+\frac{\varepsilon}{2})].
\]

Thus, when $\varepsilon<\frac{1}{2}$, the gap between the above two bounds is at least
\begin{align*}
 -\frac{(n-1)^2}{2(W-1)}(2+\varepsilon^2)3\varepsilon
    +\frac{(2n-1)(n-1)}{2}\varepsilon\geq \varepsilon(n-1)\left(\frac{2n-1}{2}-\frac{3}{2}\cdot\frac94\cdot\frac{n-1}{W-1}\right)\geq \frac{1}{4}\varepsilon n^2,
\end{align*}

where the last inequality follows from our assumption that $W>10$. 

Thus, any algorithm that estimates the cost $\cost^{(s)}(G)$ with a relative error of $\varepsilon/C$ for some large constant $C$, can be used to distinguish if $G\in\mathcal{P}_0$ or $G\in \mathcal{P}_1$. By \Cref{lem:CRT_lowerlemma}, any algorithm that solves the latter problem requires $\Omega(W/\varepsilon^2)$ queries into $G$ on average. 

To extend our proof to the case of average degree $d$ we proceed in the same way as in the previous proof (except that the weight of the new edges will be $1$ instead of $W$).  
\end{proof}

\newpage
\input{main.bbl}

\newpage
\appendix

\begin{center}
\huge \textbf{Appendix}
\end{center}

\paragraph{Organization of Appendix}

The appendix is organized as follows. We present sublinear algorithms for metric spaces in \Cref{sec:appendix-metric}, covering both the distance and similarity settings. In \Cref{sec:appendix-extension}, we provide extensions of our algorithms, including generalizations to an unknown average degree $d$ and non-integer edge weights, as well as deferred proofs for basic properties. Finally, experimental validations are provided in \cref{sec:experiments}.

\section{Sublinear Algorithms in Metric Space}\label{sec:appendix-metric}
\subsection{The Distance Setting}
\label{sec:appendix-metric-distance}
In this scenario, we consider a weighted graph $G$ in metric space, where the weight 
of each edge is the distance between its two connected vertices, and edge weights satisfy the triangle inequality. 
Our only assumption is that the distance between any pair of vertices can be accessed in constant time.
Using the same definition of $w_i$ and $\cost(G)$ as in the distance case, 
it holds that $\cost_k=\sum_{i=1}^{n-k}w_i$, and
$\cost(G)=\sum_{k=1}^{n} \cost_k=\sum_{i=1}^{n-1}(n-i)\cdot w_i$.

Here, the maximum weight $W$ is the longest distance in $G$, which can be approximated in $O(n)$ time, 
with an approximation factor of $1/2$.
Let this estimate be denoted as $W^*$, then $W^*\le W\le 2W^*$. We then rescale the distances such that $W^*=2n^2/\varepsilon$.
After scaling, all distances are in $[0,4n^2/\varepsilon]$, since the
longest distance is at most $2W^*$.
By the triangle inequality, the cost of the Minimum Spanning Tree (MST), denoted as $\cost(\MST)$,
is at least as large as the longest distance, and hence at least $W^*$. 
Therefore, $\cost(\MST)\geq 2n^2/\varepsilon$.
Since $\cost(\MST)=\cost_1\leq\cost(G)$, it follows that $\cost(G)\geq 2n^2/\varepsilon$.

To simplify our analysis, we round up all edge weights less than $1$ to be exactly $1$. Considering the formula $\cost(G)=\sum_{i=1}^{n-1}(n-i)\cdot w_i$,
this rounding operation affects the cost as follows:
if every edge in the MST has a weight below 1, then $\cost(G)$ is increased by at most
$\sum_{i=1}^{n-1}(n-i)\cdot 1=\frac{n(n-1)}{2}\le \frac{\varepsilon}{4}\cdot\frac{2n^2}{\varepsilon}\le \frac{\varepsilon}{4}\cost(G)$,
thereby introducing an error term of $\varepsilon\cdot \cost(G)/4$.
Moreover, this rounding operation preserves the triangle inequality.
For any triangle with edge weights $a,\ b$ and $c$, there are two relevant cases. 
First, consider the inequality $a+b>c$. The worst scenario is when $c$ is increased from 0 to 1 and one of $a$ or $b$ is less than 1.
After rounding, $a+b\ge 1$ and $c=1$, so $a+b>c$ still holds.
Second, consider $a-b<c$. The worst scenario is when $b\leq a<1$ and both become $1$ after rounding.
Then $a-b=0< c$ still holds. Thus, the triangle inequality is preserved.

We then increase the upper bound of weights to the nearest power of $(1+\varepsilon)$,
i.e., we let $W=(1+\varepsilon)^r$, where $r=\lceil\log_{1+\varepsilon}(4n^2/\varepsilon)\rceil=O(\log(n/\varepsilon)/\varepsilon)$.
Consequently, for the remainder of the metric space, we can assume that all the edge weights are within
$[1,(1+\varepsilon)^r]$.

\subsubsection{Cost Formula for $\cost(G)$}
We begin by deriving a formula for the clustering cost, $\cost(G)$, which simplifies the estimation process.

\begin{lemma}\label{lem:metric-formula}
Let $G$ be an $n$-point graph in metric space such that all pairwise distances are in the interval $[1,W]$, where $W=(1+\varepsilon)^r$. Let $\ell_j=(1+\varepsilon)^j$. For any $0\leq j\leq r$,
we let $G_j$ denote the subgraph of $G$ spanned by all edges with weights \emph{at most} $\ell_j$, and let 
$c_j$ denote the number of connected components in $G_j$. Then we have
\[
\cost(G)\leq \frac{n(n-1)}{2} + \frac{\varepsilon}{2}\cdot \sum_{j=0}^{r-1}(1+\varepsilon)^j\cdot (c_j^2-c_j) \leq (1+\varepsilon) \cost(G)
\]
\end{lemma}
\begin{proof}
We let $G'$ be the weighted graph obtained by rounding every edge weight in $G$ to the nearest power of
$(1+\varepsilon)$. For example, for any pair $(u,v)$, if the distance is 
$(1+\varepsilon)^j<d(u,v)\leq(1+\varepsilon)^{j+1}$, then we round $d(u,v)$ to be $(1+\varepsilon)^{j+1}$. 
After rounding, the $i$-th smallest edge weight on the MST of $G'$ becomes $w_i'$.
Since we only increased the edge weights by at most a factor of $(1+\varepsilon)$, the cost of clustering is
also increased, by at most a factor of $(1+\varepsilon)$. That is, 
\[
\cost(G)\leq \cost(G')\leq (1+\varepsilon)\cost(G).
\]

Note that after rounding, the threshold graphs in $G'$ are not changed, and thus each $c_j$ has the same value with respect to $G$. By running Kruskal's algorithm on $G'$, we will add $n-c_0$ edges of weight $\ell_0=1$, $c_0-c_1$ edges of weight  $\ell_1$, and so on. Thus, the clustering cost of $G'$ can be written as,
\begin{align*}
\cost(G')&=\sum_{i=1}^{n-1}(n-i)\cdot w_i'\\
&=\sum_{i=1}^{n-c_0}(n-i)\cdot \ell_0+\sum_{i=n-c_0+1}^{n-c_1}(n-i)\cdot \ell_1 +\cdots+\sum_{i=n-c_{r-1}+1}^{n-c_r}(n-i)\cdot \ell_r \\
&=\sum_{i=1}^{n-c_r}(n-i)\cdot \ell_0 - \sum_{i=n-c_0+1}^{n-c_r}(n-i)\cdot \ell_0 + \sum_{i=n-c_0+1}^{n-c_r}(n-i)\cdot \ell_1 - \sum_{i=n-c_1+1}^{n-c_r}(n-i)\cdot \ell_1 + \\
&\quad\sum_{i=n-c_1+1}^{n-c_r}(n-i)\cdot \ell_2 - \sum_{i=n-c_2+1}^{n-c_r}(n-i)\cdot \ell_2 +\cdots + \sum_{i=n-c_{r-1}+1}^{n-c_r}(n-i)\cdot \ell_r\\
&=\frac{n(n-1)}{2}+ (\ell_1-\ell_0)\cdot \frac12\cdot (c_0^2-c_0)  +(\ell_2-\ell_1)\cdot \frac12\cdot (c_1^2-c_1)+\cdots +(\ell_r-\ell_{r-1})\cdot \frac12\cdot (c_{r-1}^2-c_{r-1})\\
&=\frac{n(n-1)}{2} + \frac12 \sum_{j=0}^{r-1}(\ell_{j+1}-\ell_{j})(c_j^2-c_j) \\
&=\frac{n(n-1)}{2} + \frac{\varepsilon}{2} \sum_{j=0}^{r-1}(1+\varepsilon)^j\cdot (c_j^2-c_j)
\tag{since $\ell_{j+1}-\ell_j=(1+\varepsilon)^{j+1}-(1+\varepsilon)^j=\varepsilon(1+\varepsilon)^j$}
\end{align*}

This completes the proof of the lemma.
\end{proof}

\subsubsection{Algorithm to Estimate Clustering Cost}

We make use of the algorithm developed by \cite{czumaj2009estimating} for estimating MST.
This algorithm is fundamentally based on estimating the number of connected components of each subgraph $G_j$. 
The core method for estimating the number of connected components is called \textsc{Clique-Tree-Traversal}. 

The main idea of \textsc{Clique-Tree-Traversal} is as follows: If two vertices $u,v$ have distance less than 
$\varepsilon(1+\varepsilon)^j$, then they have the same neighbors in $G_j$, according to triangle inequality.
Therefore, we can select vertices with pairwise distances at least $\varepsilon(1+\varepsilon)^j$ as
\emph{representative vertices}, and do traversal based on them. To save running time, the two thresholds during the traversal are:
\begin{itemize}
    \item Before traversal, pick a value $X$ with distribution $\Pr[X\geq k]=1/k$. If more than $X$
    vertices are explored (including both \emph{representative} and \emph{non-representative vertices}), then quit
    the traversal.
    \item If more then $4r/\varepsilon$ \emph{representative vertices} are explored, then quit the traversal.
\end{itemize}

Having the subroutine \textsc{Clique-Tree-Traversal}, we propose our algorithm in \cref{alg:appcost_metric}.

\begin{algorithm}[h!]
    \DontPrintSemicolon
    \caption{\textsc{AppCostMetric}($G,\varepsilon$)}
    \label{alg:appcost_metric}
for each $j\in\{0,\dots,r-1\}$, invoke \textsc{Clique-Tree-Traversal} to obtain $\hat{c}_j$\;
output $\widehat{\cost}(G)=\frac{n(n-1)}{2} + \frac{\varepsilon}{2} \sum_{j=0}^{r-1}(1+\varepsilon)^j\cdot (\hat{c}_j^2-\hat{c}_j)$\;
\end{algorithm}

\metricdistance*

\subsubsection{Analysis of \cref{alg:appcost_metric}}

According to \cite{czumaj2009estimating}, \textsc{Clique-Tree-Traversal} has the following performance guarantee.

\begin{lemma}\label{lem:metric-c-hat}
Let $U_{rep}^{(j)}$ be defined as in \cite{czumaj2009estimating} (a set of 
\emph{representative vertices} in the whole graph. This set is obtained by full 
\textsc{Clique-Tree-Traversal}, and with maximum cardinality). 
For any given $\varepsilon,\ j\in[1,r]$, there exists an algorithm that computes in time $\Tilde{O}(n/\varepsilon^6)$ and
outputs a value $\hat{c}_j$ such that
\[
\abs{\hat{c}_j-c_j}\leq c_j-c_{j+1}+\frac{3\varepsilon}{8r}\abs{U_{rep}^{(j)}}
\]
\end{lemma}

\begin{proof}
    From \cite{czumaj2009estimating}, we know that $c_{j+1}-K \le \E[\hat c_j] \le c_j$, where $K$ is the number of connected components in $G_{j+1}$ where there exists a vertex $p$, such that starting from $p$ will stop with more than $\frac{4r}{\varepsilon}$ representative vertices.
    By the definition of $K$, we have that $K\cdot\frac{4r}{\varepsilon}\le|U_{rep}^{(j)}|$.
    Thus,
    \[
    c_{j+1}-\frac{\varepsilon}{4r}|U_{rep}^{(j)}| \le \E[\hat c_j] \le c_j
    \]

    Besides, $\Pr[|\hat{c_j}-\E[\hat{c}_j]|\ge\frac{\varepsilon}{8r}|U_{rep}^{(j)}|]\le\frac{1}{16}$.
    Thus, with probability more than $15/16$, we have
    \[
    c_{j+1}-\frac{3\varepsilon}{8r}|U_{rep}^{(i)}| = c_{j+1} -\frac{\varepsilon}{4r}|U_{rep}^{(j)}| -\frac{\varepsilon}{8r}|U_{rep}^{(j)}| \le \hat{c}_j \le c_j+\frac{\varepsilon}{8r}|U_{rep}^{(j)}|
    \]

    Therefore, $\hat{c}_j-c_j\le \frac{\varepsilon}{8r}|U_{rep}^{(j)}|$ and $c_j-\hat{c}_j\le c_j-c_{j+1}+\frac{3\varepsilon}{8r}|U_{rep}^{(j)}|$, which completes the proof.

\end{proof}

Once we bound the error of estimate $\hat{c}_j$ by $\abs{U_{rep}^{(j)}}$, we can bound
$\abs{U_{rep}^{(j)}}$ by $\cost(G)$, and this will help us prove the approximation ratio 
of estimating clustering cost.

\begin{lemma}\label{lem:metric-Urep-cj-cost}
Assume that all edge weights are in $[1,4n^2/\varepsilon]$. It holds that for any $1\leq j\leq r$,
\[
\cost(G)\geq \frac{\varepsilon (1+\varepsilon)^j}{16} \cdot |U_{rep}^{(j)}|\cdot c_j
\]
\end{lemma}
\begin{proof}

When the number of clusters developed in single-linkage is exactly $c_j$, the cost of them is 
$\cost_{c_j}$. As subgraphs $G_1,\dots,G_r$ are built in increasing edge weights, 
the vertices in the connected components in $G_j$ exactly build these single-linkage clusters.
Suppose for a certain connected component $S\in G_j$, there are $s$ vertices in total, and $x$ of them are \emph{representative vertices}. 
As the pairwise distance between \emph{representative vertices} is at least $\varepsilon(1+\varepsilon)^j$, and all the weights are at least 1,
then the cost of MST in $S$ is at least $\varepsilon(1+\varepsilon)^j\cdot(x-1)+1\cdot(s-x)$. 
Since $x\leq s$, then $\cost_S(\MST) \geq\varepsilon\cdot(1+\varepsilon)^j(x-1)$. 
Therefore, by summing up cost for each component $S\in G_j$, we have
\[
\cost_{c_j}\geq\sum_{S\in G_j}\cost_S(\MST)\geq
\sum_{S\in G_j}\varepsilon\cdot(1+\varepsilon)^j(x-1)
=\varepsilon(1+\varepsilon)^j\left(\abs{U_{rep}^{(j)}}-c_j\right)
\]

If we have fewer clusters than $c_j$, we need to use more expensive edges to connect two clusters; then we will have more cost there. 
That is to say, for all $1\leq k\leq c_j$, $\cost_k\geq\cost_{c_j}$. Therefore,
\[
\cost(G)=\sum_{k=1}^{n}\cost_k \geq\sum_{k=1}^{c_j}\cost_k \geq \cost_{c_j}\cdot c_j
\geq\varepsilon(1+\varepsilon)^j\left(\abs{U_{rep}^{(j)}}-c_j\right)\cdot c_j
\]

\textbf{Case I.} When $\abs{U_{rep}^{(j)}} \geq 2 c_j$, we have $\abs{U_{rep}^{(j)}}-c_j\geq \frac{1}{2}\abs{U_{rep}^{(j)}}$,
and so $\cost(G)\geq\varepsilon(1+\varepsilon)^j\abs{U_{rep}^{(j)}}c_j/2$.

\textbf{Case II.} When $\abs{U_{rep}^{(j)}} < 2 c_j$. From \cref{lem:metric-formula},
we have
\[
\cost(G)\ge\frac{n(n-1)}{2(1+\varepsilon)}+\frac{\varepsilon}{2(1+\varepsilon)}\sum_{j=0}^{r-1}(1+\varepsilon)^j\cdot(c_j^2-c_j)
\ge\frac{\varepsilon}{2(1+\varepsilon)}(1+\varepsilon)^j\cdot(c_j^2-c_j)
\ge\frac{\varepsilon}{4}(1+\varepsilon)^j\cdot(c_j^2-c_j).
\]

If $c_j=1$, from \cite{czumaj2009estimating} we know that 
$\cost(\MST)\geq\varepsilon\cdot(1+\varepsilon)^j\cdot\abs{U_{rep}^{(j)}}/4$, and as
$\cost(G)\geq \cost_1=\cost(\MST)$, the lemma is true.
Otherwise, $c_j\geq 2$, then $c_j^2-c_j\geq\frac{c_j^2}{2}$, and so
$\cost(G)\geq\varepsilon(1+\varepsilon)^j c_j^2/8>\varepsilon(1+\varepsilon)^j\abs{U_{rep}^{(j)}} c_j/16$.

This completes the proof of the lemma.
\end{proof}

Similarly, we have the following lemma.
\begin{lemma}\label{lem:metric-Urep2-cost}
    It holds that for any $j\in[1,r]$,
    \[
    \cost(G)\ge\frac{\varepsilon^2(1+\varepsilon)^j}{8r}\cdot|U_{rep}^{(j)}|^2
    \]
\end{lemma}
\begin{proof}
    First of all, suppose $\varepsilon(1+\varepsilon)^j>1$.
    Recall that for any two vertices $u,b\in U_{rep}^{(j)}$, their distance is at least $d(u,v)\ge\varepsilon(1+\varepsilon)^j$.
    Then, for any $i$ such that $(1+\varepsilon)^i\le\varepsilon(1+\varepsilon)^j$,
    the subgraph $G^i$ contains all the edges with weight at most $(1+\varepsilon)^i$.
    Thus, every vertex in the $U_{rep}^{(j)}$ falls in a separate connected component.
    Therefore, $c_i\ge|U_{rep}^{(j)}|$ and
    \begin{align*}
        \cost(G)&\ge\frac{\varepsilon}{2}\sum_{i:(1+\varepsilon)^i\le\varepsilon(1+\varepsilon)^j}(1+\varepsilon)^i\cdot(c_i^2-c_i)\\
        &\ge\frac{\varepsilon}{4}\sum_{i:(1+\varepsilon)^i\le\varepsilon(1+\varepsilon)^j}(1+\varepsilon)^i\cdot c_i^2\\
        &\ge \frac{\varepsilon}{4}\cdot|U_{rep}^{(j)}|^2\cdot[(1+\varepsilon)^0+(1+\varepsilon)^1+\dots+\varepsilon(1+\varepsilon)^j]\\
        &\ge\frac{\varepsilon}{8}\cdot|U_{rep}^{(j)}|^2\cdot\varepsilon(1+\varepsilon)^j\\
        &\ge\frac{\varepsilon^2(1+\varepsilon)^j}{8r}\cdot|U_{rep}^{(j)}|^2
    \end{align*}

    Besides, if $\varepsilon(1+\varepsilon)^j\le 1$, we have $\frac{\varepsilon^2(1+\varepsilon)^j}{8r}|U_{rep}^{(j)}|^2\le\frac{\varepsilon}{8r}n^2\le\cost(G)$,
    as $|U_{rep}^{(j)}|\le n$ and $\cost(G)\ge\frac{n(n-1)}{2}\ge\frac{\varepsilon}{8r}n^2$.
    This completes the proof of the lemma.
    
\end{proof}

Given the above analysis, we are ready to prove the guarantee of estimating clustering cost in metric space.

\begin{proof}[Proof of \cref{thm:cost_metric_distance}]

From \cref{lem:metric-c-hat} we know that $ \hat{c}_j \le c_j+\frac{\varepsilon}{8r}|U_{rep}^{(j)}| $, and then
\begin{align*}
    &|\hat{c}_j^2-c_j^2| = |\hat{c}_j-c_j|\cdot|\hat{c}_j+c_j| \le |\hat{c}_j-c_j|
    \cdot(2c_j+\frac{\varepsilon}{8r}|U_{rep}^{(j)}|)
\end{align*}

Let $\widehat{\cost}(G)=\frac{n(n-1)}{2} + \frac{\varepsilon}{2} \sum_{j=0}^{r-1}(1+\varepsilon)^j\cdot (\hat{c}_j^2-\hat{c}_j)$,
and $\cost(G')=\frac{n(n-1)}{2} + \frac{\varepsilon}{2} \sum_{j=0}^{r-1}(1+\varepsilon)^j\cdot (c_j^2-c_j)$.
Then we have,
\begin{align*}
    \abs{\widehat{\cost}(G)-\cost(G')}&\le \frac{\varepsilon}{2}\sum_{j=0}^{r-1}(1+\varepsilon)^j\cdot|(\hat c_j^2-\hat c_j)-(c_j^2-c_j)|\\
    &\leq\frac{\varepsilon}{2} \sum_{j=0}^{r-1}(1+\varepsilon)^j
    \cdot(\abs{\hat{c}_j^2-c_j^2}+\abs{\hat{c}_j-c_j})\\
    &\le \frac{\varepsilon}{2}\sum_{j=0}^{r-1}(1+\varepsilon)^j \cdot \abs{\hat{c}_j-c_j} \cdot \left(2c_j+\frac{\varepsilon}{8r}|U_{rep}^{(j)}|+1 \right)
    \tag{by the bound of $\abs{\hat{c}_j^2-c_j^2}$}\\
    &\le \frac{\varepsilon}{2} \sum_{j=0}^{r-1}(1+\varepsilon)^j \cdot \left( c_j-c_{j+1}+\frac{3\varepsilon}{8r}|U_{rep}^{(j)}| \right) \cdot 
    \left( 3c_j+\frac{\varepsilon}{8r}|U_{rep}^{(j)}| \right) \tag{by the bound of $\abs{\hat{c}_j-c_j}$, and $1\le c_j$}\\
    &\le\varepsilon\sum_{j=0}^{r-1} (1+\varepsilon)^j\cdot \left( \frac{3}{2}c_j^2-\frac{3}{2}c_jc_{j+1}-\frac{\varepsilon}{16r}|U_{rep}^{(j)}|\cdot c_{j+1} + \frac{10\varepsilon}{16r}|U_{rep}^{(j)}|\cdot c_j
    +\frac{3\varepsilon^2}{128r^2}\cdot|U_{rep}^{(j)}|^2 \right) \\
    &\le \varepsilon\sum_{j=0}^{r-1} (1+\varepsilon)^j\cdot \left( \frac{3}{2}c_j^2-\frac{3}{2}c_jc_{j+1}-\frac{\varepsilon}{16r}|U_{rep}^{(j)}|\cdot c_{j+1}\right) +  \varepsilon\sum_{j=0}^{r-1}(\frac{10}{r}\cost(G)+\frac{3}{16r}\cost(G)) \tag{by \cref{lem:metric-Urep-cj-cost} and \cref{lem:metric-Urep2-cost}}\\
    &\le \varepsilon\sum_{j=0}^{r-1} (1+\varepsilon)^j\cdot \left( \frac{3}{2}c_j^2-\frac{3}{2}c_jc_{j+1}\right) + \varepsilon\cdot r\cdot (\frac{10}{r}+\frac{3}{16r})\cost(G)\\
    &\le (*) + 11\varepsilon\cdot\cost(G)
\end{align*}

Where $(*) = \varepsilon\sum_{j=0}^{r-1} (1+\varepsilon)^j\cdot \left( \frac{3}{2}c_j^2-\frac{3}{2}c_jc_{j+1}\right)$. For the first term in the summation, by the definition of $\cost(G')$, we have that $\varepsilon\sum_{j=0}^{r-1}(1+\varepsilon)^j\cdot\frac{3}{2}c_j^2=3(\cost(G)-\frac{n(n-1)}{2}+\frac{\varepsilon}{2}\sum_{j=0}^{r-1}c_j)$. For the second term in the summation, 
\[
\varepsilon\sum_{j=0}^{r-1}(1+\varepsilon)^j\cdot\frac{3}{2}c_j\cdot c_{j+1}
\ge \frac{3}{1+\varepsilon}\cdot\frac{\varepsilon}{2}\sum_{j=0}^{r-1}(1+\varepsilon)^{j+1}c_{j+1}^2 
= \frac{3}{1+\varepsilon}\cdot\frac{\varepsilon}{2}(\sum_{j=0}^{r-1}(1+\varepsilon)^{j}c_{j}^2-(1+\varepsilon)^0c_0^2+(1+\varepsilon)^rc_r^2)
\]

Denote the weight of the minimum spanning tree in the graph $G'$ as $\cost(\MST)$. Note that $\cost_1=\cost(\MST)=n-W+\varepsilon\cdot\sum_{j=0}^{r-1}(1+\varepsilon)^jc_j$ according to \cite{czumaj2009estimating}. As $(1+\varepsilon)^0c_0^2=n^2$ and $(1+\varepsilon)^rc_r^2=W$, we have that
\begin{align*}
    (*) & \le 3 \left( \cost(G')-\frac{n(n-1)}{2}+\frac{\varepsilon}{2}\sum_{j=0}^{r-1}c_j \right)
    - \frac{3}{1+\varepsilon}\cdot\frac{\varepsilon}{2} \left(\sum_{j=0}^{r-1}(1+\varepsilon)^{j}c_{j}^2-n^2+W \right)\\
    &\le 3 \left( \cost(G')-\frac{n(n-1)}{2}+\frac{\varepsilon}{2}\sum_{j=0}^{r-1}c_j \right)
    -\frac{3}{1+\varepsilon} \left( \cost(G')-\frac{n(n-1)}{2}+\frac{\varepsilon}{2}\sum_{j=0}^{r-1}c_j-\frac{\varepsilon}{2}n^2+\frac{\varepsilon}{2}W \right)
    \tag{by the definition of $\cost(G')$}\\
    &\le 3 \left( \cost(G')-\frac{n(n-1)}{2}+\frac{\varepsilon}{2}\sum_{j=0}^{r-1}c_j \right)
    -3(1-\varepsilon) \left( \cost(G')-\frac{n(n-1)}{2}+\frac{\varepsilon}{2}\sum_{j=0}^{r-1}c_j-\frac{\varepsilon}{2}n^2 \right) \tag{since $\frac{1}{1+\varepsilon}\ge 1-\varepsilon$ and $\varepsilon<1$}\\
    &= 3\varepsilon\cdot\cost(G') - 3\varepsilon\frac{n(n-1)}{2}+\frac{3\varepsilon^2}{2}\sum_{j=0}^{r-1}c_j +\frac{3\varepsilon}{2}n^2 \tag{since $1-\varepsilon\le1$}\\
    &\le 3\varepsilon\cdot\cost(G') + \frac{3\varepsilon}{2}(\cost_1-n+W) + \frac{3\varepsilon}{2}n^2 \tag{by the definition of $\cost_1$}\\
    &\le 3\varepsilon\cdot\cost(G') + \frac{3\varepsilon}{2}\cdot2\cost_1 + \frac{3\varepsilon}{2}\cdot4\cost(G') \tag{since $\cost_1=\cost(\MST)\ge W$ and $n^2\le2n(n-1)\le4\cost(G')$ for $n\ge2$}\\
    &\le 12\varepsilon\cdot\cost(G') \tag{since $\cost_1\le\cost(G')$}
\end{align*}

Therefore, $\abs{\widehat{\cost}(G)-\cost(G')} \le 12\varepsilon\cdot\cost(G') + 11\varepsilon\cdot\cost(G') = 23\varepsilon\cdot\cost(G')$.
From \cref{lem:metric-formula}, we have that $\abs{{\cost(G')}-\cost(G)}\leq\varepsilon\cdot\cost(G)$, and so
\[
\abs{\widehat{\cost}(G)-\cost(G)} \leq \abs{\widehat{\cost}(G)-\cost(G')} + \abs{\cost(G')-\cost(G)} 
\leq 23\varepsilon\cdot\cost(G')+\varepsilon\cdot\cost(G)\le 48\varepsilon\cdot\cost(G)
\]

Replacing $\varepsilon$ with $\varepsilon/48$, we get a $(1+\varepsilon)$ estimate of $\cost(G)$.

\paragraph{Running time analysis.} since each invocation on \textsc{Clique-Tree-Traversal} takes $\tilde{O}(n/\varepsilon^6)$ time, 
and we invoke it for $r=O(\log(n/\varepsilon)/\varepsilon)$ times, the total running time is $r\cdot\tilde{O}(n/\varepsilon^6)=\tilde{O}(n/\varepsilon^7)$.
\end{proof}

\subsection{The Similarity Setting}
\label{sec:appendix-metric-similarity}
We now give sublinear algorithms for SLC cost when the edge weight in a metric graph represents similarity relationship. We obtain the following result. 
\begin{restatable}{thm}{metricsimilarity}
\label{thm:cost_metric_similarity}
Let $G$ be an $n$-point graph in metric space, where each edge weight represents \emph{similarity} between two vertices, and $0< \varepsilon<1$ be a 
    parameter. \cref{alg:appcost_metric_sim} outputs an estimate
    $\widehat{\cost^{(s)}}(G)$ of single-linkage clustering cost $\cost^{(s)}(G)$ in
    metric space, such that with probability at least $3/4$,
    \[
    (1-\varepsilon)\cost^{(s)}(G)\leq \widehat{\cost^{(s)}}(G)\leq (1+\varepsilon)\cost^{(s)}(G).
    \]
    
    The query complexity and running time of the algorithm are 
    $\Tilde{O}(n/\varepsilon^7)$ in expectation.
\end{restatable}

\subsubsection{Cost Formula for $\cost^{(s)}(G)$}

Let $\ell_j=(1+\varepsilon)^j$ and $G_j^{(s)}$ be the subgraph of $G$ spanned by all edges with weight at least $\ell_j$, and let $c^{(s)}_j$ be the number of connected components in $G_j^{(s)}$.
Note that $1=c^{(s)}_0\le c^{(s)}_1\le\dots\le c^{(s)}_{r+1}=n$.
Let $w_1^{(s)}\ge w_2^{(s)}\ge\dots\ge w_{n-1}^{(s)}$ be the sorted weights on the \emph{maximum} spanning tree.
Assume that any edge weight is $(1+\varepsilon)^j$ in the graph $G'$ for some $0\le j\le r$, and $(1+\varepsilon)^r=W$. 
Let $n_j$ be the number of edges with weight $j$ on the maximum spanning tree, and we have $\sum_{j<\ell}n_j=c^{(s)}_{\ell}-1$, and thus $n_j=c^{(s)}_{j+1}-c^{(s)}_{j}$.
Then the cost function can be derived in the following theorem.

\begin{lemma}\label{lem:metric-formula-sim}
Let $G$ be an $n$-point graph in metric space such that all pairwise distances are in the interval $[1,W]$, where $W=(1+\varepsilon)^r$. Let $\ell_j=(1+\varepsilon)^j$. For any $0\leq j\leq r$,
we let $G^{(s)}_j$ denote the subgraph of $G$ spanned by all edges with weights \emph{at least} $\ell_j$, and let 
$c^{(s)}_j$ denote the number of connected components in $G^{(s)}_j$. Then we have
\[
\cost^{(s)}(G)\leq \frac{n(n-1)}{2} + \varepsilon\sum_{j=1}^{r}(1+\varepsilon)^{j-1}\frac{(c^{(s)}_j+n-1)(n-c^{(s)}_j)}{2} \leq (1+\varepsilon) \cost^{(s)}(G)
\]
\end{lemma}
\begin{proof}
We let $G'$ be the weighted graph obtained by rounding every edge weight in $G$ to the nearest power of
$(1+\varepsilon)$. For example, for any pair $(u,v)$, if the distance is 
$(1+\varepsilon)^j<d(u,v)\leq(1+\varepsilon)^{j+1}$, then we round $d(u,v)$ to be $(1+\varepsilon)^{j+1}$. 
After rounding, the $i$-th largest edge weight on the MaxST of $G'$ becomes $w_i'$.
Since we only increased the edge weights by at most a factor of $(1+\varepsilon)$, the cost of clustering is
also increased, by at most a factor of $(1+\varepsilon)$. That is, 
\[
\cost^{(s)}(G)\leq \cost^{(s)}(G')\leq (1+\varepsilon)\cost^{(s)}(G).
\]

Note that after rounding, the threshold graphs in $G'$ are not changed, and thus each $c^{(s)}_j$ has the same value with respect to $G$. Thus, the clustering cost of $G'$ can be written as,
\begin{align*}
    \cost^{(s)}(G')&=\sum_{i=1}^{n-1} (n-i)\cdot w_i^{(s)} \tag{by definition}\\
    &=\sum_{i=1}^{n_r} (n-i)\cdot\ell_r+\sum_{i=n_r+1}^{n_r+n_{r-1}} (n-i)\cdot\ell_{r-1}+\dots
    +\sum_{i=n_r+\dots+n_1+1}^{n_r+\dots+n_0}(n-i)\cdot\ell_0 \tag{reorganizing the sum by grouping the terms according to edge weights}\\
    &=\sum_{i=n-c^{(s)}_{r+1}+1}^{n-c^{(s)}_r}(n-i)\cdot\ell_r+\sum_{i=n-c^{(s)}_r+1}^{n-c^{(s)}_{r-1}}(n-i)\cdot\ell_{r-1}+\dots+\sum_{i=n-c^{(s)}_1+1}^{n-c^{(s)}_0}(n-i)\cdot\ell_0 \tag{since $n_j=c^{(s)}_{j+1}-c^{(s)}_{j}$}\\
    &=\sum_{i=n-c^{(s)}_{r+1}+1}^{n-c^{(s)}_0}(n-i)\cdot\ell_r - \sum_{i=n-c^{(s)}_{r}+1}^{n-c^{(s)}_0}(n-i)\cdot\ell_r + \sum_{i=n-c^{(s)}_r+1}^{n-c^{(s)}_{0}}(n-i)\cdot\ell_{r-1} - \sum_{i=n-c^{(s)}_{r-1}+1}^{n-c^{(s)}_{0}}(n-i)\cdot\ell_{r-1} + \dots \\
    &\quad + \sum_{i=n-c^{(s)}_2+1}^{n-c^{(s)}_0}(n-i)\cdot\ell_1 - \sum_{i=n-c^{(s)}_1+1}^{n-c^{(s)}_0}(n-i)\cdot\ell_1 + \sum_{i=n-c^{(s)}_1+1}^{n-c^{(s)}_0}(n-i)\cdot\ell_0\\
    &=\ell_r\cdot\sum_{i=1}^{n-1}(n-i) - \sum_{i=n-c^{(s)}_r+1}^{n-1}(n-i)\cdot(\ell_r-\ell_{r-1}) - \dots - \sum_{i=n-c^{(s)}_1+1}^{n-1}(n-i)\cdot(\ell_1-\ell_0) \tag{since $c^{(s)}_0=1$ and $c^{(s)}_{r+1}=n$}\\
    &=(1+\varepsilon)^r \frac{n(n-1)}{2} - \varepsilon(1+\varepsilon)^{r-1}\cdot\frac{c^{(s)}_r(c^{(s)}_r-1)}{2} - \dots - \varepsilon(1+\varepsilon)^0 \cdot\frac{c^{(s)}_1(c^{(s)}_1-1)}{2}\\
    &=\frac{n(n-1)}{2} + \varepsilon\sum_{j=1}^{r}(1+\varepsilon)^{j-1}\cdot\frac{n(n-1)}{2} - \varepsilon\sum_{j=1}^{r}(1+\varepsilon)^{j-1}\cdot\frac{c^{(s)}_j(c^{(s)}_j-1)}{2} \tag{since $(1+\varepsilon)^r=1+\varepsilon\sum_{j=1}^r(1+\varepsilon)^{j-1}$}\\
    &=\frac{n(n-1)}{2} + \varepsilon\sum_{j=1}^{r}(1+\varepsilon)^{j-1}\frac{(c^{(s)}_j+n-1)(n-c^{(s)}_j)}{2}
\end{align*}

This completes the proof of the lemma.
\end{proof}

The cost of the maximum spanning tree of graph $G'$ can be derived in a similar approach.
\begin{align}
    \cost(\mathrm{MaxST}) &= \sum_{j=0}^r (1+\varepsilon)^j \cdot n_j \nonumber\\
    &= \sum_{j=0}^r (1+\varepsilon)^j(c^{(s)}_{j+1}-c^{(s)}_{j}) \nonumber\\
    &= \sum_{j'=1}^{r+1} (1+\varepsilon)^{j'-1} \cdot c^{(s)}_{j'}-\sum_{j=0}^{r} (1+\varepsilon)^{j} \cdot c^{(s)}_{j} \tag{substituting $j'=j+1$ in the first sum} \nonumber\\
    &= \sum_{j=0}^{r} (1+\varepsilon)^{j-1} \cdot c^{(s)}_{j} + (1+\varepsilon)^{r} \cdot c^{(s)}_{r+1} - (1+\varepsilon)^{-1} \cdot c^{(s)}_0 -\sum_{j=0}^{r} (1+\varepsilon)^{j} \cdot c^{(s)}_{j} \nonumber\\
    &= (1+\varepsilon)^{r} \cdot n - \frac{1}{1+\varepsilon} - \varepsilon\sum_{j=0}^{r} (1+\varepsilon)^{j-1} \cdot c^{(s)}_{j} \tag{since $c^{(s)}_{r+1}=n, c^{(s)}_0=1$} \nonumber\\
    &= Wn - 1 - \varepsilon\sum_{j=1}^{r} (1+\varepsilon)^{j-1} \cdot c^{(s)}_{j} \label{eqn:cost_maxst_metric}
\end{align}

\subsubsection{Estimating $c^{(s)}_j$}

For a pair of vertices $p$ and $q$, if $d(p,q)<\varepsilon(1+\varepsilon)^{j-1}$, then they share the same neighborhood inside the subgraph $G^{(s)}_j$.
Because for any neighbor $v$ of $p$, we have $d(p,v)\ge(1+\varepsilon)^j$, then 
$d(q,v)\ge d(p,v)-d(p,q)>(1+\varepsilon)^j-\varepsilon(1+\varepsilon)^{j-1}=(1+\varepsilon)^{j-1}$, which means that $v$ is also a neighbor of $q$.
We modify the original algorithm \textsc{Clique-Tree-Traversal} to obtain \textsc{Clique-Tree-Traversal-Similarity}, and the only difference is that the representative set $V_{rep}$ regarding to subgraph $G^{(s)}_j$ has pairwise distance at least $\varepsilon(1+\varepsilon)^{j-1}$.

Analogously to the metric space representing distance, we have the following lemma.

\begin{lemma}\label{lem:metric-c-hat-sim}
Let $U_{rep}^{(j)}$ be a set of 
\emph{representative vertices} in the whole graph. This set is obtained by full 
\textsc{Clique-Tree-Traversal-Similarity}, and with maximum cardinality. 
For any given $\varepsilon,\ j\in[1,r]$, there exists an algorithm that computes in time $\Tilde{O}(n/\varepsilon^6)$ and
outputs a value $\hat{c}^{(s)}_j$ such that
\[
\abs{\hat{c}^{(s)}_j-c^{(s)}_j}\leq c^{(s)}_j-c^{(s)}_{j-1}+\frac{3\varepsilon}{8r}\abs{U_{rep}^{(j)}}
\]
\end{lemma}

\begin{proof}
    Let $V_p^{(j)}$ be the set of vertices in the connected component of vertex $p$ in the graph $G^{(s)}_j$, and $V_{exp}$ be the set of visited vertices when calling subroutine \textsc{Clique-Tree-Traversal-Similarity} and starting traversal from $p$.
    We have that $V_p^{(j)}\subseteq V_{exp}\subseteq V_p^{(j-1)}$.
    
    Therefore, $c^{(s)}_{j-1}-\frac{\varepsilon}{4r}|U_{rep}^{(j)}| \le \E[\hat c^{(s)}_j] \le c^{(s)}_j$ and $\Pr[|\hat{c}^{(s)}_j-\E[\hat{c}^{(s)}_j]|\ge\frac{\varepsilon}{8r}|U_{rep}^{(j)}|]\le\frac{1}{16}$.
    Thus, with probability more than $15/16$, we have
    \[
    c^{(s)}_{j-1}-\frac{3\varepsilon}{8r}|U_{rep}^{(i)}| = c^{(s)}_{j-1} -\frac{\varepsilon}{4r}|U_{rep}^{(j)}| -\frac{\varepsilon}{8r}|U_{rep}^{(j)}| \le \hat{c}^{(s)}_j \le c^{(s)}_j+\frac{\varepsilon}{8r}|U_{rep}^{(j)}|
    \]

    Therefore, $\hat{c}^{(s)}_j-c^{(s)}_j\le \frac{\varepsilon}{8r}|U_{rep}^{(j)}|$ and $c^{(s)}_j-\hat{c}^{(s)}_j\le c^{(s)}_j-c^{(s)}_{j-1}+\frac{3\varepsilon}{8r}|U_{rep}^{(j)}|$, which completes the proof.

\end{proof}

\subsubsection{Estimating $\cost^{(s)}(G)$}

\begin{algorithm}[h!]
    \DontPrintSemicolon
    \caption{\textsc{AppCostMetricSim}($G,\varepsilon$)}
    \label{alg:appcost_metric_sim}
    for each $j\in\{1,\dots,r\}$, invoke \textsc{Clique-Tree-Traversal-Similarity} to obtain $\hat{c}^{(s)}_j$\;
    output $\widehat{\cost^{(s)}}(G)=\frac{n(n-1)}{2} + \frac{\varepsilon}{2} \sum_{j=1}^{r}(1+\varepsilon)^{j-1}\cdot (\hat{c}^{(s)}_j+n-1)(n-\hat{c}^{(s)}_j)$\;
\end{algorithm}

\metricsimilarity*

\begin{proof}[Proof of \cref{thm:cost_metric_similarity}]
    By \cref{lem:metric-c-hat-sim}, $\abs{\hat{c}^{(s)}_j-c^{(s)}_j}\leq c^{(s)}_j-c^{(s)}_{j-1}+\frac{3\varepsilon}{8r}\abs{U_{rep}^{(j)}}$, and we denote this additive error as $a$.
    Denote ${c}^{(s)}_j+n-1$ as $A_j$ and $n-{c}^{(s)}_j$ as $D_j$. Then for each $1\le j\le r$, 
    \begin{align*}
        &A_j-a \le \widehat{A}_j\le A_j+a\\
        &D_j-a \le \widehat{D}_j\le D_j+a\\
        &(A_j-a)(D_j-a)\le \widehat{A}_j\cdot\widehat{D}_j\le (A_j+a)(D_j+a)\\
    \end{align*}

    Therefore, $|\widehat{A}_j\cdot\widehat{D}_j-A_j\cdot D_j| \le a(A_j+D_j)+a^2$. 
    Since $a=c^{(s)}_j-c^{(s)}_{j-1}+\frac{3\varepsilon}{8r}\abs{U_{rep}^{(j)}}\le n+\frac{3\varepsilon}{8r}\cdot n\le 2n-1 = A_j+D_j$,
    then $|\widehat{A}_j\cdot\widehat{D}_j-A_j\cdot D_j| \le 2a(A_j+D_j)\le(c^{(s)}_j-c^{(s)}_{j-1}+\frac{3\varepsilon}{8r}\abs{U_{rep}^{(j)}})\cdot 4n$.

    Thus, the error of the estimate $\widehat{\cost^{(s)}}(G)$ is,
    \begin{align*}
        |\widehat{\cost^{(s)}}(G)-\cost^{(s)}(G')| &\le \frac{\varepsilon}{2}\sum_{j=1}^r(1+\varepsilon)^{j-1}|\widehat{A}_j\cdot\widehat{D}_j-A_j\cdot D_j|\\
        &\le \frac{\varepsilon}{2}\sum_{j=1}^r(1+\varepsilon)^{j-1}(c^{(s)}_j-c^{(s)}_{j-1}+\frac{3\varepsilon}{8r}|U_{rep}^{(j)}|)\cdot 4n\\
    \end{align*}

    On the one hand, since all the representative vertices inside $U_{rep}^{(j)}$ have pairwise distances at least $\varepsilon(1+\varepsilon)^{j-1}$, the cost of the maximum spanning tree of $G$ is at least $\cost(\mathrm{MaxST})\ge \varepsilon(1+\varepsilon)^{j-1}(|U_{rep}^{(j)}|-1)\ge\frac{\varepsilon}{2}(1+\varepsilon)^{j-1}|U_{rep}^{(j)}|$, if $|U_{rep}^{(j)}|\ge2$.

    On the other hand, we have
    \begin{align*}
        \sum_{j=1}^r(1+\varepsilon)^{j-1}(c^{(s)}_j-c^{(s)}_{j-1}) &= \sum_{j=1}^r(1+\varepsilon)^{j-1}c^{(s)}_j - (1+\varepsilon)\sum_{j=1}^r(1+\varepsilon)^{j-2}c^{(s)}_{j-1}\\
        &= \sum_{j=1}^r(1+\varepsilon)^{j-1}c^{(s)}_j - (1+\varepsilon)\sum_{j'=0}^{r-1}(1+\varepsilon)^{j'-1}c^{(s)}_{j'} \tag{substituting $j'=j-1$ in the second sum}\\
        &= \sum_{j=1}^r(1+\varepsilon)^{j-1}c^{(s)}_j - (1+\varepsilon) \left( \sum_{j=1}^r(1+\varepsilon)^{j-1}c^{(s)}_{j}+(1+\varepsilon)^{-1}c^{(s)}_0-(1+\varepsilon)^{r-1}c^{(s)}_{r} \right)\\
        &= -\varepsilon\sum_{j=1}^r(1+\varepsilon)^{j-1}c^{(s)}_j - 1 + (1+\varepsilon)^r c^{(s)}_r \tag{since $c^{(s)}_0=1$}\\
        &= \cost(\mathrm{MaxST})-Wn+1 -1 + Wc^{(s)}_r \tag{by \cref{eqn:cost_maxst_metric} and since $(1+\varepsilon)^r=W$}\\
        &\le \cost(\mathrm{MaxST}) \tag{since $c^{(s)}_r\le n$}
    \end{align*}

    Therefore, the error of $\widehat{\cost^{(s)}}(G)$ is
    \begin{align*}
        |\widehat{\cost^{(s)}}(G)-\cost^{(s)}(G')|&\le \frac{\varepsilon}{2}\cost(\mathrm{MaxST})\cdot 4n + \sum_{j=1}^r \frac{3\varepsilon}{8r}\cost(\mathrm{MaxST})\cdot 4n \tag{by the above analysis}\\
        &= \frac{7\varepsilon}{2}\cost(\mathrm{MaxST})\cdot n
    \end{align*}

    Furthermore, the lower bound of $\cost^{(s)}(G')$ is,
    \begin{align*}
        \cost^{(s)}(G') &=\frac{n(n-1)}{2} + \varepsilon\sum_{j=1}^{r}(1+\varepsilon)^{j-1}\frac{(c^{(s)}_j+n-1)(n-c^{(s)}_j)}{2} \tag{by \cref{lem:metric-formula-sim}}\\
        &\ge \frac{n(n-1)}{2} + \frac{n}{2}\cdot \varepsilon\sum_{j=1}^{r}(1+\varepsilon)^{j-1}(n-c^{(s)}_j) \tag{since $c^{(s)}_j\ge1$}\\
        &= \frac{n(n-1)}{2} + \frac{n}{2}\cdot \left( \varepsilon\cdot\frac{(1+\varepsilon)^{0}((1+\varepsilon)^r-1)}{\varepsilon}n-\varepsilon\sum_{j=1}^{r}(1+\varepsilon)^{j-1} c^{(s)}_j \right)\\
        &= \frac{n(n-1)}{2} + \frac{n}{2}\cdot(Wn-n-\varepsilon\sum_{j=1}^{r}(1+\varepsilon)^{j-1} c^{(s)}_j) \tag{since $(1+\varepsilon)^r=W$}\\
        &= \frac{n(n-1)}{2} + \frac{n}{2}\cdot(\cost(\mathrm{MaxST})+1-n) \tag{by \cref{eqn:cost_maxst_metric}}\\
        &= \frac{n}{2}\cdot \cost(\mathrm{MaxST})
    \end{align*}

    Thus, $|\widehat{\cost^{(s)}}(G)-\cost^{(s)}(G')|\le 7\varepsilon\cdot \cost^{(s)}(G')$. Since $|{\cost^{(s)}}(G')-\cost^{(s)}(G)|\le\varepsilon\cdot\cost^{(s)}(G)$, we have that
    \begin{align*}
        |\widehat{\cost^{(s)}}(G)-\cost^{(s)}(G)|&\le |\widehat{\cost^{(s)}}(G)-\cost^{(s)}(G')| + |{\cost^{(s)}}(G')-\cost^{(s)}(G)|\\ 
        &\le 7\varepsilon\cdot \cost^{(s)}(G') + \varepsilon\cdot\cost^{(s)}(G)\\ 
        &\le 16\varepsilon\cdot\cost^{(s)}(G)
    \end{align*}

    Replacing $\varepsilon$ with $\varepsilon/16$, we get a $(1+\varepsilon)$ estimate of $\cost^{(s)}(G)$.

    \paragraph{Running time analysis.} since each invocation on \textsc{Clique-Tree-Traversal-Similarity} takes $\tilde{O}(n/\varepsilon^6)$ time, 
    and we invoke it for $r=O(\log(n/\varepsilon)/\varepsilon)$ times, the total running time is $r\cdot\tilde{O}(n/\varepsilon^6)=\tilde{O}(n/\varepsilon^7)$.

\end{proof}

\subsection{Lower Bounds for the Metric Space}
Finally, we establish the information-theoretic query complexity lower bounds for the metric space setting. The following proof demonstrates that even when the underlying graph satisfies the triangle inequality and permits constant-time distance or similarity queries, any relative $(1+\eps)$-approximation of the total hierarchy cost fundamentally requires a linear number of queries.

\begin{proposition}[Linear-query lower bound in metric spaces]
\label{prop:metric_linear_lower_bound}
For every fixed $\eps\in(0,1)$, any randomized algorithm that, given oracle
access to an $n$-point metric space and succeeding with probability at least
$3/4$, outputs a $(1+\eps)$-approximation to the single-linkage hierarchy cost
must make $\Omega(n)$ queries. The same holds in the metric-space similarity
setting.
\end{proposition}
\begin{proof}
We prove the lower bound in the distance setting first, and then in the
similarity setting. In both cases, the randomized lower bound follows from
Yao's minimax principle by exhibiting a hard distribution over instances for
deterministic algorithms.

\medskip
\noindent
\textbf{Distance setting.}
Assume $n$ is even; the odd case follows by ignoring one point. Fix a constant
$\Delta>1$, to be chosen later. Consider the following two distributions over
$n$-point metrics on a common point set $X$.

\begin{itemize}
    \item Under $\mathcal{D}_0$, all pairwise distances are equal to $\Delta$:
    \[
    d_0(x,y)=\Delta \qquad \text{for all } x\neq y.
    \]

    \item Under $\mathcal{D}_1$, choose a uniformly random perfect matching
    $M$ on $X$, and define
    \[
    d_1(x,y)=
    \begin{cases}
        1, & \text{if } \{x,y\}\in M,\\[1mm]
        \Delta, & \text{otherwise.}
    \end{cases}
    \]
\end{itemize}

Both are valid metrics: every triangle has edge multiset either
$\{\Delta,\Delta,\Delta\}$ or $\{1,\Delta,\Delta\}$, and the triangle
inequality is immediate in either case.

Let $C_0$ and $C_1$ denote the corresponding SLC hierarchy costs.

Under $\mathcal{D}_0$, every edge in an MST has weight $\Delta$, and hence
\[
C_0
=
\Delta \sum_{i=1}^{n-1}(n-i)
=
\Delta \cdot \frac{n(n-1)}{2}.
\]

Under $\mathcal{D}_1$, an MST contains all $n/2$ matching edges of weight $1$,
and then $n/2-1$ additional edges of weight $\Delta$ to connect the $n/2$
matched pairs. Thus
\[
C_1
=
\sum_{i=1}^{n/2}(n-i)
+
\Delta \sum_{i=n/2+1}^{n-1}(n-i)
=
\frac{n(3n-2)}{8}
+
\Delta \frac{n(n-2)}{8}.
\]

Therefore
\[
\frac{C_0}{C_1}
=
\frac{4\Delta(n-1)}{(\Delta+3)n-2(\Delta+1)}.
\]
As $\Delta\to\infty$, this ratio tends to $4$. Hence, for every fixed
$0<\varepsilon<1$, we may choose $\Delta=\Delta(\varepsilon)$ large enough so
that
\[
\frac{C_0}{C_1} > (1+\varepsilon)^2.
\]
With this choice, a $(1+\varepsilon)$-approximation to the hierarchy cost
distinguishes $\mathcal{D}_0$ from $\mathcal{D}_1$.

Now fix any deterministic algorithm making $q$ oracle queries. On an instance
from $\mathcal{D}_0$, every answer is $\Delta$, so the sequence of queried pairs
is completely determined. On an instance from $\mathcal{D}_1$, the transcript is
identical unless one of the queried pairs is one of the $n/2$ matching edges.
For any fixed queried pair $e$,
\[
\Pr_{M\sim\mathcal{D}_1}[e\in M]=\frac{1}{n-1}.
\]
Hence, by a union bound,
\[
\Pr[\text{some queried pair lies in }M]\le \frac{q}{n-1}.
\]
So with probability at least $1-\frac{q}{n-1}$, the entire query-answer
transcript under $\mathcal{D}_1$ is identical to that under $\mathcal{D}_0$.

Consider now the mixed distribution
\[
\mu = \frac{1}{2}\mathcal{D}_0 + \frac{1}{2}\mathcal{D}_1.
\]
On the common transcript in which all answers equal $\Delta$, the algorithm must
output a single fixed value. Since the $(1+\varepsilon)$-approximation intervals
around $C_0$ and $C_1$ are disjoint, that output can be correct for at most one
of the two instances. Therefore the success probability of the algorithm under
$\mu$ is at most
\[
\frac{1}{2} + \frac{1}{2}\cdot \frac{q}{n-1}.
\]
In particular, to achieve success probability at least $3/4$, one must have
$q\ge (n-1)/2$. By Yao's principle, every randomized algorithm with success
probability at least $3/4$ must make $\Omega(n)$ queries.

\medskip
\noindent
\textbf{Similarity setting.}
Again assume $n$ is even only for convenience. Fix a constant
$\Lambda>(1+\varepsilon)^2$. Consider the following two distributions over
$n$-point similarity metrics.

\begin{itemize}
    \item Under $\mathcal{S}_0$, all pairwise similarities are equal to $1$:
    \[
    s_0(x,y)=1 \qquad \text{for all } x\neq y.
    \]

    \item Under $\mathcal{S}_1$, choose a uniformly random hidden center
    $c\in X$, and define
    \[
    s_1(x,y)=
    \begin{cases}
        \Lambda, & \text{if } x=c \text{ or } y=c,\\[1mm]
        1, & \text{otherwise.}
    \end{cases}
    \]
\end{itemize}

These instances satisfy the triangle inequality in the similarity-metric model:
every triangle has edge multiset either $\{1,1,1\}$ or $\{\Lambda,\Lambda,1\}$.

Let $C_0^{(s)}$ and $C_1^{(s)}$ denote the corresponding hierarchy costs. Under
$\mathcal{S}_0$, every edge in a MaxST has weight $1$, so
\[
C_0^{(s)}=\sum_{i=1}^{n-1}(n-i)=\frac{n(n-1)}{2}.
\]
Under $\mathcal{S}_1$, the unique optimal MaxST is the star centered at $c$,
whose $n-1$ edges all have weight $\Lambda$. Hence
\[
C_1^{(s)}
=
\Lambda \sum_{i=1}^{n-1}(n-i)
=
\Lambda \cdot \frac{n(n-1)}{2}.
\]
Since $\Lambda>(1+\varepsilon)^2$, a $(1+\varepsilon)$-approximation again
distinguishes the two cases.

Now fix any deterministic algorithm making $q$ queries. On an instance from
$\mathcal{S}_0$, every answer is $1$, so the queried pairs are fixed. Under
$\mathcal{S}_1$, the transcript differs only if one of the queried pairs is
incident to the hidden center $c$. For any fixed queried pair, this happens with
probability at most $2/n$. Therefore
\[
\Pr[\text{some queried pair is incident to }c]\le \frac{2q}{n}.
\]
Thus, under the mixed distribution
\[
\nu = \frac{1}{2}\mathcal{S}_0 + \frac{1}{2}\mathcal{S}_1,
\]
the success probability is at most
\[
\frac{1}{2} + \frac{q}{n}.
\]
Hence any algorithm succeeding with probability at least $3/4$ must have
$q\ge n/4$, and so the query complexity is again $\Omega(n)$.

Combining the two parts proves the proposition.
\end{proof}

\section{Extensions of the Algorithm and Supplementary Proofs}\label{sec:appendix-extension}
\subsection{Generalization to Unknown $d$}
\label{sec:appendix-no_given_d}

    In \cref{alg:appncc} (or \cref{alg:appncc_sim}), in order to estimate the number of connected components $c$ (or $D=n-c$), we sample several vertices, and do BFS from them, regarding to two thresholds, $\Gamma=\lceil4k/\varepsilon\rceil$ and $d^{(G)}=d\cdot\Gamma$.
    Note that $k=\sqrt{W}$ in \cref{alg:appncc} and $k=W$ in \cref{alg:appncc_sim}.

    In case when $d$ is not given, we refer to \cite{chazelle2005approximating} and sample $2\Gamma$ vertices,
    and let $d^{(G)}$ to be the maximum degree among them.
    Note that we can compute $d^{(G)}$ in $d\cdot\Gamma$ time in expectation,
    because for each sampled vertex $v$, we know its degree in $\deg(v)$ time.
    We also argue that with high constant probability the rank of $d^{(G)}$ is
    $\Theta(\frac{n}{\Gamma})$, and the number of vertices 
    with degree greater than $d^{(G)}$, is at most $\frac{n}{4\Gamma}$. 
    We have
    \[
    \Pr\left[\text{rank of }d^{(G)}>\frac{n}{4\Gamma}\right]
    \leq \left(1-\frac{1}{4\Gamma}\right)^{2\Gamma}
    \leq \mathrm{e}^{-8}
    \]

    On the other hand, we also have
    \[
    \Pr\left[\text{rank of }d^{(G)}>\frac{n}{16\Gamma}\right]
    = \left(1-\frac{1}{16\Gamma}\right)^{2\Gamma}
    \geq \mathrm{e}^{-\frac{2}{8}}
    \geq 1-\frac{1}{4}=\frac{3}{4}
    \]

    Therefore, the rank of $d^{(G)}$ is in the interval
    $(\frac{n}{16\Gamma},\frac{n}{4\Gamma})$
    with probability at least $1-(e^{-8} + 1/4)$,
    which implies that $d^{(G)} \le 16 d\cdot\Gamma$ and that the number of connected components containing vertices with degree
    greater than $d^{(G)}$ (with rank lower than $d^{(G)}$'s rank), is at most 
    $\frac{n}{\Gamma}$. 
    Furthermore, the number of connected components with size greater
    than $\Gamma$ is at most $\frac{n}{\Gamma}$. Then we have that, 
    \[c-\frac{2 n}{\Gamma}\leq c_U\leq c.\]

    For the remaining part to prove theoretical guarantee of $c$ or $D=n-c$, one can refer to proofs of \cref{lem:appncc} or \cref{lem:appncc_sim}.

\subsection{Generalization to Non-integer Edge Weights}
\label{sec:appendix-rounding-weights}
When the weights come from the interval $[1,W]$ and are not necessarily integers and we have $1>\varepsilon>0$, we can first multiply the weights with $1/\varepsilon$ and then round them down to the nearest integer value. The resulting weights are integers from $\{1,\dots, \lfloor W/\varepsilon \rfloor\}$. Let us call the scaled weights $w'$ and the rounded and scaled weights $w''$.
Each edge weight differs at most $1$ from its correct value, that is for every each $e$ we have $|w'(e)- w''(e)| \le 1$.
If we use $M$ to denote the cost of the MST with the original weights, $M'$ the cost of the MST with the scaled weights and 
$M''$ the cost of the MST with the scaled and rounded weights, we get $M=\varepsilon M'$ and 
$|M'-M''| \le n-1$. It follows that $|M-\varepsilon M''| \le \varepsilon (n-1) \le \varepsilon M$. Thus, we can approximate the cost of the MST with scaled and rounded weights. The resulting running time in the distance case will be $\tilde O(\sqrt{W}/\varepsilon^{3.5})$. The error will be at most $(1\pm \varepsilon)^2 \le 1\pm 3\varepsilon$.
Replacing $\varepsilon$ with $\varepsilon/3$ gives a $(1\pm \varepsilon)$-approximation.

\subsection{Deferred Proofs of Basic Properties}\label{subsec:deferred_proofs}

In this section, we provide the proofs omitted from the main text. We begin with the structural foundation of our cost metric, formally verifying that the greedy SLC process inherently optimizes the flat spanning-forest objective at every level $k$.

\flatforest*

\begin{proof}This is the standard cut/cycle property of MSTs (equivalently, the correctness of
Kruskal's algorithm). Kruskal's algorithm processes edges in nondecreasing order and,
after $n-k$ successful insertions, has exactly $k$ components. The chosen edges are
precisely the $n-k$ lightest edges of an MST, equivalently the MST with its $k-1$ largest
edges removed. Any lower-cost $k$-component forest could be extended to a spanning tree,
contradicting MST optimality (or, equivalently, would contradict the greedy correctness
of Kruskal's algorithm).
\end{proof}

Next, we turn to the algorithmic limitations of profile approximation. 

\perlevelimpossible*

\begin{proof}For $k=n-1$, the SLC cost is
$\cost_{n-1}=w_1$,
where $w_1$ is the lightest edge in the MST. This is also the minimum edge weight in $G$.

Consider two distributions over weight assignments on a fixed connected graph topology
with $m$ edges. In the first distribution, every edge has weight $2$. In the second,
one uniformly random edge has weight $1$ and all other edges have weight $2$. Then
$\cost_{n-1}=2$ in the first case,
and
$\cost_{n-1}=1$ in the second case. 

A $(1+\epsilon)$-multiplicative estimate distinguishes these two cases whenever
$(1+\epsilon)^2 < 2$. However, unless the algorithm inspects the unique hidden light
edge, the two instances are indistinguishable. With $q$ inspected edges, the probability
of finding the hidden edge is at most $q/m$, so constant success probability requires
$q=\Omega(m)$.
\end{proof}

\section{Experiments and Evaluation}\label{sec:experiments}

In this section, we first validate the structural utility of the proposed SLC cost profile $\mathrm{\mathbf{p}}(G)=(\cost_1,\dots,\cost_n)$ as a meaningful clustering metric. Then, we rigorously benchmark our sublinear-time algorithms to demonstrate their approximation accuracy and runtime efficiency across various real-world datasets.
We provide both a concise overview of performance and an exhaustive catalog of extended results.

\subsection{Validation of the SLC Cost Profile}\label{subsec:profile_validation}

We first show experimental validation of the proposed $\cost_k$ metric across various synthetic and real-world datasets (\cref{fig:validate_k,fig:validate_hierarchy}). 

\begin{figure*}[!htbp]
    \centering
    \includegraphics[width=.8\textwidth]{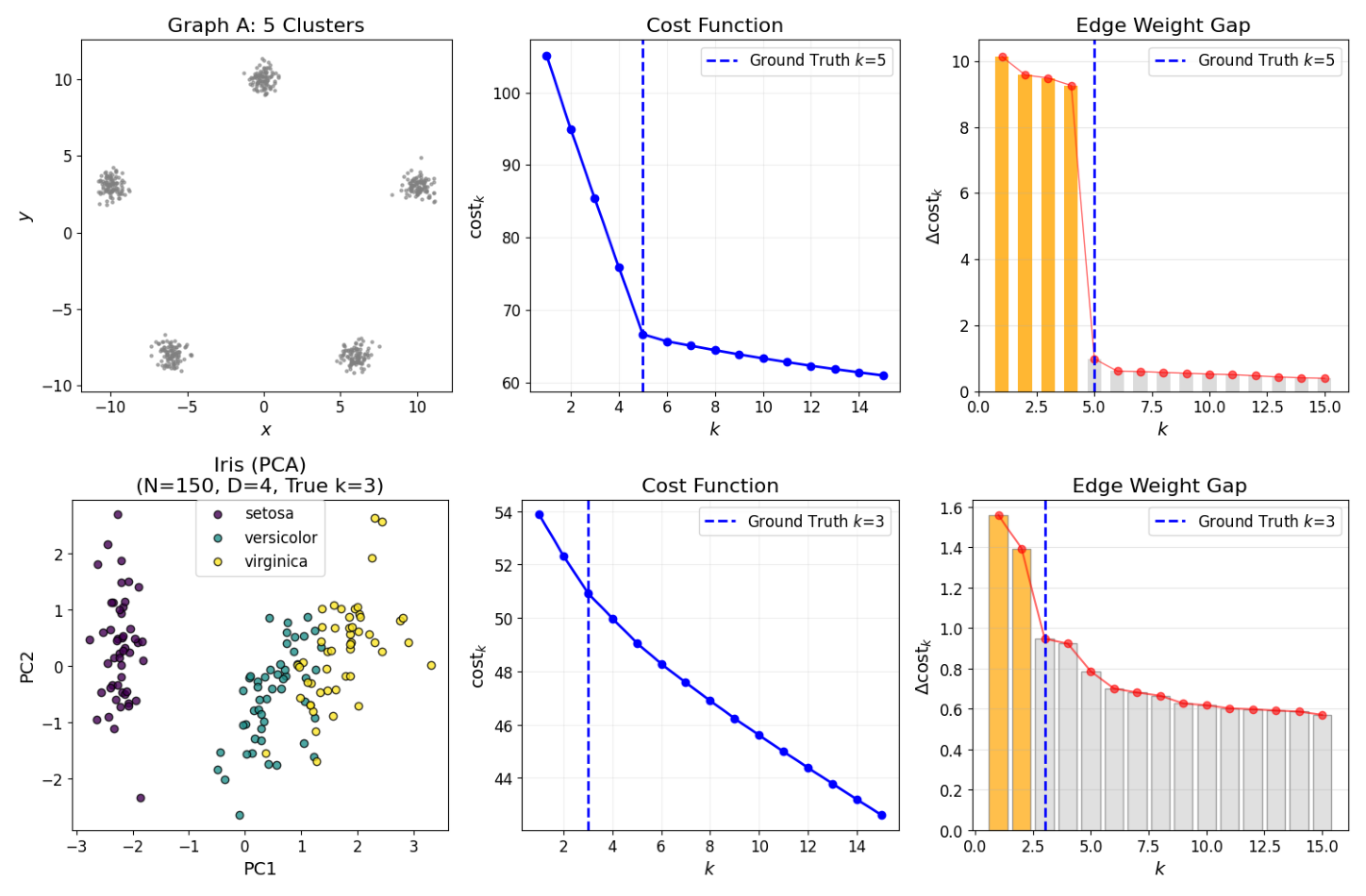}
    \caption{Identifying the optimal number of clusters via the SLC profile. A pronounced elbow in the curve corresponds to a large gap in the discrete derivative $\Delta_k$, indicating the natural scale at which further merging becomes significantly expensive.}
    \label{fig:validate_k}
\end{figure*}

\begin{figure*}[!htbp]
    \centering
    \includegraphics[width=.8\textwidth]{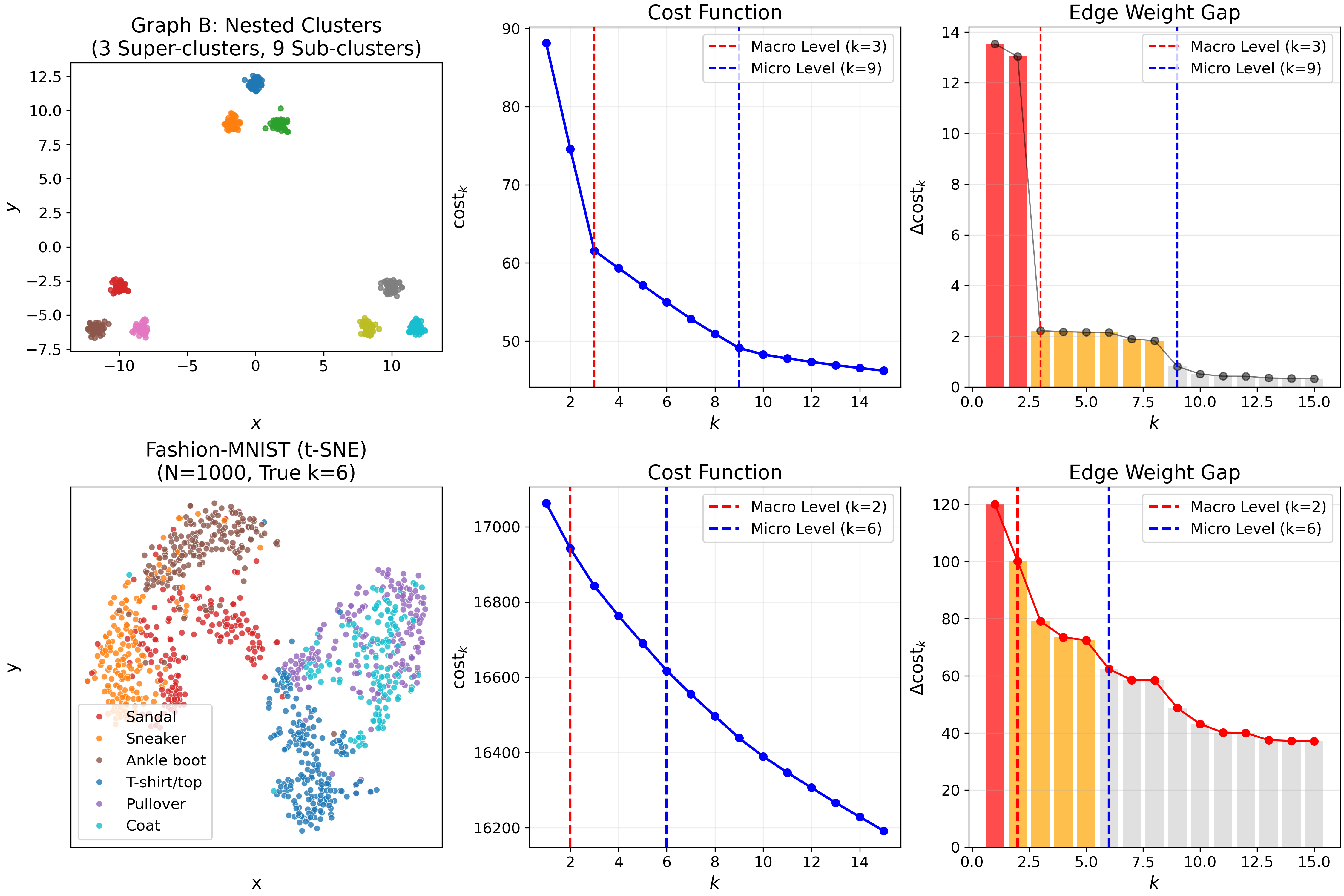}
    \caption{Capturing the multiscale hierarchical structure. Multiple distinct jumps in the profile reveal natural clustering resolutions at different scales.}
    \label{fig:validate_hierarchy}
\end{figure*}

Specifically, we observe two primary utilities of our proposed SLC metric:

\begin{enumerate}
    \item \textbf{Finding the number of clusters.} In datasets with well-separated clusters (\cref{fig:validate_k}), the inflection points on the $\cost_k$ curve (and also on the $\Delta_k$ curve, where $\Delta_k \;:=\; \cost_k-\cost_{k+1} \;=\; w_{n-k}$) align closely with the ground-truth number of clusters. A significant gap where $\Delta_k \gg \Delta_{k+1}$ indicates that the optimal number of clusters is $k+1$, as further merging would require a large cost.

    \item \textbf{Characterizing of hierarchical structures.} For datasets with clear hierarchical structures (\cref{fig:validate_hierarchy}), the $\Delta_k$ reveals the hierarchy of data points at multiple scales. Multiple distinct inflection points in the $\Delta_k$ curve correspond to the merging of sub-clusters at different levels of the hierarchy, providing a clear visual heuristic for understanding data organization at various granularities.
\end{enumerate}

Moreover, since our algorithm can efficiently estimate the $\cost_k$ vs. $k$ curve in sublinear time, these visualizations offer a practical and scalable approach for identifying meaningful clustering ranges in large-scale data.

\subsection{Experimental Setup for Algorithm Evaluation}
To evaluate the performance (accuracy and speedup) of our sublinear algorithms, we conducted extensive experiments on real-world datasets.

\paragraph{Environment and Baseline.}
All experiments were implemented in C++, using an Intel(R) Xeon(R) Platinum 8358 Processor @ 2.60 GHZ, with 504 GB RAM.
We use Kruskal's  algorithm for minimum/maximum spanning tree computation as our baseline, implemented in C++ boost library \cite{boost}.\footnote{Our source code can be accessed in: https://zenodo.org/records/21278105}  

\paragraph*{Dataset Preprocessing.} Under our assumption of connectivity in graphs, we preprocessed the datasets to find the largest connected components
as input of the algorithm.
The information of preprocessed datasets is detailed in \cref{tbl:dataset}.

For datasets in distance setting (localization based datasets and road networks),
we assign the distance of two neighboring vertices to the edge weight. 
For datasets in similarity setting (Spotify co-listening graph, and co-citation graphs),
we assign the number of collaborations (or the number of times two songs are co-listened) to the edge weight between two neighboring vertices. 

\begin{table}[h]
\caption{All of the datasets used: for distance case, we have road network for
different countries, and friendship networks with location 
(loc-brightkite/gowalla);
For similarity case, we have co-authorship datasets on different fields,
and a co-listened dataset for an music application called Spotify.
Each dataset listed in the table has been preprocessed to extract the largest connected 
component from the original graph. The parameters $n$, $m$, and $W$ represent 
the number of vertices, the number of edges, and the largest weight value, respectively, 
in these preprocessed graphs.
}
\label{tbl:dataset}
\centering
\begin{tabular}{|llll|}
\hline
Name of dataset  & $n$          & $m$             & $W$  \\
\hline
\multicolumn{4}{|c|}{Dataset in distance case} \\
\hline
Location-Brightkite \citep{cho2011friendship}
& 49,011 & 386,716 & 19,985 \\
Location-Gowalla \citep{cho2011friendship}
& 96,953 & 910,052 & 19,883 \\
Luxembourg road network \citep{10thDIMACS}
& 114,599 & 119,666 & 2,065 \\
Belgium road network \citep{10thDIMACS}
& 1,441,295 & 1,549,970 & 6,408 \\
Netherland road network \citep{10thDIMACS}
& 2,216,688 & 2,441,238 & 7,027 \\
Italy road network \citep{10thDIMACS}
& 6,686,493 & 7,013,978 & 9,719 \\
Great-Britain road network \citep{10thDIMACS}
& 7,733,822 & 8,156,517 & 10,520 \\
Germany road network \citep{10thDIMACS}
& 11,548,845 & 12,369,181 & 15,337 \\
Asia road network \citep{10thDIMACS}
& 11,950,757 & 12,711,603 & 87,377 \\
USA road network \citep{9thDIMACS}
& 23,947,347 & 58,333,344 & 24,394 \\
Europe road network \citep{10thDIMACS}
& 50,912,018 & 54,054,660 & 368,855 \\
\hline
\multicolumn{4}{|c|}{Dataset in similarity case} \\
\hline
Co-authorship Business \citep{Benson-2018-sos,Benson-2018-simplicial,Sinha-2015-MAG}
& 40,383 & 59,513 & 63 \\
Co-authorship CS (MAG) \citep{Amburg-2020-categorical,Sinha-2015-MAG}
& 59,709 & 415,430 & 36 \\
Co-authorship History \citep{Benson-2018-simplicial,Sinha-2015-MAG}
& 219,435 & 1,614,992 & 606 \\
Co-authorship Geology \citep{Benson-2018-simplicial,Sinha-2015-MAG}
& 898,648 & 9,782,224 & 192 \\
Co-authorship CS (DBLP) \citep{Benson-2018-simplicial}
& 1,431,475 & 7,886,713 & 121 \\
Co-listen Spotify \citep{Kumar-2020-top-weighted}
& 3,061,417 & 85,467,545 & 99,536 \\
\hline
\end{tabular}
\end{table}

\paragraph*{Implementation Changes to Our Algorithm.} 
In the theoretical part, we assumed that the average degree $d$ is given, and set $d^{(G)}$ to be $d\cdot\Gamma$ as the threshold degree. To be able to deal with a setting when the average degree is not given, our algorithms more general, we sample $\Gamma$ vertices in the experiments, and let $d^{(G)}$ be the maximum degree among these sampled vertices,
similar as in \cite{chazelle2005approximating}.

We pick different constants from theory when choosing parameters. Because we are always assuming the worst case in theory, and try to bound the performance of each run; while in the real world, we can get good results with much smaller constants, and we can also refine the result by calculating the average cost among multiple runs, so we don't need to bound the performance of each run.
In the implementation, the algorithm is given parameter $r$ as sample size, and parameter $\Gamma$ is set accordingly: $\Gamma=1\cdot k\sqrt{\frac{r}{k}}=\sqrt{rk}$, as $r=O(k/\varepsilon^2)$
and $\Gamma=O(k/\varepsilon)$, where $k$'s value depends on different measurement settings.

To evaluate the robustness of our algorithms, we computed the deviation among 30 experiments by the standard error of them: $std\_err = \sqrt{\frac{\sum_{i=1}^{30} \abs{x_i - avg}^2}{30}}$, where $x_i$ is the estimated $\cost(G)$ value returned from the algorithm.

\subsection{Algorithmic Performance: Accuracy and Efficiency Overview}

\begin{figure*}[!ht]
  \centering

  \begin{minipage}[t]{0.32\textwidth}
    \centering
    \includegraphics[width=\textwidth]{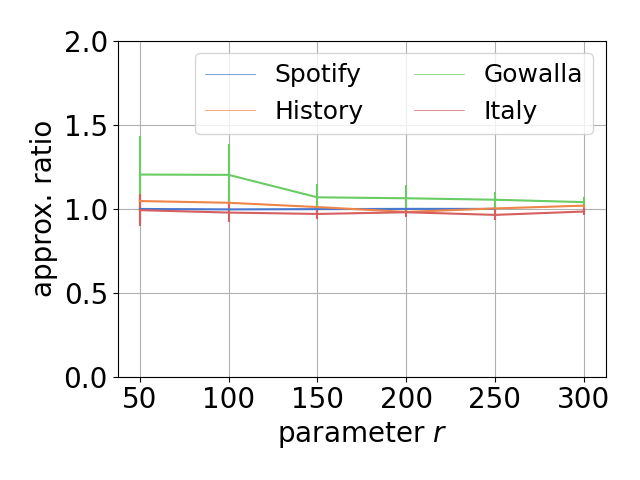}
    \par\smallskip
    {\small\textbf{(a)} Approx.\ ratio for distance (Gowalla and Italy) and similarity (Spotify and History) datasets}
  \end{minipage}\hfill
  \begin{minipage}[t]{0.32\textwidth}
    \centering
    \includegraphics[width=\textwidth]{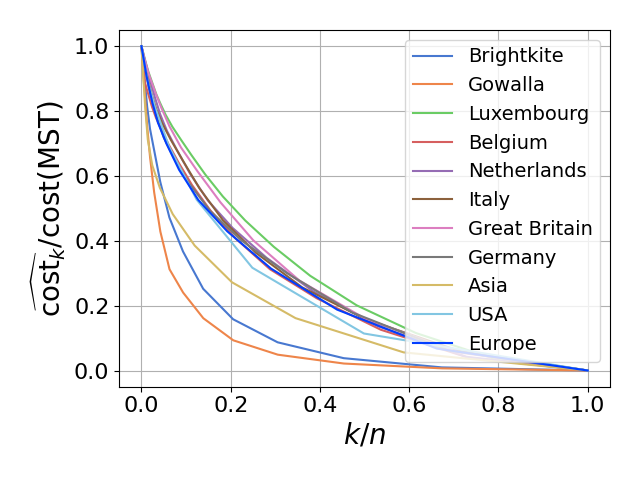}
    \par\smallskip
    {\small\textbf{(b)} Profiles for distance datasets}
  \end{minipage}\hfill
  \begin{minipage}[t]{0.32\textwidth}
    \centering
    \includegraphics[width=\textwidth]{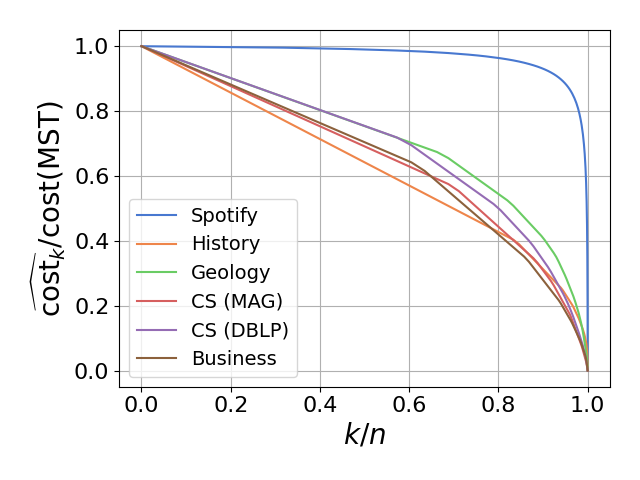}
    \par\smallskip
    {\small\textbf{(c)} Profile for similarity datasets}
    \label{fig:normalized-profile-sim}
  \end{minipage}

  \caption{Approximation ratio and normalized profiles.}
  \label{fig:performance}
\end{figure*}

\begin{figure*}[!htbp]
  \centering

  \begin{minipage}[b]{0.32\textwidth}
    \centering
    \includegraphics[width=\textwidth]{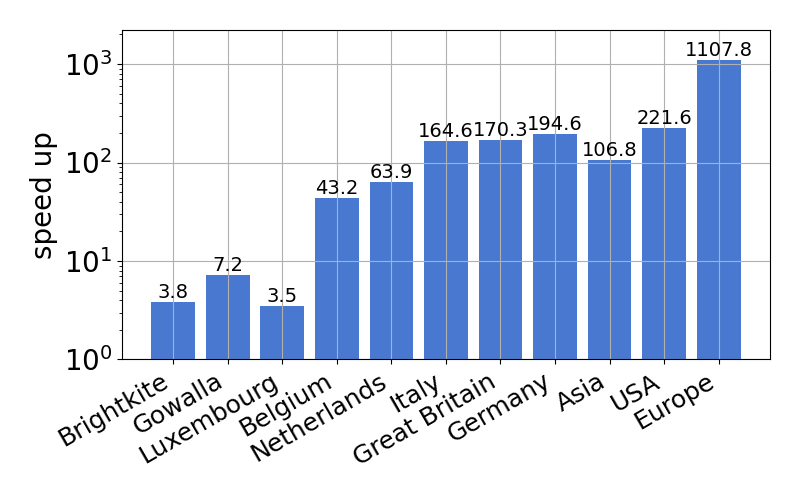}
    \par\smallskip
    {\small\textbf{(a)} Distance with $r=100$}
  \end{minipage}\hfill
  \begin{minipage}[b]{0.32\textwidth}
    \centering
    \includegraphics[width=\textwidth]{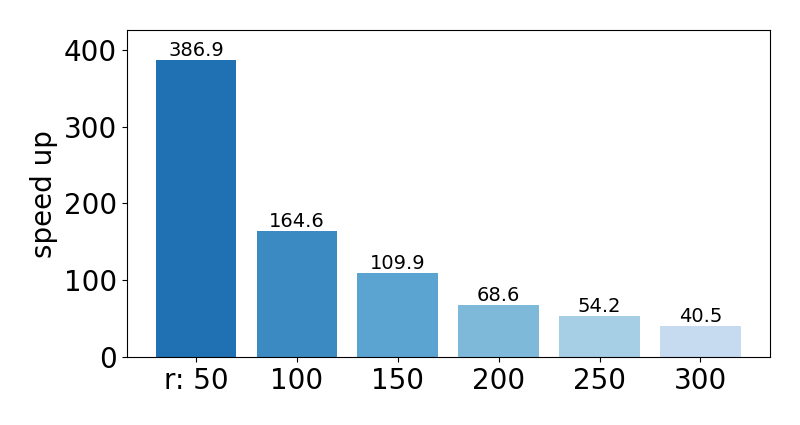}
    \par\smallskip
    {\small\textbf{(b)} Speed up on `Italy' dataset, with various $r$}
  \end{minipage}\hfill
  \begin{minipage}[b]{0.32\textwidth}
    \centering
    \includegraphics[width=\textwidth]{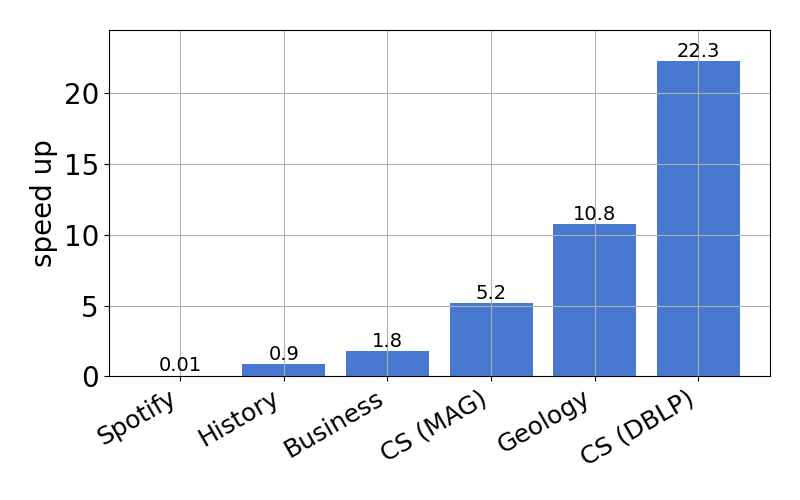}
    \par\smallskip
    {\small\textbf{(c)} Similarity with $r=100$}
  \end{minipage}

  \caption{Datasets speed up.}
  \label{fig:speedup}
\end{figure*}

\cref{fig:performance}\textbf{(a)} shows the accuracy of our algorithm to approximate clustering cost $\cost(G)$ and $\cost^{(s)}(G)$, among both distance and similarity datasets. Our algorithm already has a very good approximation ratio, even when $r$ is small; and accuracy becomes better as $r$ increases, as we have finer intervals and larger sample size.
For the sake of space, we put more experiments in appendix.
For the bias where the approximation ratio begins at a value higher than $1.0$
when $r$ is small, it is likely from the algorithm
overlooking some large connected components, when the number of sampled vertices is small.
However, our results verified that the bias diminishes as $r$ increases.

For most of the datasets, our running time beats the baseline, Kruskal's MST algorithm,
shown in \cref{fig:speedup}.
We remark that when the dataset contains a relatively small number of vertices and edges, then it is best to simply run Kruskal's MST algorithm to compute
$\cost(G)$. Our speedup on the datasets Brightkite, Gowalla, and Luxembourg are lower,
because they are small datasets.
As Spotify has a relatively large value of $W$, and
since our running time is linear in $W$, it is more efficient to simply run an MST algorithm.

\begin{table}
\caption{The accumulated profile error compared to total clustering cost,
${\sum_{k=1}^{n}|\widehat{\cost}_k-\cost_k|}/{\cost(G)}$}
\label{tbl:sum_profile_error_over_costG}
\centering
\begin{tabular}{l|llll}
\hline
r     & Gowalla & Italy & History & Geology \\
\hline
100   & 0.578   & 0.230 & 0.054   & 0.023   \\
1000  & 0.123   & 0.067 & 0.009   & 0.005   \\
10000 & 0.033   & 0.024 & 0.003   & 0.003   \\
20000 & 0.023   & 0.017 & 0.001   & 0.002  \\
\hline
\end{tabular}
\end{table}

\cref{fig:performance}\textbf{(b)} and \cref{fig:performance}\textbf{(c)} shows the hierarchical clustering structures. 
We normalized the profile vectors
by computing the fraction of the clustering size and cost, i.e., 
$\frac{k}{\max\{1,\dots,n\}}=\frac{k}{n}$ and $\frac{\cost_k}{\max\{\cost_1,\dots,\cost_n\}}=\frac{\cost_k}{\cost(\MST)}$.
As shown in the figures, different countries have different structures, especially, Asia has a very
different structure from other countries; and road networks differ from localization-based datasets.
Besides, co-citation graphs in various areas also have different structures, and Spotify is far more different from them.
The accumulated approximation error of $\widehat{\cost}_k$ is shown in \cref{tbl:sum_profile_error_over_costG},
which proves that we can estimate the profiles efficiently and with bounded error, compared to the total clustering cost $\cost(G)$.

\subsection{Extended Results for Distance Graphs}

For road networks, \cref{fig:appratio-dist-road1}, \cref{fig:appratio-dist-road2} and \cref{fig:appratio-dist-road3} report the approximation ratio for estimating clustering cost $\cost(G)$ among these datasets;
\cref{fig:profile-dist-road1}, \cref{fig:profile-dist-road2} and \cref{fig:profile-dist-road3} compare the exact and estimated profile values for different choices of the sample size $r$.
In addition, \cref{fig:appratio-profile-dist-loc} presents the approximation ratio for estimating both $\cost(G)$ and profiles on localization datasets.

\subsection{Extended Results for Similarity Graphs}

When setting parameters, we choose smaller constants than theory suggests, since real-world datasets often admit good performance with more modest settings; in practice, one could also average results across multiple runs to reduce variance, relaxing the need for per-run accuracy guarantees.

Specifically, in \cref{alg:appcost_sim}, we estimate two types of values: we use \cref{alg:appncc} to estimate $\hat{c}_j$,
and use \cref{alg:appncc_sim} to estimate $\overline{D}_j=n-\bar{c}_j$. These two algorithms use different sample sizes $r$. For the first, we set $r$ to be the input sample size of the code.
For the second, we use a larger sample size $r'$. From theory, $r=O(\frac{1}{\varepsilon^2})$ and 
$r'=\frac{W}{\varepsilon^2}$, thus $r'=O(r\cdot W)$. In practice, we set $r'=\max\{r\cdot W/\log n,r\}$. 
Dividing $r\cdot W$ by $\log n$ balances the running time of the two algorithms, 
and there are at most $O(\log n)$ intervals overall.
For the intervals, since $r=O(\frac{1}{\varepsilon^2})$, we set $\varepsilon=\frac{1}{\sqrt{r}}$ and define interval values as in \cref{def:interval_similarity}.

\cref{fig:appratio-sim1} and \cref{fig:appratio-sim2} show the approximation ratio for estimating the clustering cost $\cost^{(s)}(G)$. Larger sample size $r$ results in better average approximation ratio and smaller deviation: as $r$ increases, $r'$ also increases, improving the estimates of each $\hat{c}_j$ and $\widehat{D}_m$; moreover, with $\varepsilon$ tied to $r$, the intervals become finer, which improves accuracy for $c_j$ and $D_j$ for $1\leq j\leq W$. 

We also compute profiles for similarity datasets, as shown in \cref{fig:profile-sim1} and \cref{fig:profile-sim2}, which demonstrate that our algorithm estimates profiles accurately.
\cref{fig:normalized-profile-sim} presents normalized profiles for similarity datasets: The Spotify curve differs markedly from the co-authorship graphs, indicating that
the profile method distinguishes different types of graphs. A higher curve suggests more
collaboration; for instance, the Spotify dataset exhibits more `collaboration' between songs. 
Note that both \textit{Computer Science} and \textit{DBLP} are co-authorship graphs in the
computer science field, but \textit{DBLP} is more complete and thus reveals more structure, including greater collaboration among authors.

\FloatBarrier
\begingroup
\setlength{\intextsep}{0pt}
\setlength{\abovecaptionskip}{2pt}
\begin{figure}[H]
  \centering
  \begin{minipage}[t]{0.32\textwidth}
    \centering
    \includegraphics[width=\textwidth]{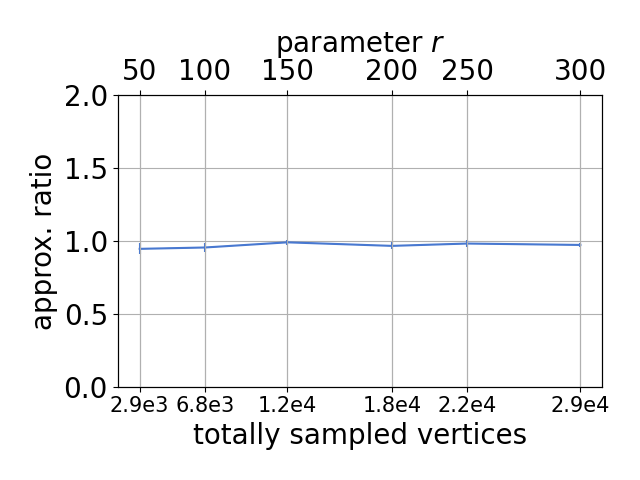}
    \par\vspace{1pt}
    {\small Luxembourg}
  \end{minipage}\hfill
  \begin{minipage}[t]{0.32\textwidth}
    \centering
    \includegraphics[width=\textwidth]{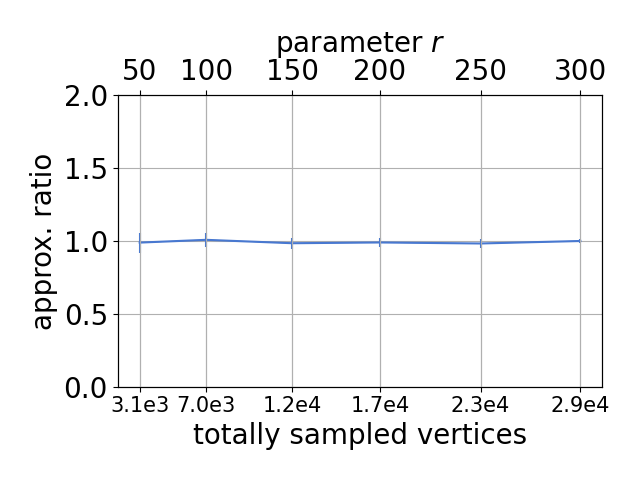}
    \par\vspace{1pt}
    {\small Belgium}
  \end{minipage}\hfill
  \begin{minipage}[t]{0.32\textwidth}
    \centering
    \includegraphics[width=\textwidth]{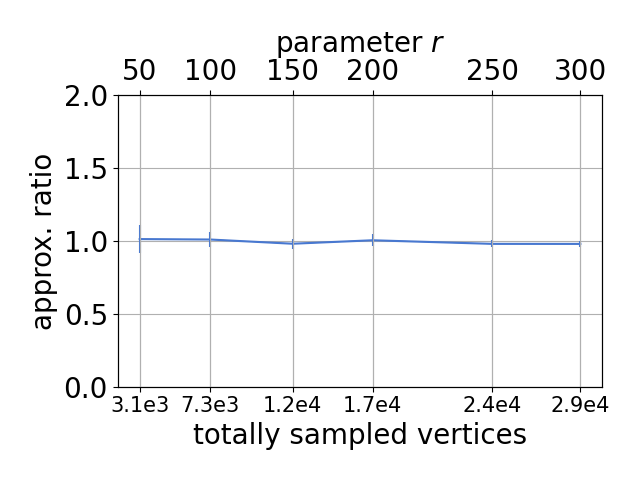}
    \par\vspace{1pt}
    {\small Netherlands}
  \end{minipage}
  \caption{Approximation ratio in distance graphs: road networks.}
  \label{fig:appratio-dist-road1}
\end{figure}

\begin{figure}[H]
  \centering
  \begin{minipage}[t]{0.32\textwidth}
    \centering
    \includegraphics[width=\textwidth]{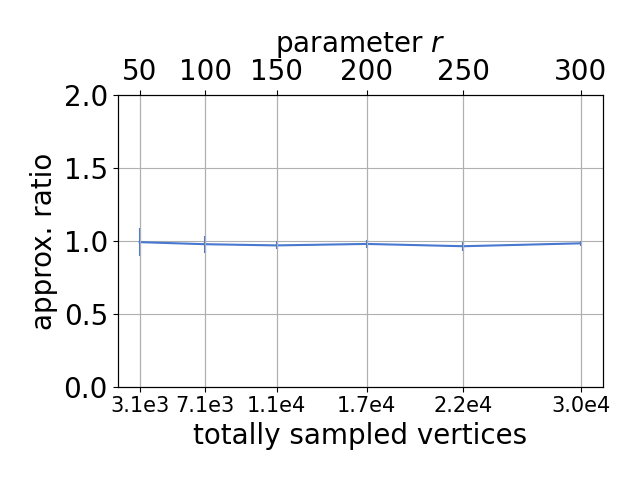}
    \par\vspace{1pt}
    {\small Italy}
  \end{minipage}\hfill
  \begin{minipage}[t]{0.32\textwidth}
    \centering
    \includegraphics[width=\textwidth]{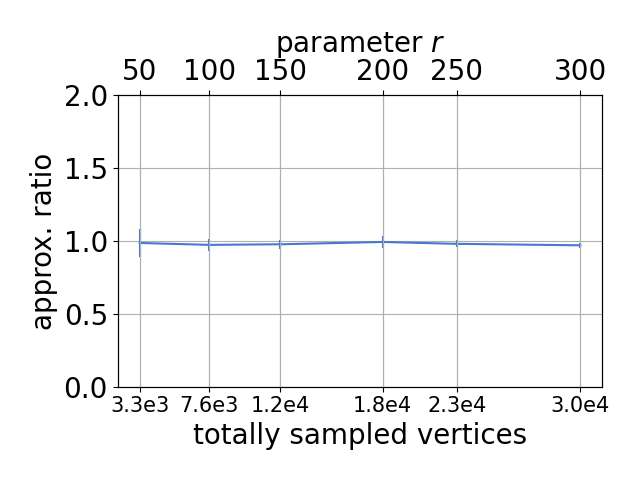}
    \par\vspace{1pt}
    {\small Great Britain}
  \end{minipage}\hfill
  \begin{minipage}[t]{0.32\textwidth}
    \centering
    \includegraphics[width=\textwidth]{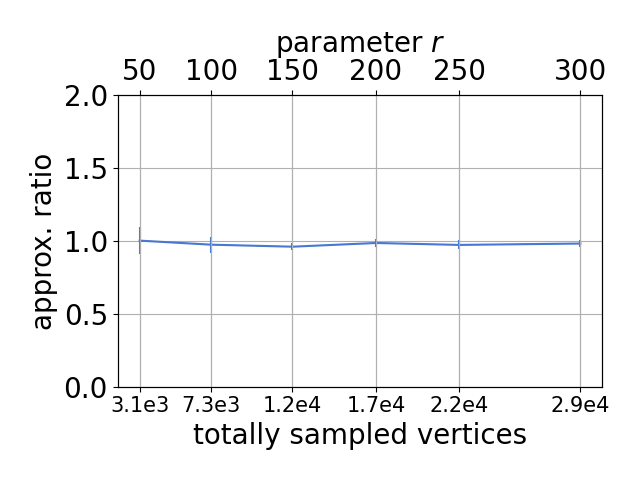}
    \par\vspace{1pt}
    {\small Germany}
  \end{minipage}
  \caption{Approximation ratio in distance graphs: road networks.}
  \label{fig:appratio-dist-road2}
\end{figure}

\begin{figure}[H]
  \centering
  \begin{minipage}[t]{0.32\textwidth}
    \centering
    \includegraphics[width=\textwidth]{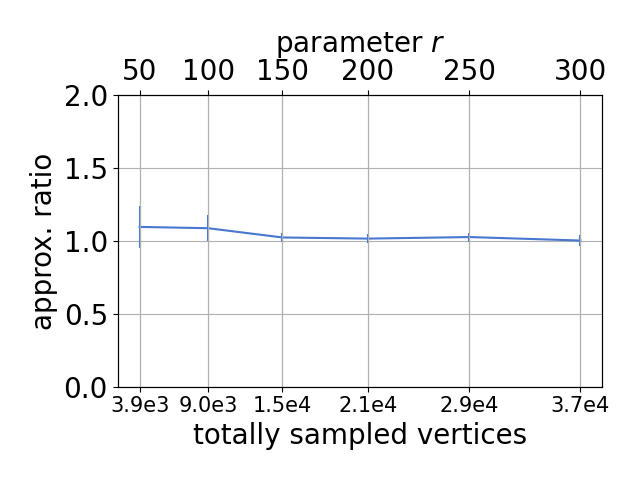}
    \par\vspace{1pt}
    {\small Asia}
  \end{minipage}\hfill
  \begin{minipage}[t]{0.32\textwidth}
    \centering
    \includegraphics[width=\textwidth]{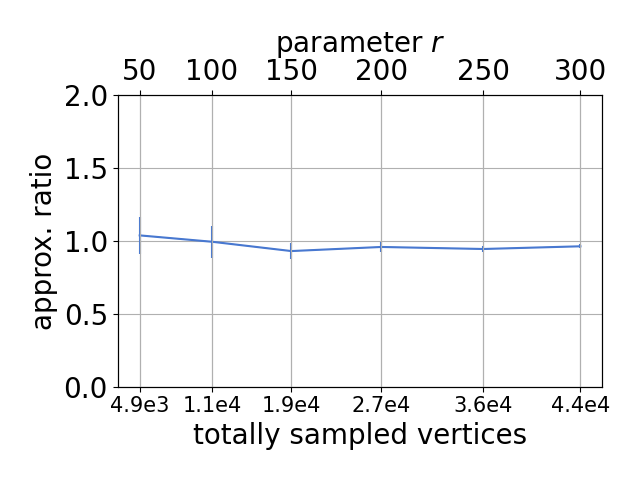}
    \par\vspace{1pt}
    {\small USA}
  \end{minipage}\hfill
  \begin{minipage}[t]{0.32\textwidth}
    \centering
    \includegraphics[width=\textwidth]{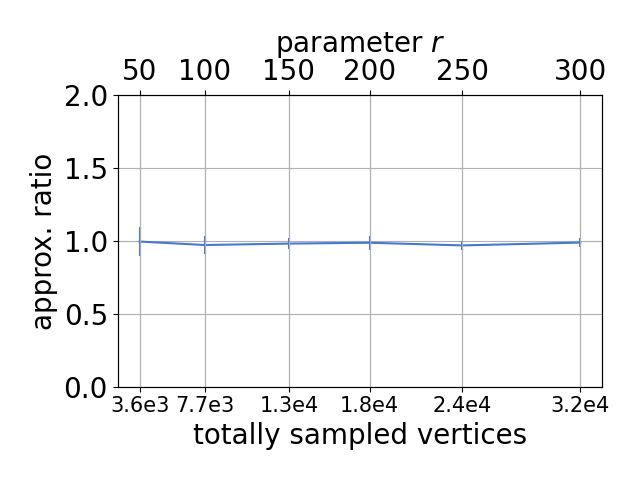}
    \par\vspace{1pt}
    {\small Europe}
  \end{minipage}
  \caption{Approximation ratio in distance graphs: road networks.}
  \label{fig:appratio-dist-road3}
\end{figure}

\begin{figure}[H]
  \centering
  \begin{minipage}[t]{0.24\textwidth}
    \centering
    \includegraphics[width=\textwidth]{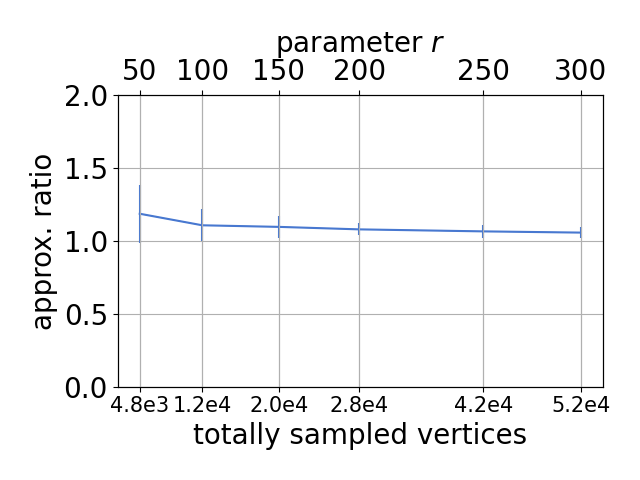}
    \par\vspace{1pt}
    {\small Accuracy of Brightkite}
  \end{minipage}\hfill
  \begin{minipage}[t]{0.24\textwidth}
    \centering
    \includegraphics[width=\textwidth]{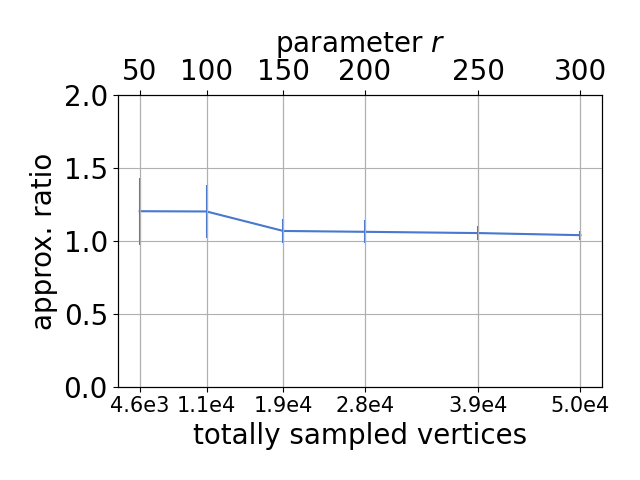}
    \par\vspace{1pt}
    {\small Accuracy of Gowalla}
  \end{minipage}\hfill
  \begin{minipage}[t]{0.24\textwidth}
    \centering
    \includegraphics[width=\textwidth]{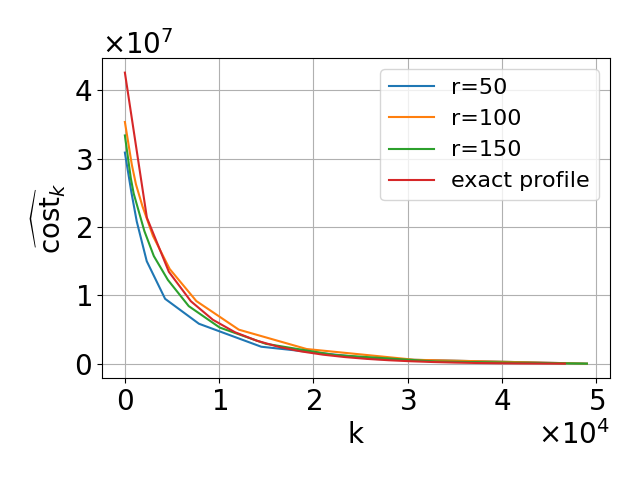}
    \par\vspace{1pt}
    {\small Profiles of Brightkite}
  \end{minipage}\hfill
  \begin{minipage}[t]{0.24\textwidth}
    \centering
    \includegraphics[width=\textwidth]{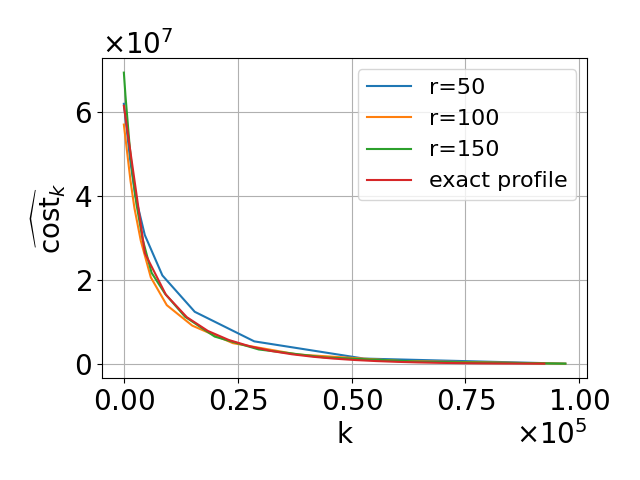}
    \par\vspace{1pt}
    {\small Profiles of Gowalla}
  \end{minipage}
  \caption{Approximation ratio and profiles for distance graphs in localization based datasets.}
  \label{fig:appratio-profile-dist-loc}
\end{figure}

\begin{figure}[H]
  \centering
  \begin{minipage}[t]{0.32\textwidth}
    \centering
    \includegraphics[width=\textwidth]{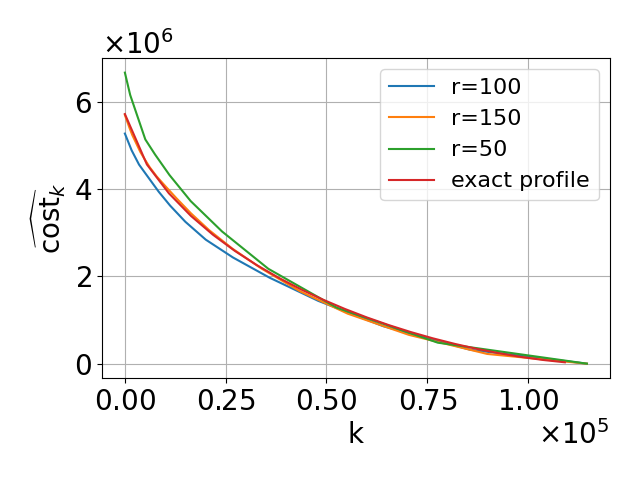}
    \par\vspace{1pt}
    {\small Luxembourg}
  \end{minipage}\hfill
  \begin{minipage}[t]{0.32\textwidth}
    \centering
    \includegraphics[width=\textwidth]{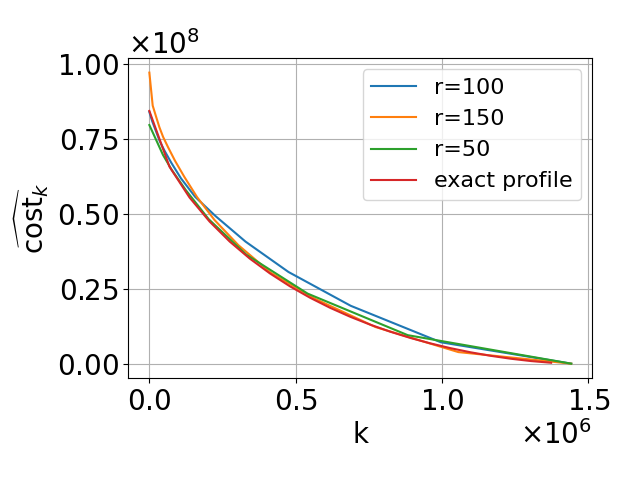}
    \par\vspace{1pt}
    {\small Belgium}
  \end{minipage}\hfill
  \begin{minipage}[t]{0.32\textwidth}
    \centering
    \includegraphics[width=\textwidth]{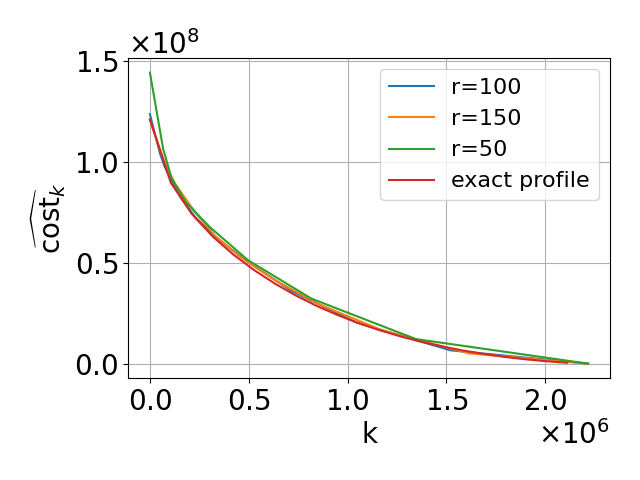}
    \par\vspace{1pt}
    {\small Netherlands}
  \end{minipage}
  \caption{Profiles for distance case in road networks.}
  \label{fig:profile-dist-road1}
\end{figure}

\begin{figure}[H]
  \centering
  \begin{minipage}[t]{0.32\textwidth}
    \centering
    \includegraphics[width=\textwidth]{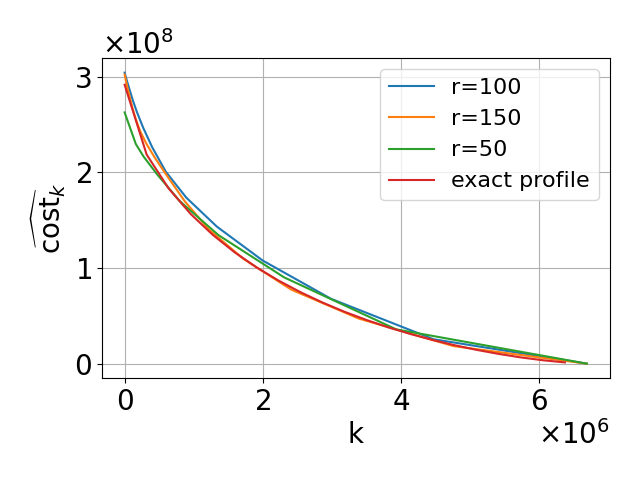}
    \par\vspace{1pt}
    {\small Italy}
  \end{minipage}\hfill
  \begin{minipage}[t]{0.32\textwidth}
    \centering
    \includegraphics[width=\textwidth]{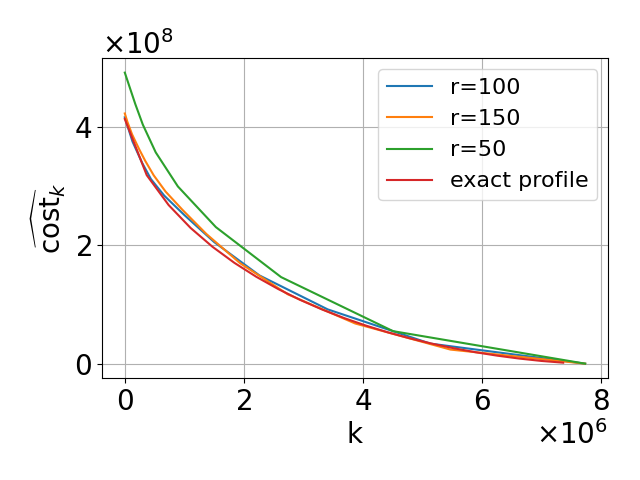}
    \par\vspace{1pt}
    {\small Great Britain}
  \end{minipage}\hfill
  \begin{minipage}[t]{0.32\textwidth}
    \centering
    \includegraphics[width=\textwidth]{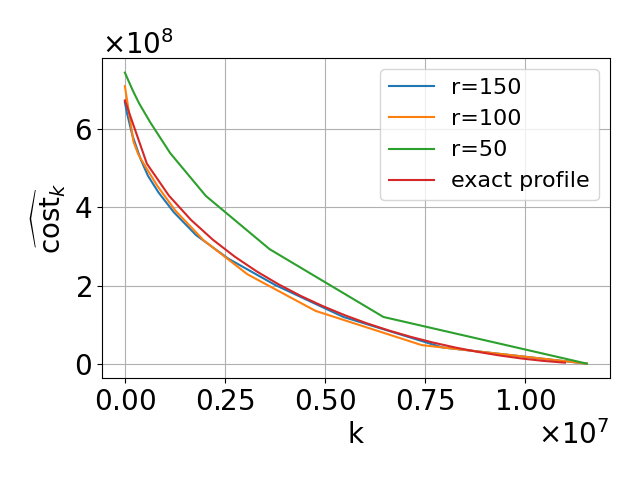}
    \par\vspace{1pt}
    {\small Germany}
  \end{minipage}
  \caption{Profiles for distance case in road networks.}
  \label{fig:profile-dist-road2}
\end{figure}

\begin{figure}[H]
  \centering
  \begin{minipage}[t]{0.32\textwidth}
    \centering
    \includegraphics[width=\textwidth]{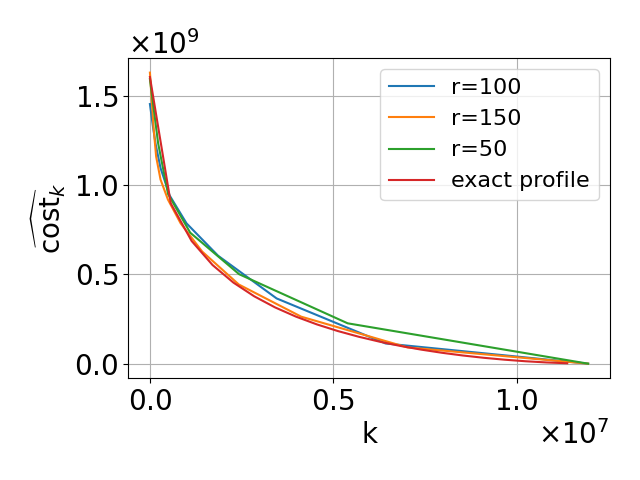}
    \par\vspace{1pt}
    {\small Asia}
  \end{minipage}\hfill
  \begin{minipage}[t]{0.32\textwidth}
    \centering
    \includegraphics[width=\textwidth]{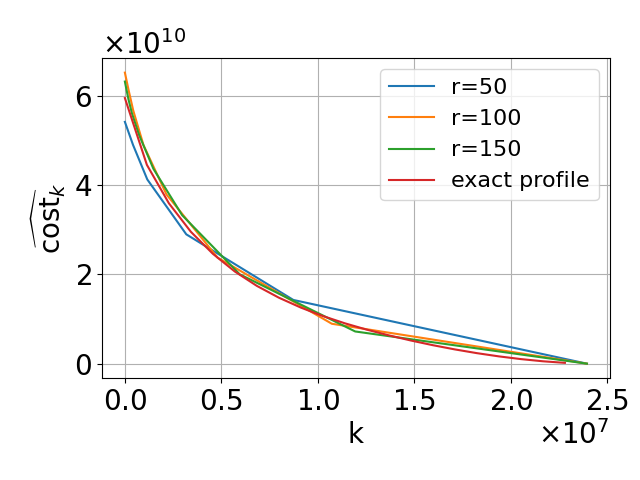}
    \par\vspace{1pt}
    {\small USA}
  \end{minipage}\hfill
  \begin{minipage}[t]{0.32\textwidth}
    \centering
    \includegraphics[width=\textwidth]{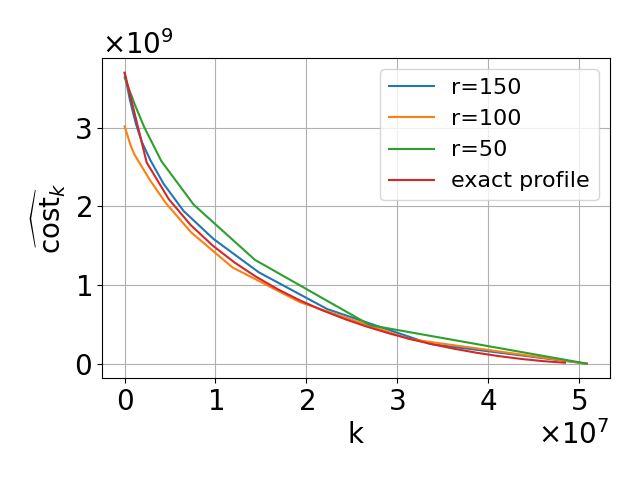}
    \par\vspace{1pt}
    {\small Europe}
  \end{minipage}
  \caption{Profiles for distance case in road networks.}
  \label{fig:profile-dist-road3}
\end{figure}

\begin{figure}[H]
  \centering
  \begin{minipage}[t]{0.32\textwidth}
    \centering
    \includegraphics[width=\textwidth]{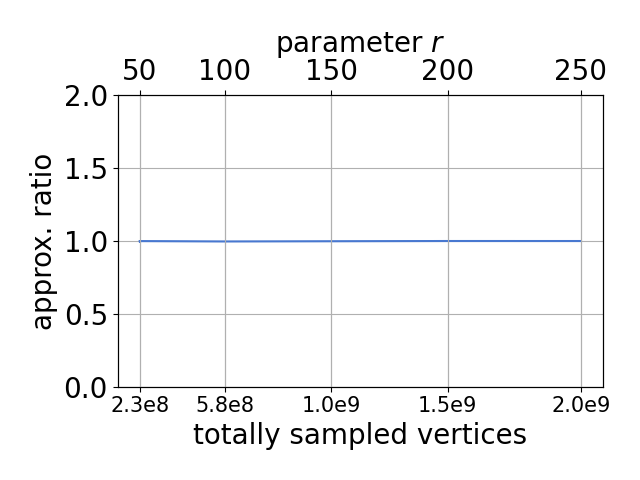}
    \par\vspace{1pt}
    {\small Spotify}
  \end{minipage}\hfill
  \begin{minipage}[t]{0.32\textwidth}
    \centering
    \includegraphics[width=\textwidth]{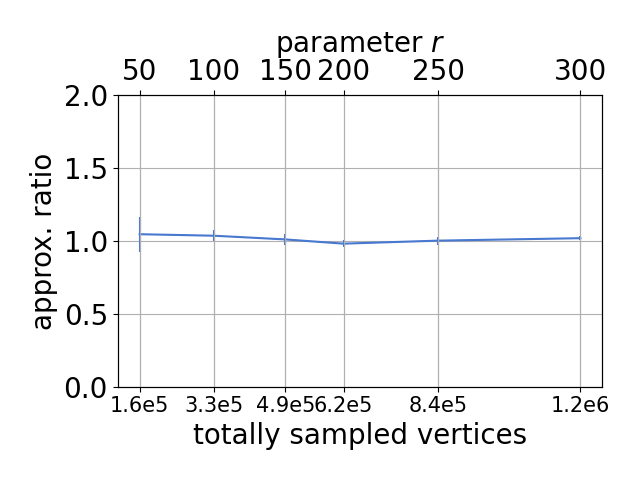}
    \par\vspace{1pt}
    {\small History}
  \end{minipage}\hfill
  \begin{minipage}[t]{0.32\textwidth}
    \centering
    \includegraphics[width=\textwidth]{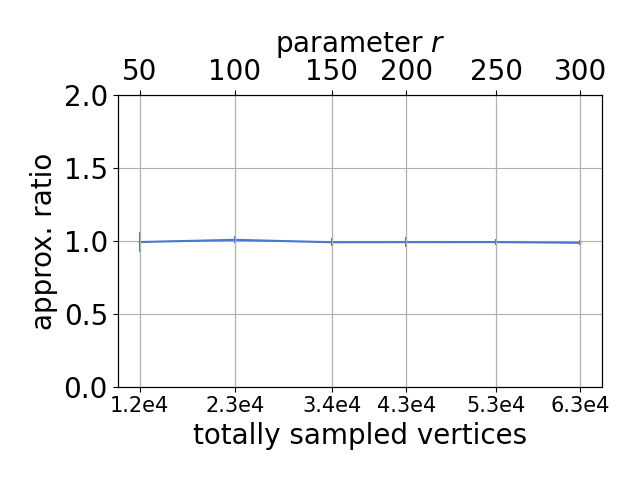}
    \par\vspace{1pt}
    {\small Business}
  \end{minipage}
  \caption{Approximation ratio in similarity graphs.}
  \label{fig:appratio-sim1}
\end{figure}

\begin{figure}[H]
  \centering
  \begin{minipage}[t]{0.32\textwidth}
    \centering
    \includegraphics[width=\textwidth]{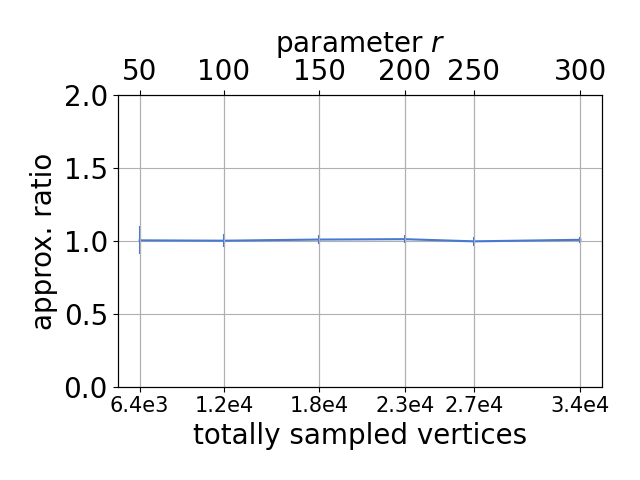}
    \par\vspace{1pt}
    {\small CS (MAG)}
  \end{minipage}\hfill
  \begin{minipage}[t]{0.32\textwidth}
    \centering
    \includegraphics[width=\textwidth]{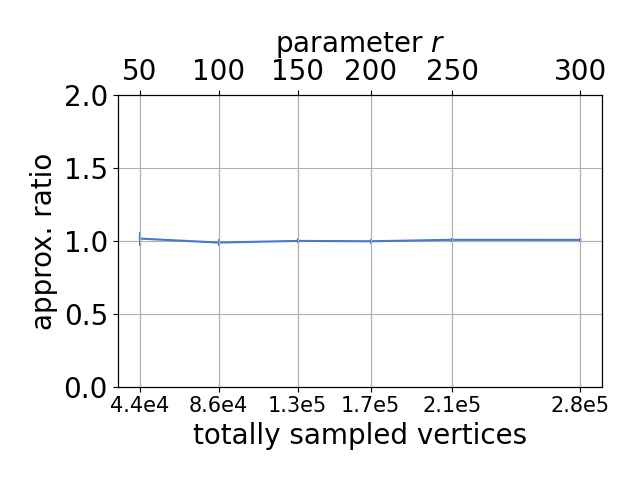}
    \par\vspace{1pt}
    {\small Geology}
  \end{minipage}\hfill
  \begin{minipage}[t]{0.32\textwidth}
    \centering
    \includegraphics[width=\textwidth]{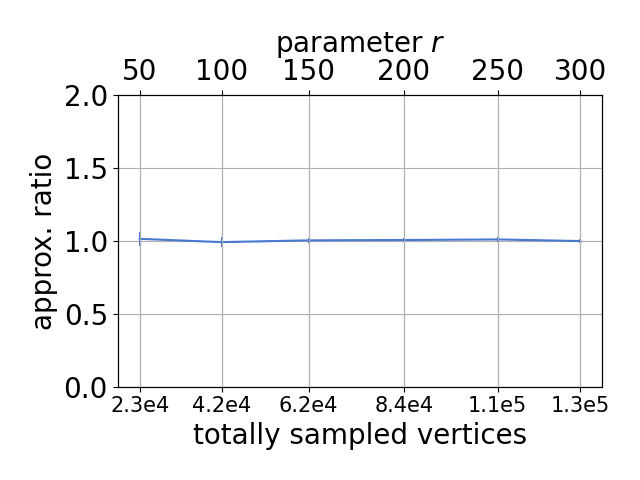}
    \par\vspace{1pt}
    {\small CS (DBLP)}
  \end{minipage}
  \caption{Approximation ratio in similarity graphs.}
  \label{fig:appratio-sim2}
\end{figure}

\begin{figure}[H]
  \centering
  \begin{minipage}[t]{0.32\textwidth}
    \centering
    \includegraphics[width=\textwidth]{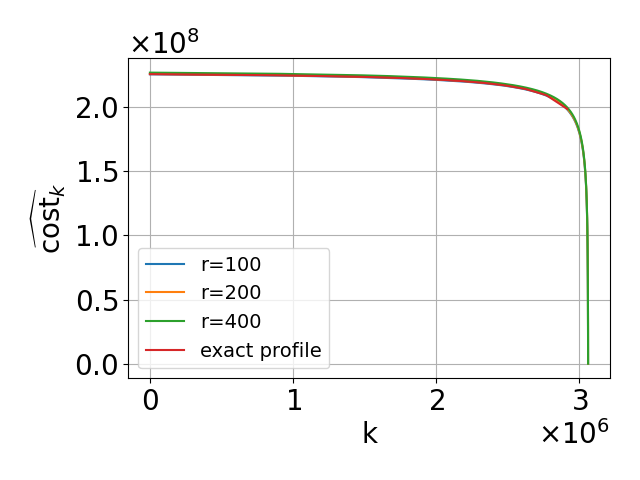}
    \par\vspace{1pt}
    {\small Spotify}
  \end{minipage}\hfill
  \begin{minipage}[t]{0.32\textwidth}
    \centering
    \includegraphics[width=\textwidth]{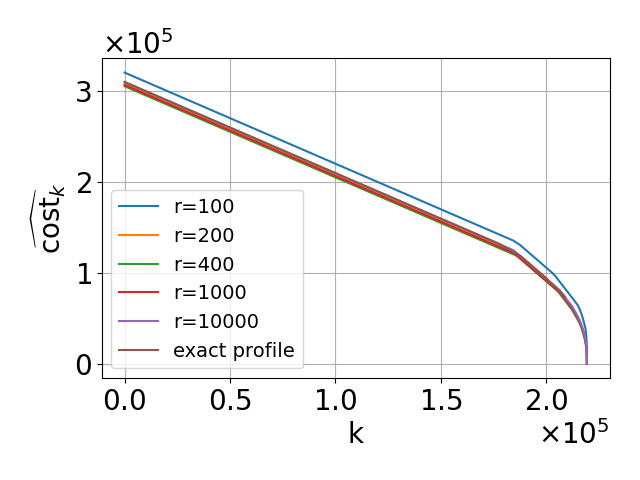}
    \par\vspace{1pt}
    {\small History}
  \end{minipage}\hfill
  \begin{minipage}[t]{0.32\textwidth}
    \centering
    \includegraphics[width=\textwidth]{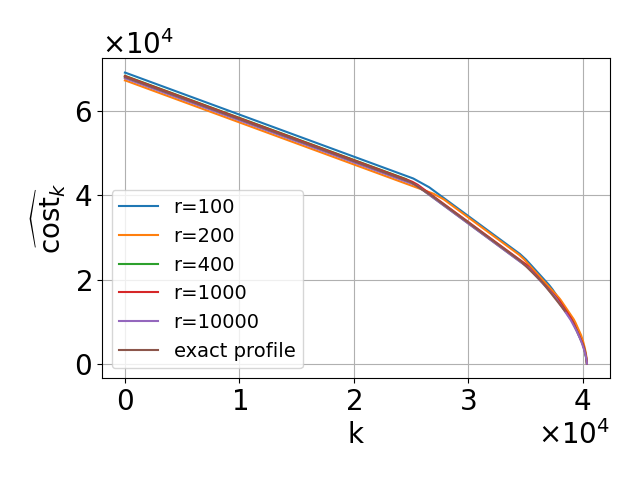}
    \par\vspace{1pt}
    {\small Business}
  \end{minipage}
  \caption{Profiles for similarity graphs.}
  \label{fig:profile-sim1}
\end{figure}

\begin{figure}[H]
  \centering
  \begin{minipage}[t]{0.32\textwidth}
    \centering
    \includegraphics[width=\textwidth]{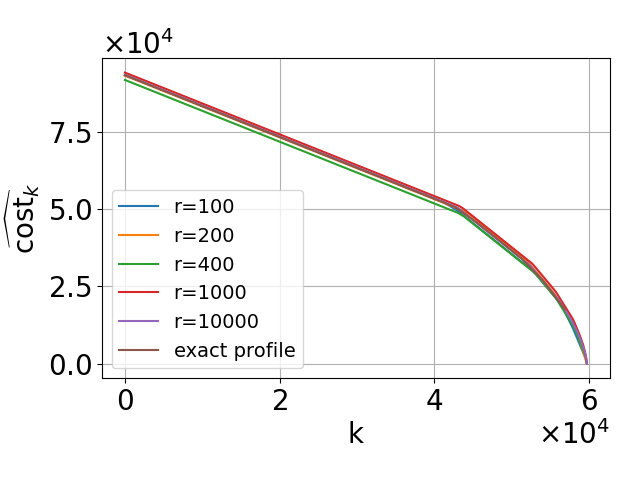}
    \par\vspace{1pt}
    {\small CS (MAG)}
  \end{minipage}\hfill
  \begin{minipage}[t]{0.32\textwidth}
    \centering
    \includegraphics[width=\textwidth]{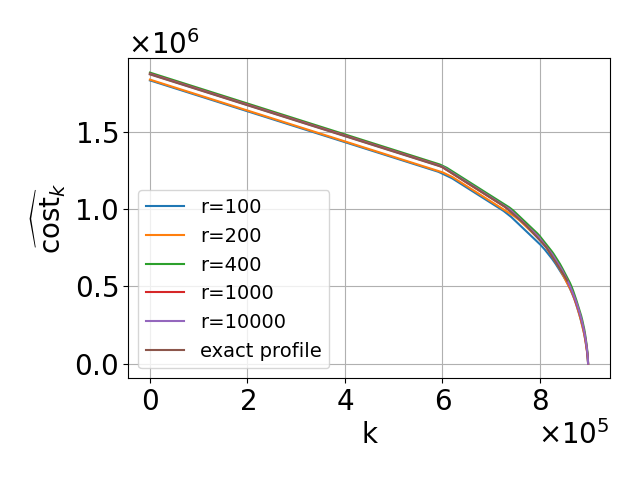}
    \par\vspace{1pt}
    {\small Geology}
  \end{minipage}\hfill
  \begin{minipage}[t]{0.32\textwidth}
    \centering
    \includegraphics[width=\textwidth]{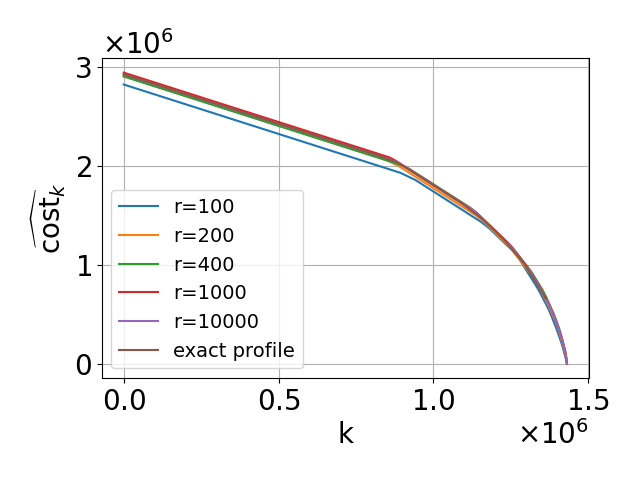}
    \par\vspace{1pt}
    {\small CS (DBLP)}
  \end{minipage}
  \caption{Profiles for similarity graphs.}
  \label{fig:profile-sim2}
\end{figure}

\endgroup
\subsection{Discussion on Bias in the Algorithm}

\cref{fig:similarity-bias}(a) shows that, when the sample size $r$ is small, the approximation ratio begins above 1. 
But when $r$ becomes very large,
\cref{fig:similarity-bias}(b) shows that this bias diminishes and the curve converges to 1.
This behavior may stem from a characteristic of the similarity datasets, as indicated in \cref{tbl:last-interval}:
in the last interval $[0,\frac{\varepsilon n}{W}]$,
more $D_i$ values lie below the interval average.
Consequently, when $r$ is small, the last interval is too sparse,
and $D_i$ values within it are consistently over-estimated, producing the observed bias.
Setting the last interval to the smaller range $[0,1]$ removes this effect, as shown in \cref{fig:similarity-bias}(c): the average ratio concentrates around 1, with deviation on both sides, meaning some experiments under-estimate the cost while others over-estimate it.

\begin{figure}[h!]
  \centering
  \begin{minipage}[t]{0.32\textwidth}
    \centering
    \includegraphics[width=\textwidth]{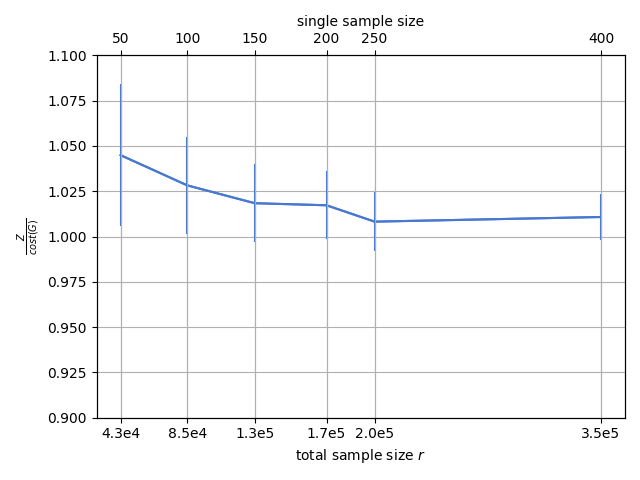}
    \par\smallskip
    {\small Zoom in on y-axis}
  \end{minipage}\hfill
  \begin{minipage}[t]{0.32\textwidth}
    \centering
    \includegraphics[width=\textwidth]{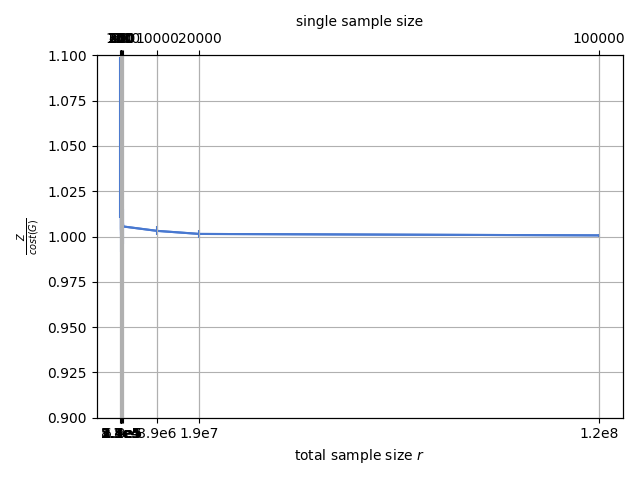}
    \par\smallskip
    {\small Extremely large $r$}
  \end{minipage}\hfill
  \begin{minipage}[t]{0.32\textwidth}
    \centering
    \includegraphics[width=\textwidth]{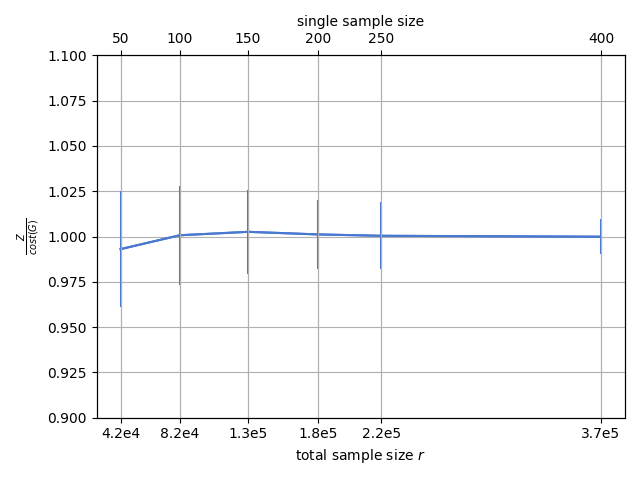}
    \par\smallskip
    {\small Set the last interval $=[0,1]$}
    \label{fig:similarity-bias-remove}
  \end{minipage}
  \caption{Bias of co-citation datasets, take Geology as an example.}
  \label{fig:similarity-bias}
\end{figure}

\begin{table}[h!]
\caption{In the last interval, $D_i$'s distribution}
\label{tbl:last-interval}
\centering
\begin{tabular}{llll}
\hline
dataset & W   & \#$c_j<\frac{\varepsilon n}{W}$ & \#$c_j<\frac{\varepsilon n}{W}/2$ \\
\hline
History & 606 & 430                     & 329                       \\
Geology & 192 & 133                     & 119                       \\
CS (MAG)     & 36  & 22                      & 19\\
\hline         
\end{tabular}
\end{table}

\section{Limitations}

Our techniques rely on the MST/threshold-connectivity characterization of single-linkage clustering.
Therefore, they do not directly extend to average linkage or complete linkage, whose merge
decisions depend on cluster-level averages or diameters rather than connected components in
threshold graphs. The noisy monotone-sequence sketch developed here may be reusable in
other settings where a hierarchy can be represented through a monotone sequence of locally
estimable graph statistics, but new estimators would be needed for average or complete linkage.

\section*{Acknowledgment}
The work was conducted in part while Pan Peng and Christian Sohler were visiting the Simons Institute for the Theory of Computing as participants in the Sublinear Algorithms program. 

\end{document}

%% file: main.bbl
\newcommand{\etalchar}[1]{$^{#1}$}